\pdfoutput=1

\newif\ifdraft\draftfalse   %
\newif\ifanon\anontrue      %
\newif\ifcamera\camerafalse %
\newif\iflongrefs\longrefsfalse %
\newif\ifsooner\soonerfalse %
\newif\iflater\laterfalse   %
\newif\iffull\fullfalse   %
\newif\ifneedspace\needspacefalse %
\newif\ifhighlightnewtext\highlightnewtextfalse %
\makeatletter \@input{texdirectives.tex} \makeatother

\documentclass[runningheads,orivec]{llncs} %
\usepackage[T1]{fontenc}
\usepackage[numbers,sort&compress]{natbib} %
\usepackage{amsmath,amssymb} %
\usepackage[dvipsnames]{xcolor}          %
\usepackage[hidelinks]{hyperref} %
\usepackage{orcidlink} %
\renewcommand{\orcidID}[1]{\,\orcidlink{#1}}
\ifcamera\else %
\hypersetup{bookmarksdepth=3,bookmarksnumbered=true,bookmarksopen=true,bookmarksopenlevel=1}
\fi
\usepackage{stmaryrd}
\usepackage{mathpartir}
\usepackage{mathtools}
\usepackage[capitalise]{cleveref}

\crefname{section}{\S}{\S\S}
\Crefname{section}{\S}{\S\S}
\crefformat{section}{\S#2#1#3}
\Crefformat{section}{\S#2#1#3}

\AtBeginDocument{%
  }

\usepackage{booktabs}   %

\usepackage[normalem]{ulem}

\definecolor{dkblue}{rgb}{0,0.1,0.5}
\definecolor{dkgreen}{rgb}{0,0.4,0}
\definecolor{dkred}{rgb}{0.6,0,0}
\definecolor{dkpurple}{rgb}{0.7,0,1.0}
\definecolor{purple}{rgb}{0.9,0,1.0}
\definecolor{olive}{rgb}{0.4, 0.4, 0.0}
\definecolor{teal}{rgb}{0.0,0.4,0.4}
\definecolor{azure}{rgb}{0.0, 0.5, 1.0}
\definecolor{gray}{rgb}{0.5, 0.5, 0.5}
\definecolor{dkgray}{rgb}{0.3, 0.3, 0.3}
\definecolor{orange}{rgb}{1.0, 0.5, 0.0}
\definecolor{lhorange}{rgb}{1.0, 0.8, 0.6}

\usepackage{xspace}

\newcommand*{\EG}{e.g.,\xspace}
\newcommand*{\IE}{i.e.,\xspace}

\newcommand{\ret}{\ensuremath{\mathsf{ret}}\xspace}
\newcommand{\bind}{\ensuremath{\mathsf{bind}}\xspace}
\newcommand{\result}[1]{\ensuremath{\mathsf{val}(#1)}}

\newcommand\fstar{F$^{\star}$\xspace}

\newcommand{\ArrowType}[3]{\mathsf{#1}(\mathsf{#2},\mathsf{#3})}
\newcommand{\MonadType}[2]{\mathsf{#1}(\mathsf{#2})}
\newcommand{\StTArrType}[4]{\StArrT(\mathsf{#1})(\mathsf{#2},\mathsf{#3},\mathsf{#4})}
\newcommand{\StTType}[3]{\StT(\mathsf{#1})(\mathsf{#2},\mathsf{#3})}
\newcommand{\Arr}{\mathsf{Arr}}
\newcommand{\SPArr}{\mathsf{SPArr}}
\newcommand{\arr}[2][]{\mathsf{arr}\IfBlankF{#1}{_{\mathsf{#1}}}\IfBlankF{#2}{\,#2}}
\newcommand{\first}[2][]{\mathsf{first}\IfBlankF{#1}{_{\mathsf{#1}}}\IfBlankF{#2}{\,#2}}
\newcommand{\arrcomp}[3][]{#2 \ggg\IfBlankF{#1}{_{\mathsf{#1}}} #3}
\newcommand{\app}{\mathsf{app}}
\newcommand{\op}{\mathrm{op}}
\newcommand{\WPure}{\mathsf{W^{Pure}}}
\newcommand{\WP}{\mathsf{WP}}
\newcommand{\SP}{\mathsf{SP}}
\newcommand{\WSt}{\mathsf{\WP^{St}}}
\newcommand{\WStOp}{\WSt^\op}%
\newcommand{\WStExc}{\mathsf{W^{StExc}}}
\newcommand{\vDWP}{\mathsf{DWP}}
\newcommand{\vDWLP}{\mathsf{DWLP}}
\newcommand{\vAWP}{\mathsf{AWP}}
\newcommand{\vAWLP}{\mathsf{AWLP}}
\newcommand{\vASP}{\mathsf{ASP}}
\newcommand{\vDSLP}{\mathsf{DSLP}}
\newcommand{\vDSP}{\mathsf{DSP}}
\newcommand{\vASLP}{\mathsf{ASLP}}
\newcommand{\DWP}{\mathsf{\WP_{\vDWP}}}
\newcommand{\DWLP}{\mathsf{\WP_{\vDWLP}}}
\newcommand{\AWP}{\mathsf{\WP_{\vAWP}}}
\newcommand{\AWLP}{\mathsf{\WP_{\vAWLP}}}
\newcommand{\DWPSt}{\mathsf{\WP^{St}_{\vDWP}}}
\newcommand{\DWLPSt}{\mathsf{\WP^{St}_{\vDWLP}}}
\newcommand{\AWPSt}{\mathsf{\WP^{St}_{\vAWP}}}
\newcommand{\AWLPSt}{\mathsf{\WP^{St}_{\vAWLP}}}
\newcommand{\SPSt}{\mathsf{\SP^{St}}}
\newcommand{\ASP}{\mathsf{\SP_{\vASP}}}
\newcommand{\DSLP}{\mathsf{\SP_{\vDSLP}}}
\newcommand{\DSP}{\mathsf{\SP_{\vDSP}}}
\newcommand{\ASLP}{\mathsf{\SP_{\vASLP}}}
\newcommand{\ASPSt}{\mathsf{\SP^{St}_{\vASP}}}
\newcommand{\DSLPSt}{\mathsf{\SP^{St}_{\vDSLP}}}
\newcommand{\DSPSt}{\mathsf{\SP^{St}_{\vDSP}}}
\newcommand{\ASLPSt}{\mathsf{\SP^{St}_{\vASLP}}}
\newcommand{\SPost}{\mathsf{SPost}}
\newcommand{\SPostAng}{\mathsf{SPost_{ang}}}
\newcommand{\SPostDem}{\mathsf{SPost_{dem}}}
\newcommand{\Pred}{\mathsf{Pred}}
\newcommand{\PredSt}{\mathsf{Pred^{St}}}
\newcommand{\PredStOp}{\PredSt^\op}
\newcommand{\PrePost}{\mathsf{PrePost}}
\newcommand{\ASPArr}{\Arr_{\vASP}}
\newcommand{\DSLPArr}{\Arr_{\vDSLP}}
\newcommand{\DSPArr}{\Arr_{\vDSP}}
\newcommand{\ASLPArr}{\Arr_{\vASLP}}

\newcommand{\StT}{\mathsf{StT}}
\newcommand{\StArrT}{\mathsf{StArrT}}

\newcommand{\Ladj}{\mathsf{L}}
\newcommand{\Radj}{\mathsf{R}}

\usepackage{listings}
\usepackage[framemethod=tikz]{mdframed}
\usetikzlibrary{arrows.meta}
\newcommand\maybecolor[1]{\color{#1}}
\newcommand\nomaybecolor[1]{}
\colorlet{codecol}{dkgreen}

\ifneedspace
\fi
\definecolor{light-gray}{gray}{0.95}
\usepackage{combelow} %
\usepackage{tikz}
\usepackage{wrapfig}
\usepackage[inline]{enumitem}
\newlist{inlist}{enumerate*}{1}
\setlist[inlist]{label=(\arabic*)}

\ExplSyntaxOn
\msg_redirect_name:nnn { kernel } { variant-same-as-base } { warning }
\ExplSyntaxOff
\usepackage{quiver}

\usetikzlibrary{shapes,arrows}
\tikzset{%
  block/.style    = {draw, thick, rectangle, minimum height = 1.5em,
    minimum width = 3em},
}

\def\Snospace~{\S{}}

\makeatletter
\g@addto@macro\normalsize{%
  \abovedisplayskip2\p@ \belowdisplayskip2\p@
  \abovedisplayshortskip\z@ \belowdisplayshortskip2\p@}
\g@addto@macro\small{%
  \abovedisplayskip2\p@ \belowdisplayskip2\p@
  \abovedisplayshortskip\z@ \belowdisplayshortskip2\p@}
\makeatother

\begin{document}

\title{A Taxonomy of Predicate Transformer Monads for Specifying Monadic Programs}
\titlerunning{A Taxonomy of Predicate Transformer Monads}

\ifanon
\author{}
\institute{}
\else
\author{Abigail Pribisova\inst{1}\orcidID{0000-0002-7462-9384} \and
  Yao Li\inst{2}\orcidID{0000-0001-8720-883X} \and
  Cezar-Constantin Andrici\inst{1}\orcidID{0009-0002-7525-2440} \and \\
  C\u{a}t\u{a}lin Hri\cb{t}cu\inst{1}\orcidID{0000-0001-8919-8081} \and
  Kenji Maillard\inst{3}\orcidID{0000-0001-5554-3203} \and
  Exequiel Rivas\inst{4}\orcidID{0000-0002-2114-624X}}
\institute{MPI-SP, Bochum, Germany\\
  \email{\{abigail.pribisova,cezar.andrici,catalin.hritcu\}@mpi-sp.org} \and
  Portland State University\ifcamera, Portland, OR\fi, USA \quad
  \email{liyao@pdx.edu} \and
  Inria, Nantes, France \quad
  \email{kenji.maillard@inria.fr} \and
  Tallinn University of Technology\ifcamera, Tallinn\fi, Estonia \quad
  \email{exequiel.rivas@ttu.ee}}
\fi

\ifanon\else
\authorrunning{Pribisova et al.}
\fi

\maketitle

\begin{abstract}
Monads of weakest precondition predicate transformers
are the foundation of several verification frameworks for monadic programs.
For specifying monadic computations with nondeterminism and nontermination,
in addition to weakest preconditions, we introduce the first strongest
postcondition monads covering all combinations of angelic or demonic
nondeterminism and partial or total interpretations of nontermination.
For extending these base strongest postcondition monads with side effects
such as mutable state, monad transformers do not yield the expected monads,
so we instead generalize the base monads to arrows and use arrow transformers.
We give a taxonomy of weakest precondition and strongest postcondition monads,
for the base monads and for their versions with state, and show
that the known relationships
between under- and over-approximate triples based on predicate
transformers already hold at the level of the predicate transformers themselves.
We use the predicate transformer monads to define under- and
over-approximate triples and, where possible, Dijkstra monads.
Finally, we relate these two specification styles.

\ifanon\else
\fi
\end{abstract}

\section{Introduction}
Program logics provide a formal way to \emph{specify} what programs can or
should do, and to \emph{verify} that programs meet their specifications. For
programs with nondeterminism, a logic usually either takes a demonic
interpretation, requiring all executions to satisfy the specification, or an
angelic one requiring the existence of such an execution. For divergence, a
total interpretation requires termination, while a partial/liberal one allows
divergence. These logics are often presented as deductive systems over {\em
  triples} relating a precondition $p$, a program $c$, and a postcondition $q$.
The semantics of triples can often be phrased via Dijkstra's {\em weakest
  precondition (wp)} or {\em strongest postcondition (sp)} predicate
transformers~\cite{Dijkstra75,Dijkstra97,DijkstraS90}.
For example, the triple $\{p\}c\{q\}$ is valid in partial correctness Hoare
logic~\cite{Hoare69,Apt84} exactly when
$p \Rightarrow \mathit{dwlp} \llbracket c \rrbracket (q)$, where $\mathit{dwlp}$
is the demonic, liberal $\mathit{wp}$ variant~\cite{Dijkstra97}. In fact, this
logic admits a second such characterization~\cite{VerschtK25}: the triple is
also valid exactly when $\mathit{asp} \llbracket c \rrbracket(p) \Rightarrow q$,
for $\mathit{asp}$, the angelic, total {\em sp} variant.

These predicate transformer characterizations of triples are not limited to
reasoning about correctness. O'Hearn's incorrectness logic~\cite{OHearn20},
which specifies which states a program \emph{can} reach (\EG error states), uses
the same transformer with the implication reversed; the triple $[p]c[q]$ is
valid exactly when $q \Rightarrow \mathit{asp}\llbracket c \rrbracket(p)$.
Unlike Hoare logic, there is no equivalent {\em wp}
characterization~\cite{VerschtK25}. Incorrectness logic triples are called {\em
  under-approximate}, as $q$ under-approximates the reachable states
$\mathit{asp} \llbracket c \rrbracket(p)$. In contrast, Hoare logic triples are {\em
  over-approximate}, as $q$ over-approximates the reachable states.
For simple imperative programs, \citet{VerschtK25} taxonomize these
under- and over-approximate triples for four {\em wp} variants and four {\em
  sp} variants~\citep{ZhangK22,VerschtK25}, corresponding to the standard
interpretations of nondeterminism and nontermination: angelic-demonic and
total-liberal.\footnote{For {\em sps},
  angelic and demonic refer to whether some or all of the initial states from
  which a final state is reachable must satisfy the precondition. This matters
  when several initial states reach the same final state, even in deterministic
  programs~\citep{VerschtK25}. Similarly, the
  postcondition of a total {\em sp} only contains final states reached from the
  precondition, while that of a liberal {\em sp} also contains final states that
  nothing reaches.}

Triples and {\em wps} have also been used to verify {\em monadic}
programs~\cite{nmb08htt, ynot-icfp08, fstar-pldi13, mumon, dm4all, SwierstraB19,
  SilverZ21, pulsecore,loom}, which can be shallowly embedded in proof
assistants, making verification convenient.
A key insight was that the specifications of monadic programs also have monadic structure,
so they compose with $\ret$ and $\bind$ as the programs do.
\Citet{fstar-pldi13,mumon} showed that {\em wps} form such a monad.
But how do we define {\em wps} modularly, one effect at a time? \Citet{dm4all}
observe that if we start with a base monotone {\em wp} monad
$\MonadType{\WPure}{A} \coloneqq (\mathsf{A} \to \mathbb{P}) \to \mathbb{P}$,
where $\mathbb{P}$ is the type of propositions, we can add effects with the same
monad transformers as for computation monads.
For instance, the state monad
transformer turns $\WPure$ into the usual monad of stateful {\em wps}, of type
$\StTType{\WPure}{S}{A} \coloneq \mathsf{S} \to (\mathsf{A} \times \mathsf{S} \to \mathbb{P}) \to \mathbb{P} \cong
(\mathsf{A} \times \mathsf{S} \to \mathbb{P}) \to (\mathsf{S} \to \mathbb{P})$.
Such {\em wp} monads have been used in verification frameworks including 
\fstar{}~\cite{fstar-pldi13, mumon, dm4free, dm4all} and Loom~\cite{loom}.

Monads of {\em sps} are less understood.
For a start, monotone {\em sps} alone do \emph{not} form a lawful monad. \Citeauthor{dm4all} obtain a base $\mathit{sp}$ monad of type $\MonadType{\SPost}{A}
\coloneqq \mathbb{P} \to (\mathsf{A} \to
\mathbb{P})$ by further restricting the carrier
to predicate transformers whose postconditions imply the precondition.
We observe that the resulting monad can specify not only pure programs, but also
nondeterministic and nonterminating ones, under an angelic and total
interpretation.
In this paper we introduce the first {\em sp} monads for all four
interpretations of nondeterminism and nontermination by instead restricting the carrier
to standard classes of predicate transformers~\cite{BackW98}: the angelic
monad variants are {\em disjunctive} (\IE preserve potentially infinite non-empty
disjunctions), demonic ones {\em conjunctive} (\IE preserve such conjunctions),
total ones {\em strict} (\IE $\bot$ precondition results in $\bot$ postcondition),
and liberal ones {\em terminating} (\IE $\top$ precondition results in $\top$ postcondition).
Applying similar restrictions to the monotone {\em wp} monad
gives one {\em wp} monad per interpretation.
While these four {\em wp} monads are submonads of the $\WPure$ monad above,
the four {\em sp} monads are more interesting, since each of them needs
its own $\ret$ and $\bind$ tailored to its carrier restriction.

When also considering effects such as state, a further challenge is
{\em effect modularity}: adding the ability to reason about new effects by
applying a monad transformer, which yields a lawful monad for free. The state
monad transformer, however, does not turn these base $\mathit{sp}$ monads into
the monads we expect for stateful {\em sps}, producing instead:
$\StTType{\SPost}{S}{A} = \mathsf{S} \to \mathbb{P} \to (\mathsf{A} \times \mathsf{S} \to \mathbb{P})$.
Such an $\mathit{sp}$ takes one initial state and, separately, one precondition proposition,
whereas a stateful {\em sp} should take a
precondition that is a predicate over initial states, as in the type
$(\mathsf{S} \to \mathbb{P}) \to (\mathsf{A} \times \mathsf{S} \to \mathbb{P})$.
Underlying this is a more basic challenge: even setting modularity aside,
there was no known way to
construct lawful stateful {\em sp} monads of the expected shape, which need a
restricted carrier and a matching \ret and \bind, separately for each of the
four interpretations.\iflater{}
\fi

We solve both challenges by generalizing our
{\em sp} monads to Hughes arrows with
application~\cite{Hughes00}, since arrow transformers then yield stateful {\em sp}
monads of the expected shape in an effect-modular way.
We start with a base arrow
of type $\ArrowType{\SPArr}{B}{A} \coloneqq (\mathsf{B} \to \mathbb{P}) \to
(\mathsf{A} \to \mathbb{P})$ satisfying a variant of the carrier restrictions above.
By fixing the input type $\mathsf{B}$ to $\mathsf{unit}$,
an arrow with application yields a monad~\cite{LindleyWY11}:
$$\MonadType{\SP}{A} \coloneqq \ArrowType{\SPArr}{unit}{A}
= (\mathsf{unit} \to \mathbb{P}) \to (\mathsf{A} \to \mathbb{P}) \cong \mathbb{P} \to (\mathsf{A} \to \mathbb{P})$$
Moreover, since arrows have an input type, which arrow transformers can extend,
applying the state arrow transformer gives us
$\StTArrType{\SP}{S}{B}{A} = (\mathsf{B} \times \mathsf{S} \to \mathbb{P}) \to (\mathsf{A} \times \mathsf{S} \to \mathbb{P})$,
which, when the input type is fixed, provides the expected shape of a stateful
{\em sp} monad.
Crucially, the resulting stateful monads inherit the carrier restriction and
lawfulness of the base arrow, since the state arrow transformer only adds the
state to the input and output types.

This allows us to give a taxonomy of these base and stateful {\em sp} and {\em
  wp} monads (\autoref{fig:monad-taxonomy}, in \autoref{sec:taxonomy}) and show
that the relationships established by \citet{VerschtK25} between under- and
over-approximate triples
already hold for the
predicate transformers themselves.
Additionally, we use the predicate transformer monads to define under- and
over-approximate triples and, where possible, Dijkstra monads~\cite{fstar-pldi13,
dm4free,dm4all}. A {\em Dijkstra monad} is a monad-like structure of computations
indexed by elements of a specification monad~\cite{dm4all}.
A {\em specification monad} is ordered by specification precision so that $\bind$
is monotone~\cite{KatsumataS13,HasuoGWP15}.

We make the following \textbf{contributions}:
\begin{itemize}[leftmargin=13pt,nosep,itemsep=2pt,label=$\blacktriangleright$]
\item We introduce the first lawful {\em sp} monads covering all four
interpretations of nondeterminism and nontermination.
They are derived from arrows with application
        restricted to preserve non-empty disjunctions (angelic) or conjunctions
        (demonic),
        and $\bot$ (total) or
        $\top$ (liberal).
This allows adding state using arrow transformers,
whereas monad transformers do not yield the expected monads.
We also obtain four base {\em wp} monads by applying similar restrictions directly to
the monotone {\em wp} monad~\cite{dm4all}.
Unlike the {\em wp} monads, which are all submonads of the same monad,
each {\em sp} monad has its own $\ret$ and $\bind$.

\item We recast \citeauthor{VerschtK25}'s taxonomy of triples~\citep{VerschtK25} as a taxonomy of the eight
{\em wp} and {\em sp} monads, and show that the relationships
they establish between under- and over-approximate triples already
hold for predicate transformers, both for the base monads and for their
versions with state.

\item We construct triples (precondition, monadic computation, postcondition) in terms of
monad morphisms from our base computation monad of nondeterminism and
nontermination into our eight {\em wp} and {\em sp} monads.
We also add state and exceptions to both the computation monad and the {\em wp} and {\em sp} monads
and lift the monad morphisms for free, through monad transformers as in
DM4All~\cite{dm4free,dm4all} and through arrow transformers for the {\em sps}.
This, for instance, allows us to define the semantics of incorrectness logic
triples $[p]c[q]$ for monadic programs $c$ as $q \Rightarrow \mathit{asp}\llbracket c \rrbracket (p)$,
analogous to O'Hearn's original definition above~\cite{OHearn20}.
We use the monad morphisms to prove that simple programs satisfy triples for
both correctness and incorrectness in a style inspired by Loom~\cite{loom}, though
focused not on automation but on generality, since we consider both over- and
under-approximation, {\em sps}, and general {\em wps}.

\item We show that our new $\mathit{asp}$ and $\mathit{dslp}$ monads
and all $\mathit{wp}$ monads form specification monads~\cite{dm4all}.
This lets us verify monadic programs by comparing the specification produced by
one of our monad morphisms with the intended specification using the
specification monad's order.
This covers both over- and under-approximation by flipping the order,
and as for triples, we use this to verify simple monadic programs.
As these predicate transformer monads are specification monads, our monad
morphisms targeting them give rise to Dijkstra monads~\cite{dm4all}.
The $\mathit{dsp}$ and $\mathit{aslp}$ monads have $\bind$s that are not
monotone with respect to their natural orders, so they are not specification monads,
which limits how specifications can be weakened and composed.

\item For intended specifications expressed as triples, we show that
the Loom-inspired verification style and the Dijkstra monad-inspired one are
conceptually equivalent for the specification monads based on {\em wps} as well
as over-approximate $\mathit{asp}$ and under-approximate $\mathit{dslp}$.
This requires turning triples into predicate transformers, which we show is
possible for the specification monads above, but not for {\em under}-approximate
$\mathit{asp}$ and {\em over}-approximate $\mathit{dslp}$.

\item
To study the expressive power of our stateful {\em sp} and {\em wp} specification monads,
we compare them to stateful versions of existing specification monads~\cite{dm4all},
such as predicates on results $\MonadType{\Pred}{A} = \mathsf{A} \to \mathbb{P}$
and pairs of pre- and postconditions,
by establishing Galois connections, some of which are isomorphisms.
In particular, we show that the known isomorphism~\cite{Maillard19}
between the predicate monad $\MonadType{\Pred}{A}$ and
the base {\em sp} monad of \citeauthor{dm4all}\ %
extends to isomorphisms between $\mathit{asp}$ and $\mathit{dwlp}$ and
between $\mathit{dslp}$ and $\mathit{awp}$ as specification
monads in the stateful setting.
\item We fully mechanize all definitions and proofs in Rocq.
\end{itemize}

\paragraph{Artifact and use of AI.}
The Rocq development is \ifanon included as anonymous supplementary material.
\else available at \url{https://github.com/lastland/SpecMonad/tree/main/Code}.\fi{} It was developed in part interactively with
the help of AI coding agents (Claude, Codex, and OpenCode), especially for
proofs, which are machine-checked by Rocq. Any AI-generated code that affects
the trusted computing base (definitions and theorem statements) was carefully
reviewed by the authors. AI agents were used to help with \LaTeX{} programming
and to improve some of the text of the paper; we wrote the
original text, checked all proposed improvements, and take full responsibility
for the paper and artifact.

\section{Background: Weakest Precondition Monads}\label{sec:background}

Monadic computations can be specified in terms of weakest precondition
predicate transformers~\cite{fstar-pldi13,mumon}. A $\mathit{wp}$ maps
predicates on the {\em output} called postconditions to predicates on
the {\em input} called preconditions: if the precondition holds,
then the computation is guaranteed to produce a result satisfying the
postcondition. A pure computation of type $\mathsf{A}$ can be
understood as having an output of type $\mathsf{A}$ and an input of
type $\mathsf{unit}$. Therefore, a $\mathit{wp}$ for a pure
computation of type $\mathsf{A}$ is a function from predicates on
$\mathsf{A}$ to predicates on $\mathsf{unit}$, specifically the
following type: %
$(\mathsf{A} \to \mathbb{P}) \to (\mathsf{unit} \to \mathbb{P}) \cong
(\mathsf{A} \to \mathbb{P}) \to \mathbb{P}$.
Prior work~\cite{Jensen78,dm4free} noted that this is the type of a continuation 
monad on the return type $\mathbb{P}$ with the corresponding $\ret$ and $\bind$:
\begin{equation}
  \label{eqn:monadic-operations-wp}
\begin{aligned}
\ret\,v
  &\coloneq \lambda~\mathit{post}.\,\mathit{post}~v,\\
\bind\,m\,f
  &\coloneq \lambda~\mathit{post}.\,
     m~(\lambda a.\,f~a~\mathit{post}).
\end{aligned}
\end{equation}
We call this base $\mathit{wp}$ monad $\WP$.
This monadic structure of $\mathit{wps}$ has been exploited to
naturally compose specifications, and to map monadic
computations to their predicate transformer semantics~\cite{dm4all}.

Predicates over any type are ordered as subsets, or equivalently by pointwise logical
implication, which we write $\sqsubseteq$: 
for $\varphi, \psi : \mathsf{A} \to \mathbb{P}$, $\varphi \sqsubseteq \psi$ iff $\forall (a : \mathsf{A}).~\varphi(a) \Rightarrow \psi(a)$.
In general, one works with \emph{monotone} predicate transformers that
preserve the order relations on
predicates~\cite{DijkstraS90,BackW98}. The monotonicity condition on
{\em wp} is, for all postconditions $q_1,q_2 : \mathsf{A} \to \mathbb{P}$:
\[
  q_1 \sqsubseteq q_2 \Rightarrow \mathit{wp}(q_1) \sqsubseteq \mathit{wp}(q_2).
\]
For example, $\ret\,x$ is a monotone predicate transformer, as the condition
$q_1 \sqsubseteq q_2 \Rightarrow (\ret\,x)(q_1) \sqsubseteq (\ret\,x)(q_2)$
unfolds into
$\left(\forall (a : \mathsf{A}).~q_1(a) \Rightarrow q_2(a)\right) \Rightarrow q_1(x) \Rightarrow q_2(x)$
which trivially holds. $\bind~m~f$ is also a monotone predicate 
transformer, provided that both $m$ and $f$ are monotone predicate 
transformers. 

Moreover, not only predicates are ordered, but also predicate
transformers themselves can be ordered in a pointwise manner:
\[
\mathit{wp}_1 \sqsubseteq \mathit{wp}_2 \quad\mbox{iff}\quad \forall (q : \mathsf{A} \to \mathbb{P}).~\mathit{wp}_1(q) \sqsubseteq \mathit{wp}_2(q).
\]
Then, a natural question is whether the monadic structure preserves
such order on predicate transformers.  For weakest precondition
monads, both $\ret$ and $\bind$ preserve the order
above.

For stateful programs with state type $\mathsf{S}$, preconditions must refer
to the initial state and postconditions to the result and final state.
Following~\citet{dm4all}, we can do
so systematically by applying the state monad transformer $\StT$ to
the $\WP$ monad described above. Now the postcondition
predicates are not only on the underlying type, but also on the state
($\mathsf{A} \times \mathsf{S} \to \mathbb{P}$), while preconditions
are predicates over the initial state ($\mathsf{S} \to \mathbb{P}$).
This gives the following monadic structure on
$\StTType{\WP}{S}{A} \coloneq
(\mathsf{A} \times \mathsf{S} \to \mathbb{P}) \to \mathsf{S} \to \mathbb{P}$,
which is essentially the same as above, but threading around the
state:
\begin{equation}
  \label{eqn:monad-wp}
\begin{aligned}
\ret\,x
  &\coloneq \lambda~s\,\mathit{post}.\,\mathit{post}~(x,s),\\
\bind\,w\,f
  &\coloneq \lambda~s\,\mathit{post}.\,
     w\,s\,\bigl(\lambda (x,s').\,f\,x\,s'\,\mathit{post}\bigr).
\end{aligned}
\end{equation}
Again, these predicates are naturally ordered, and we restrict 
the carrier to predicate transformers that are monotone.
For this monad, the monadic structure ($\ret$, $\bind$) also
respects the monotonicity condition.

\section{Strongest Postcondition Monads}\label{sec:spmonad}
In contrast to $\mathit{wps}$, $\mathit{sps}$ map preconditions to
postconditions.
In the base case, an $\mathit{sp}$ is therefore a function from
predicates on $1$ to predicates on $\mathsf{A}$, of type
$(1 \to \mathbb{P}) \to (\mathsf{A} \to \mathbb{P}) \cong \mathbb{P} \to (\mathsf{A} \to \mathbb{P})$.
Unlike for $\mathit{wps}$, there is no well-known monad of a similar shape to
inspire our definitions of \ret and \bind.

\subsection{Base $\mathit{sp}$ Monads}\label{sec:base-sp-monads}

Maillard et al.~\cite{dm4all} define a base $\mathit{sp}$
monad, $\MonadType{\SPost}{A} \coloneqq \mathbb{P} \to
(\mathsf{A} \to \mathbb{P})$, with the following
\ret and \bind, inspired by the treatment of \texttt{skip} and
sequencing in the classical strongest postcondition calculus~\cite{DijkstraS90}:
\begin{equation*}
  \label{eqn:angelic-sp}
  \tag{ang}
\begin{aligned}
\ret\,v
  &:= \lambda \mathit{pre}~a.\,\mathit{pre} \land a = v,\\
\bind\,m\,f
  &:= \lambda \mathit{pre}~b.\,
     \exists a.\,f~a~(m~\mathit{pre}~a)~b.
\end{aligned}
\end{equation*}
To satisfy the monad laws, they consider monotone $\mathit{sps}$
of the smaller type
$(\text{pre:}\mathbb{P}) \to (\mathsf{A} \to \mathbb{P}_{\text{/pre}})$,
where $\mathbb{P}_{\text{/pre}} = \{p : \mathbb{P}~|~p \Rightarrow \text{pre}\}$.
This smaller type suffices for the monad laws, and we 
observe that it is, in fact, equivalent to a carrier of 
disjunctive, strict $\mathit{sps}$.

Disjunctivity and strictness (along with conjunctivity and termination, introduced
shortly) are properties of predicate transformers that form classes 
closed under sequential composition. Therefore, they are usually 
called homomorphism properties. 
Using an $\mathit{sp}$ as an example, for every non-empty family of preconditions $(p_i)_{i \in I}$, the
disjunctivity condition on {\em sp} states that:
\[
  \mathit{sp}\Bigl(\bigvee_{i \in I} p_i\Bigr) = \bigvee_{i \in I} \mathit{sp}(p_i)
\]
where the disjunction of predicates is taken pointwise:
\[
  \Bigl(\bigvee_{i \in I} p_i\Bigr)(x) \quad\mbox{iff}\quad \exists (i : I), p_i(x).
\]
An {\em sp} is strict when it maps $\bot$ to $\bot$, \IE
preserves empty disjunctions.
The disjunctivity condition is defined with respect to preconditions 
of any shape. In order to provide intuition, we briefly 
switch to the stateful setting, where computations take 
an initial state and (potentially) produce an output. 

For disjunctivity, \citet{BackW98} showed that disjunctive predicate transformers correspond to 
angelic interpretations of program semantics. For now, fix an output $r$
that the computation can produce, and consider the initial states from
which it can produce $r$.
An angelic $\mathit{sp}$
includes $r$ in the postcondition $\mathit{sp}(p)$ when {\em some} of these
initial states satisfy the precondition $p$.
For a precondition $p = p_1 \lor p_2$,
some initial state that produces $r$
satisfies $p_1 \lor p_2$ ($r \in \mathit{sp}(p_1 \lor p_2)$)
only if that initial state satisfies $p_1$ ($r \in \mathit{sp}(p_1)$)
{\em or} satisfies $p_2$ ($r \in \mathit{sp}(p_2)$), with
the same argument working for any non-empty family of preconditions.
The precondition
is checked on each initial state separately.

For strictness, \citet{BackW98} showed that strict predicate transformers correspond to 
a total interpretation of program semantics.
A total interpretation of an $\mathit{sp}$ requires that the postcondition 
only includes outputs that the computation produces from an initial 
state in the precondition. Therefore, if the precondition is empty, $\bot$,
then no such outputs exist, $\bot$.
We therefore define   
$\MonadType{\ASP}{A}$ as the angelic, total monad,
with a carrier restricted to strict, disjunctive $\mathit{sps}$, meaning
$\mathit{sps}$ preserving all disjunctions, whether empty or non-empty, represented by 
the following refinement type:
\[
\{\mathit{sp} : \mathbb{P} \to (\mathsf{A} \to \mathbb{P}) ~|~
  \forall I~(p : I \to \mathbb{P})~a.~
  \mathit{sp}\bigl(\exists i.\,p(i)\bigr)(a) \iff \exists i.\,\mathit{sp}(p(i))(a)\}
\]
and with the same \ret and \bind as (\ref{eqn:angelic-sp}).
We note that the strictness condition is covered by the empty case. 

Two other properties of interest are conjunctivity and termination.
An {\em sp} is conjunctive when for every non-empty family of preconditions $(p_i)_{i \in I}$
\[
  \mathit{sp}\Bigl(\bigwedge_{i \in I} p_i\Bigr) = \bigwedge_{i \in I} \mathit{sp}(p_i)
\]
where the conjunction of predicates is taken pointwise:
\[
  \Bigl(\bigwedge_{i \in I} p_i\Bigr)(x) \quad\mbox{iff}\quad \forall (i : I).~p_i(x).
\]
An {\em sp} is terminating when it maps $\top$ to $\top$, \IE
preserves empty conjunctions. 

For conjunctivity, \citet{BackW98} showed that conjunctive predicate transformers correspond 
to demonic interpretations of program semantics.
Again, fix an output $r$ and consider the initial states from which the
computation can produce it. A demonic $\mathit{sp}$ includes $r$ in
$\mathit{sp}(p)$ when {\em all} of these states satisfy $p$.
Every initial state that
produces $r$ satisfies $p_1 \land p_2$ only if every such initial state
satisfies $p_1$ {\em and} every such initial state satisfies $p_2$.
Here, the precondition
is also checked on each initial state separately.
The main difference is that an angelic $\mathit{sp}$ looks for one initial state that
{\em satisfies} the precondition, while a demonic $\mathit{sp}$ looks for
one initial state that {\em falsifies} it.

For termination, \citet{BackW98} showed that terminating predicate transformers 
correspond to a partial interpretation of program semantics. 
A partial interpretation of an $\mathit{sp}$ requires that the postcondition 
includes outputs that no initial state produces, as well as outputs 
produced from initial states in the precondition. Therefore, if all 
initial states satisfy the precondition, the $\top$ precondition, then this represents all 
possible outputs, either produced from some initial state or produced 
by none, which is the $\top$ postcondition.

We care about these classes for two reasons. These classes are carrier
restrictions under which our other $\mathit{sp}$ monads are lawful.
Without disjunctivity or conjunctivity,
the associativity of our stateful \bind
fails
We work through a counterexample in \cref{app:law-counterexamples}.
Additionally, they match the four interpretations of nondeterminism and
nontermination, so each of our monads contains exactly the predicate
transformers of programs under its interpretation.

Keeping the \ret and \bind from above while changing only the carrier
restriction to conjunctive, terminating $\mathit{sps}$ does {\em not} yield
the corresponding $\MonadType{\DSLP}{A}$ monad. The restricted carrier has
the following refinement type, where the empty case expresses termination:
\[
\{\mathit{sp} : \mathbb{P} \to (\mathsf{A} \to \mathbb{P}) ~|~
  \forall I~(p : I \to \mathbb{P})~a.~
  \mathit{sp}\bigl(\forall i.\,p(i)\bigr)(a) \iff \forall i.\,\mathit{sp}(p(i))(a)\}
\]
We obtain a
lawful demonic, liberal %
$\MonadType{\DSLP}{A}$ monad by also modifying the \ret and \bind as follows:
\begin{equation*}
  \label{eqn:demonic-sp}
  \tag{dem}
\begin{aligned}
\ret\,v
  &:= \lambda \mathit{pre}~a.\,a = v \Rightarrow \mathit{pre},\\
\bind\,m\,f
  &:= \lambda \mathit{pre}~b.\,
     \forall a.\,f~a~(m~\mathit{pre}~a)~b.
\end{aligned}
\end{equation*}
The total \ret says that $a$ is produced and equals $v$, and
$\mathit{pre}$ holds ($\mathit{pre} \land a = v$). The partial one says that {\em
  if} $a$ is produced and equals $v$, then $\mathit{pre}$ holds
($a = v \Rightarrow \mathit{pre}$), which vacuously admits every $a \neq v$. The
angelic \bind says that $b$ is produced through {\em some} intermediate
value $a$ and $\mathit{pre}$ holds ($\exists a$). The demonic one says that, for
{\em every} intermediate value $a$, if $b$ is produced through $a$, then
$\mathit{pre}$ holds ($\forall a$).

To construct $\MonadType{\DSP}{A}$ and $\MonadType{\ASLP}{A}$, we restrict their
carriers to conjunctive, strict $\mathit{sps}$ and to disjunctive, terminating
$\mathit{sps}$, respectively. To define \ret and \bind, we build on the
observation of \citet[\S3.2.4]{VerschtK25} that $\mathit{dsp}$ and
$\mathit{aslp}$ predicate transformers do not have an inductive
definition but can be defined in terms
of $\mathit{asp}$ and $\mathit{dslp}$ transformers.
For \ret, we get:
\[
\begin{aligned}
\ret_{\vDSP}\,v
  &:= \ret_{\vASP}~v \land \ret_{\vDSLP}~v, \\
\ret_{\vASLP}\,v
  &:= \ret_{\vASP}~v \lor \ret_{\vDSLP}~v
\end{aligned}
\]

Defining \bind is not so straightforward though. Simply taking
$\bind~m~f$ = $\bind_{\vASP}~m~f$ $(\land/\lor)$ $\bind_{\vDSLP}~m~f$
does not type-check, since $\bind_{\vASP}$ and $\bind_{\vDSLP}$
expect $\mathit{asp}$ and $\mathit{dslp}$ predicate transformers, whereas
 $\mathsf{m}$ and $\mathsf{f}~a$ are $\mathit{dsp}$ 
or $\mathit{aslp}$ predicate transformers. 
Let us consider just $\MonadType{\DSP}{A}$. We utilize the fact that a $\mathit{dsp}$ predicate
transformer can be decomposed into 
an $\mathit{asp}$ component and a $\mathit{dslp}$
component, whose conjunction characterizes the same set of results as
the $\mathit{dsp}$ predicate transformer itself.
We calculate this $\mathit{asp}$ component as follows:
\[
\begin{aligned}
&\mathsf{aspOf} : \MonadType{\DSP}{A} \to \MonadType{\ASP}{A}\\
&\mathsf{aspOf}~\mathit{sp}
  := \lambda~\mathit{pre}~a.\, \mathit{sp}~\top~a \land
  \neg(\mathit{sp}~(\neg \mathit{pre})~a)
\end{aligned}
\]
The first conjunct, $\mathit{sp}~\top~a$, states that $a$ is produced. The
second conjunct, $\neg(\mathit{sp}~(\neg \mathit{pre})~a)$, rules out that
$\mathit{pre}$ does not hold while $a$ is produced, so $\mathit{pre}$
holds or $a$ is not produced. 
The conjunction of the two admits
$a$ exactly when $a$ is produced and $\mathit{pre}$ holds, which is the
behavior of an $\mathit{asp}$ predicate transformer.

We calculate the $\MonadType{\DSLP}{A}$ component as follows:
\[
\begin{aligned}
&\mathsf{dslpOf} : \MonadType{\DSP}{A} \to \MonadType{\DSLP}{A}\\
&\mathsf{dslpOf}~\mathit{sp}
  := \lambda~\mathit{pre}~a.\, \mathit{sp}~\mathit{pre}~a \lor
  \neg(\mathit{sp}~\top~a)
\end{aligned}
\]
The first disjunct, $\mathit{sp}~\mathit{pre}~a$, states that $\mathit{pre}$ holds and 
$a$ is produced. The second disjunct, $\neg(\mathit{sp}~\top~a)$, states that
$a$ is not produced. The disjunction of the two admits $a$ exactly when 
$a$ being produced implies $\mathit{pre}$ holding, 
which is the behavior of a 
$\mathit{dslp}$ predicate transformer. 
Both ``exactly when''
steps are classical. Deriving $\mathit{pre}$ from $a$ being produced
together with $\neg(\neg\mathit{pre} \land a \text{ produced})$, and turning
the implication back into the disjunction, use excluded middle.

We prove that for any
$\mathit{sp} : \MonadType{\DSP}{A}$, $\mathit{sp}$ admits the same outputs as the
conjunction $\mathsf{aspOf}~\mathit{sp} \land \mathsf{dslpOf}~\mathit{sp}$
(\cref{app:dsp-components}). We can therefore define
\bind as follows, for $m : \MonadType{\DSP}{A}$ and
$f : \mathsf{A} \to \MonadType{\DSP}{B}$:
\[
\begin{aligned}
\bind_{\vDSP}\,m\,f
  &:= \lambda~\mathit{pre}~b.\,
     \bind_{\vASP}~(\mathsf{aspOf}~m)~(\lambda a. \mathsf{~aspOf~(f~}a))~\mathit{pre}~b~\land \\
     &\phantom{{}:= \lambda~\mathit{pre}~b.\,}\bind_{\vDSLP}~(\mathsf{dslpOf}~m)~(\lambda a. \mathsf{~dslpOf~(f~}a))~\mathit{pre}~b.
\end{aligned}
\]
This gives us a lawful $\MonadType{\DSP}{A}$ monad. 
   For $\MonadType{\ASLP}{A}$, we have similar functions
$\mathsf{aspOf}$ and $\mathsf{dslpOf}$, but we construct \bind via the {\em
disjunction} of the $\mathit{asp}$ and $\mathit{dslp}$ components.
In \cref{app:aslp-components}, we give these function definitions and
prove that the disjunction admits the same outputs as the original
$\mathit{aslp}$ transformer.

This gives us four ``pure'' base $\mathit{sp}$ monads. When reasoning about monadic
programs like stateful programs, we would like to apply monad transformers to
pure $\mathit{sp}$ monads to derive effectful $\mathit{sp}$ monads, as is
done on the $\mathit{wp}$ side~\citep{Maillard19}. Unfortunately, using the
state monad transformer to derive their stateful versions gives an unexpected
type: $\mathsf{S} \to \mathbb{P} \to (\mathsf{A} \times \mathsf{S} \to \mathbb{P})$. Such $\mathit{sps}$
take an initial state and a separate proposition and return a postcondition,
whereas stateful $\mathit{sps}$ should take a precondition that ranges over
initial states and return a postcondition. This is a lawful monad, but it is
unclear how to interpret its elements as stateful
$\mathit{sps}$.

\subsection{Base $\mathit{sp}$ Arrows}\label{sec:base-sp-arrows}

To remedy this, we observe that a monad is only parameterized by an
output type, 
and an $\mathit{sp}$ that wants to lift the type of its preconditions automatically, 
like on the $\mathit{wp}$ side, would benefit from a structure that 
is parameterized over both an input and output type.
Based on this observation, we
propose using \citeauthor{Hughes00}' arrows and arrow transformers to 
derive stateful $\mathit{sps}$~\citep{Hughes00}. Arrows are parameterized by {\em both} 
an input type and an output type, which 
gives precise control over how lifting affects both the input and output.

In its basic form, an arrow
$\mathsf{Arr} : \mathit{Type} \to \mathit{Type} \to \mathit{Type}$ can be defined
as a structure with the following operations:
\[
\begin{aligned}
  \arr{} &: (\mathsf{B} \to \mathsf{A}) \to \ArrowType{Arr}{B}{A}, \\
  \arrcomp{}{} &: \ArrowType{Arr}{B}{A} \to \ArrowType{Arr}{A}{C} \to \ArrowType{Arr}{B}{C}, \\
  \first{} &: \ArrowType{Arr}{B}{A} \to \ArrowType{Arr}{B \times C}{A \times C}
\end{aligned}
\]
Intuitively, $\arr{}$ lifts a function to an arrow. The $\arrcomp{x}{y}$
operation composes two arrows $x$ and $y$ together by feeding the output of $x$
to the input of $y$, and $\first{}$ takes an arrow $x$ and makes it an arrow that
additionally carries an ``environment''. The $\first{}$ operator is also known as
the \emph{strength} of arrows~\citep{arrows-strong-monad}. These 
arrows need to satisfy particular arrow laws, detailed in~\cref{app:arrow-laws}.

All monads yield arrows, but not all arrows can be turned into monads~\citep{LindleyWY11}. Given
a monad $\mathsf{M}$, we can derive an arrow as follows:
$$\ArrowType{Arr}{B}{A} \coloneqq \mathsf{B} \to \MonadType{M}{A}$$
In \citet{Hughes00}, this is known as the Kleisli operation, 
producing a Kleisli arrow. This
means that arrows are more general, but less expressive than
monads. In this paper, we are only interested in structures as
expressive as monads, so we consider arrows
$\mathsf{Arr}$ with an additional \emph{application} method, or
$\app$:
$$\app : \ArrowType{Arr}{\ArrowType{Arr}{B}{A} \times B}{A}$$
This variant of arrows is also defined by \citeauthor{Hughes00}. The
$\app$ method itself is an arrow, which takes another arrow of type
$\ArrowType{Arr}{B}{A}$ and a value of type $\mathsf{B}$ as input. It
\emph{applies} the arrow to the value of $\mathsf{B}$ to obtain an output of
type $\mathsf{A}$. This operator gives arrows the same
expressiveness as monads~\cite{Atkey2011}. Indeed, given an arrow $\mathsf{Arr}$ with application,
we can turn it into a monad $\mathsf{M}$ by fixing the input:
$\MonadType{M}{A} \coloneq \ArrowType{Arr}{1}{A}$. The
monad operations can be derived as follows:
\[
\begin{aligned}
\ret\,v &\coloneqq \arr{(\lambda \_. v)}, \\
\bind\,m\,f &\coloneqq \arrcomp{(\arrcomp{\arr{(\lambda x. (x, x))}}{\first{m}})}{(\arrcomp{\arr{(\lambda (x,i). (f\,x, i))}}{\app})}
\end{aligned}
\]
Fixing the input type of a lawful arrow with
application gives a lawful monad~\cite{Hughes00,LindleyWY11}, which 
allows us to prove
the monad laws once, generically, from the arrow laws and the laws of
$\app$.

We note that we do {\em not} use the Kleisli arrows constructed by 
applying the Kleisli operation to our base $\mathit{sp}$ monads. 
Every monad has a Kleisli arrow. Applying the state arrow transformer to
that arrow and then fixing its input yields the same monad as applying
the state monad transformer to the original monad. Therefore, we would still end up
with $\mathit{sps}$ of an unexpected shape. Instead, we use arrows 
where the input type gives the type over which preconditions are defined, so threading
state through the input changes the type of the preconditions. A monad,
indexed only by its output, cannot thread state through the type 
of the preconditions the same way.

Additionally, we can construct structured maps between 
arrows called {\em arrow morphisms}, which need to satisfy 
arrow morphism laws. 
An arrow morphism
is a parametric family $\theta_{\mathsf{X},\mathsf{Y}} :
\ArrowType{A}{X}{Y} \to \ArrowType{B}{X}{Y}$ such that
$\theta(\arr[A]{f}) = \arr[B]{f}$,
$\theta(\arrcomp[A]{x}{y}) = \arrcomp[B]{\theta(x)}{\theta(y)}$, and
$\theta(\first[A]{x}) = \first[B]{\theta(x)}$, along with laws for
choice and application (\autoref{app:arrow-laws}).

Returning to strongest postconditions, we construct base $\mathit{sp}$ arrows
that mirror our base $\mathit{sp}$ monads: $\ASPArr, \DSLPArr, \DSPArr,$ and
$\ASLPArr$. We start from monotone base $\mathit{sp}$ arrows
($\mathit{sps}_{\text{arr}}$) of the type $(\mathsf{B} \to \mathbb{P}) \to (\mathsf{A} \to \mathbb{P})$,
where $(\mathsf{B} \to \mathbb{P})$ represents preconditions and $(\mathsf{A} \to \mathbb{P})$
represents postconditions. Therefore, monotonicity states that for any two preconditions
$p_1,p_2: \mathsf{B} \to \mathbb{P},~p_1 \sqsubseteq p_2 \Rightarrow \mathit{sp}_{\text{arr}}(p_1) \sqsubseteq \mathit{sp}_{\text{arr}}(p_2)$.
We present $\ASPArr$ and $\DSLPArr$ and refer to
\cref{app:dsp-aslp-arrows} for $\DSPArr$ and $\ASLPArr$.
Their $\mathit{asp}$ and $\mathit{dslp}$ components are the $\mathsf{aspOf}$
and $\mathsf{dslpOf}$ of \autoref{sec:base-sp-monads} with $\top$, $\bot$, and
$\neg$ taken pointwise. They are defined for every input and output type, so
the stateful monads of \autoref{sec:stateful-sp} need no separate definition.

For $\ArrowType{\ASPArr}{B}{A}$, we restrict the arrow carrier to disjunctive,
strict $\mathit{sps}_{\text{arr}}$, mirroring the carrier restriction of the
$\MonadType{\ASP}{A}$ monad. For every family of preconditions
$(p_i)_{i \in I}$ on the input $\mathsf{B}$:
\[
  \mathit{sp}_{\text{arr}}\Bigl(\bigvee_{i \in I} p_i\Bigr) = \bigvee_{i \in I} \mathit{sp}_{\text{arr}}(p_i)
\]
The empty families give strictness, and non-empty families 
give disjunctivity. 
Monotonicity and strictness alone are not enough: without
disjunctivity, the arrow laws fail (\cref{app:first-fst}).
Fixing $\mathsf{B} = 1$ gives a monad carrier isomorphic to the carrier of
$\MonadType{\ASP}{A}$.

We then define the arrow operations so that fixing
$\mathsf{B} = 1$ gives the same \ret and \bind as
$\MonadType{\ASP}{A}$, which we show in \cref{app:fixing-input-asp}:
\[
\begin{aligned}
  \arr{f} &: (\mathsf{B} \to \mathsf{A}) \to \ArrowType{\ASPArr}{B}{A} := \lambda \mathit{pre}~a.~\exists b.~\mathit{pre}~b \land f~b = a, \\
  \arrcomp{f}{g} &: \ArrowType{\ASPArr}{B}{A} \to \ArrowType{\ASPArr}{A}{C} \to \ArrowType{\ASPArr}{B}{C} := \lambda \mathit{pre}~c.~g~(f~\mathit{pre})~c, \\
  \first{f} &: \ArrowType{\ASPArr}{B}{A} \to \ArrowType{\ASPArr}{B \times C}{A \times C} \\
    &:= \lambda \mathit{pre}~(a,c).~f~(\lambda b.~\mathit{pre}~(b, c))~a, \\
  \app &: \ArrowType{\ASPArr}{\ArrowType{\ASPArr}{B}{A} \times B}{A} := \lambda \mathit{pre}\,a.\, \exists f. \exists b. \mathit{pre}\,(f,b) \wedge f (\lambda x.\,x = b)\,a
\end{aligned}
\]
Intuitively, an output is admitted when some input that satisfies the precondition 
produces it. The definitions mirror the \ret and \bind of (\ref{eqn:angelic-sp}),
except that the precondition is now a predicate on an explicit input, over which we also quantify.
$\ret\,v$ admits $v$ when
$\mathit{pre}$ holds, while $\arr{f}$ admits $f\,b$ for some input $b$ that
satisfies $\mathit{pre}$. At $\mathsf{B} = 1$, $\arr{(\lambda\_.\,v)}$ is
$\ret\,v$. Composition passes the postcondition of $f$ to $g$ as its
precondition. Unlike \bind, it needs no quantifier over intermediate values,
since $g$, being disjunctive, quantifies over its own input. $\first{f}$ runs
$f$ on the first component of the input and passes the second component $c$
through unchanged, so $f$ receives the precondition restricted to the inputs
paired with $c$. Finally, $\app$ takes over the existential of \bind. Where
\bind quantifies over an intermediate value $a$ to select the continuation
$f\,a$, $\app$ quantifies over its input, an arrow $f$ and an argument $b$,
and runs $f$ from exactly $b$.

For $\ArrowType{\DSLPArr}{B}{A}$, we restrict the arrow carrier 
to conjunctive, terminating $\mathit{sps}_{\text{arr}}$, mirroring the 
carrier restriction of the $\MonadType{\DSLP}{A}$ monad.
For every family of preconditions
$(p_i)_{i \in I}$ on the input $\mathsf{B}$:
\[
  \mathit{sp}_{\text{arr}}\Bigl(\bigwedge_{i \in I} p_i\Bigr) = \bigwedge_{i \in I} \mathit{sp}_{\text{arr}}(p_i)
\]
The empty families give termination, and non-empty families give 
conjunctivity. 
Fixing $\mathsf{B} = 1$ gives a monad carrier isomorphic to the carrier of
$\MonadType{\DSLP}{A}$. 

The relevant arrow operations are:
\[
\begin{aligned}
  \arr{f} &: (\mathsf{B} \to \mathsf{A}) \to \ArrowType{\DSLPArr}{B}{A} := \lambda \mathit{pre}~a.~\forall b.~f~b = a \Rightarrow \mathit{pre}~b, \\
  \app &: \ArrowType{\DSLPArr}{\ArrowType{\DSLPArr}{B}{A} \times B}{A} \\
    &:= \lambda \mathit{pre}~a.~\forall f.~\forall b.~(\forall p.~f~p~a \Rightarrow p~b) \Rightarrow \mathit{pre}~(f, b)
\end{aligned}
\]
The definitions of $\arrcomp{}{}$ and $\first{}$ are the same as for
$\ArrowType{\ASPArr}{B}{A}$ so we omit them. Fixing $\mathsf{B} = 1$ again
gives the \ret and \bind of $\MonadType{\DSLP}{A}$
(\cref{app:fixing-input-dslp}). 
Intuitively, an output is admitted when every input 
that produces it satisfies the precondition, which holds vacuously for 
outputs and no input produces. 
As for $\ASPArr$, the definitions follow the \ret and \bind of
(\ref{eqn:demonic-sp}) but quantify over the explicit input, now universally.
Where $\ret\,v$ admits $a$ when $a = v$ implies $\mathit{pre}$, $\arr{f}$
admits $a$ when every input $b$ with $f\,b = a$ satisfies $\mathit{pre}$; at
$\mathsf{B} = 1$, $\arr{(\lambda\_.\,v)}$ is $\ret\,v$. Composition and
$\first{}$ are unchanged, since a conjunctive $g$ quantifies universally over
its own input. Here, $\app$ takes over the universal quantifier of
\bind, so it admits $a$ when every arrow $f$ and argument $b$ from which $f$
produces $a$ satisfy $\mathit{pre}$.
Following \citet{Atkey2011}, a lawful
$\app$ can also be seen as a property of the arrow rather than extra
structure, reflecting that the arrow corresponds to a monad.

\subsection{Stateful $\mathit{sp}$ Arrows and Monads}\label{sec:stateful-sp}

Recall that a lawful arrow, or one that satisfies the arrow laws, 
with application gives a lawful monad once its input
type is fixed. If we start from a lawful arrow of type
$\ArrowType{Arr}{B}{A} = (\mathsf{B} \to \mathbb{P}) \to (\mathsf{A} \to \mathbb{P})$, then
applying the state arrow transformer~\citep{Hughes00} gives
$\StTArrType{Arr}{S}{B}{A} = (\mathsf{B} \times \mathsf{S} \to \mathbb{P}) \to (\mathsf{A} \times \mathsf{S} \to \mathbb{P})$.
Fixing $\mathsf{B} = 1$ gives a lawful monad whose type is isomorphic to the
expected shape of stateful $\mathit{sps}$,
$(1 \times \mathsf{S} \to \mathbb{P}) \to (\mathsf{A} \times \mathsf{S} \to \mathbb{P}) \cong (\mathsf{S} \to \mathbb{P}) \to (\mathsf{A} \times \mathsf{S} \to \mathbb{P})$.

We apply the state arrow transformer to each of our base $\mathit{sp}$
arrows and fix $\mathsf{B} = 1$. We present the resulting $\ASPSt$ and 
$\DSLPSt$ monad definitions 
here and refer to \cref{app:stateful-sp} for 
the arrow definitions, along with the $\DSPSt$ and $\ASLPSt$ monad definitions.
We check by hand that these definitions match our intuition. 

Unfolding the derived \ret and \bind of $\MonadType{\ASPSt}{S, A}$ and
$\MonadType{\DSLPSt}{S, A}$ gives:
\begin{equation*}
  \label{eqn:angelic-sp-stateful}
  \tag{ang-st}
\begin{aligned}
\ret_{\vASP}\,v
  &\equiv \lambda~\mathit{pre}~(a, s).\;\mathit{pre}~s \land a = v, \\
\bind_{\vASP}\,m\,f
  &\equiv \lambda~\mathit{pre}~(b, s).\;\exists a.\;f~a~(\lambda s'.~m~\mathit{pre}~(a, s'))~(b, s),
\end{aligned}
\end{equation*}
\begin{equation*}
  \label{eqn:demonic-sp-stateful}
  \tag{dem-st}
\begin{aligned}
\ret_{\vDSLP}\,v
  &\equiv \lambda~\mathit{pre}~(a, s).\;a = v \Rightarrow \mathit{pre}~s, \\
\bind_{\vDSLP}\,m\,f
  &\equiv \lambda~\mathit{pre}~(b, s).\;\forall a.\;f~a~(\lambda s'.~m~\mathit{pre}~(a, s'))~(b, s).
\end{aligned}
\end{equation*}
These are exactly the base definitions with the precondition now a
predicate on states. \ret checks the precondition on the current state,
and \bind hands the continuation $f~a$ the postcondition
produced by $m$ at the value $a$ and precondition $\mathit{pre}$, which is a predicate on intermediate states.

\subsection{Monads vs. Arrows}

We now revisit the mismatch between the stateful $\mathit{sp}$ carriers
obtained using monad and arrow transformers.  Focusing on angelic strongest
postconditions, we show how the two constructions become isomorphic after
appropriately restricting their carriers. As noted
in~\autoref{sec:base-sp-monads}, applying the state monad transformer to the
base monad $\ASP$ on $\mathsf{A}$ yields $\StTType{\ASP}{S}{A}$, with carrier
$\mathsf{S} \to \mathbb{P} \to (\mathsf{A} \times \mathsf{S} \to
\mathbb{P})$.  By contrast, applying the state arrow transformer to $\ASPArr$
and fixing its input type to $\mathsf{1}$ yields
$\StTArrType{\ASPArr}{S}{1}{A}$, with carrier
$(\mathsf{1} \times \mathsf{S} \to \mathbb{P}) \to
(\mathsf{A} \times \mathsf{S} \to \mathbb{P})$.  Since the unit component is
trivial, this carrier is isomorphic to $(\mathsf{S} \to \mathbb{P}) \to
(\mathsf{A} \times \mathsf{S} \to \mathbb{P})$.
Although these carriers have different shapes, restricting both to monotone,
strict, and disjunctive predicate transformers makes them isomorphic:%
\[
\StTType{\ASP}{S}{A}
  \cong \StTArrType{\ASPArr}{S}{1}{A}.
\]
To see this isomorphism more concretely, consider the following maps:
\[
\begin{aligned}
  \mathsf{toArrow}(m)
    &\coloneqq \lambda~\mathit{pre}~r.~
       \exists s.~m~s~(\mathit{pre}~s)~r,\\
  \mathsf{fromArrow}(g)
    &\coloneqq \lambda~s~p~r.~
       g~(\lambda t.~p \land t = s)~r,
\end{aligned}
\]
which are inverses of each other.

Therefore, in this paper, we choose to work with the more familiar
stateful $\mathit{sp}$ presentation based on arrows. 
The $\mathit{wp}$ side needs no
arrows, since the state monad transformer gives the expected shape:
$\mathsf{S} \to (\mathsf{A} \times \mathsf{S} \to \mathbb{P}) \to \mathbb{P} \cong 
(\mathsf{A} \times \mathsf{S} \to \mathbb{P}) \to (\mathsf{S} \to \mathbb{P})$.
We mainly use arrows as a construction device to automatically 
derive lawful stateful arrows and, correspondingly, lawful 
stateful monads. Our study focuses on the resulting monads.
We seek to showcase the 
behavior of these $\mathit{sp}$ monads in contexts that previously 
only considered $\mathit{wp}$ monads, such as in forming 
Dijkstra monads. 
Additionally, working with monads keeps the treatment uniform with
the $\mathit{wp}$ side. 

From here on, we work with $\mathit{wp}$ and $\mathit{sp}$ monads. The taxonomy of
\autoref{sec:taxonomy}, the triples, and the ordering on predicate
transformers of \autoref{sec:verification} are all stated in terms of
\ret, \bind, monad morphisms, and monadic relations. 
The state arrow transformer 
adds state to the 
$\mathit{sp}$ monads and also lifts arrow morphisms.
These maps preserve arrow structure, just as monad morphisms preserve
monad structure.
Since, in our setting, an arrow morphism gives a monad morphism once the
input is fixed~\cite{Atkey2011},
we obtain the monad morphisms into the
$\mathit{sp}$ monads from arrow morphisms.
Therefore, the state arrow transformer 
also ``lifts'' our monad morphisms. 
The monadic relations are proved for the base and 
stateful cases separately, but we believe that a similar 
lift from the base to the stateful 
case via monad and arrow transformers, could be formalized in future work.

\subsection{Refined WP Monads for Nondeterminism and Nontermination}\label{sec:refined-wp-monads}

As needed for our taxonomy in~\autoref{sec:taxonomy}, we also construct
$\mathit{wp}$ monads containing exactly the $\mathit{wps}$ that match the four
interpretations of nondeterminism and nontermination. We use the same carrier
restrictions to construct these $\mathit{wp}$ monads: disjunctive for the
angelic and conjunctive for the demonic interpretation of nondeterminism, and
strict for the total and terminating for the partial interpretation of
nontermination. Unlike the $\mathit{sp}$ monads, $\WP$ does not {\em need} such
carrier restrictions in order to be a lawful monad, but they allow us to define
the monads $\AWP$, $\DWLP$, $\DWP$, and $\AWLP$, corresponding to our four
$\mathit{sp}$ monads.
All the {\em wp} monads share the \ret and \bind of $\WP$, since
\ret satisfies all four homomorphism properties and \bind preserves each
of them.
This makes the four $\mathit{wp}$ monads submonads of $\WP$, each carrier is
a subset of $\WP$'s that is closed under its \ret and \bind, so each inherits
$\WP$'s monad structure and laws, and any two $\mathit{wps}$ can be compared inside $\WP$
(\autoref{sec:implication-relation}).

\section{Taxonomy of Predicate Transformer Monads}\label{sec:taxonomy}

We recast the taxonomy 
of triples of \citet{VerschtK25} as a taxonomy of predicate transformer monads.
In \autoref{fig:monad-taxonomy}, diagram (a) details the relationships for our
base effects of nondeterminism and nontermination, for which some of the {\em sp}
monads are isomorphic (\autoref{sec:collapse})%
while diagram (b) details the relationships between the monads obtained after adding the state effect,
which corresponds to the stateful triple taxonomy of \citeauthor{VerschtK25}. We
replicate the implications, the contrapositive relations, and the Galois
connections between triples as monadic relations between the predicate
transformer monads.

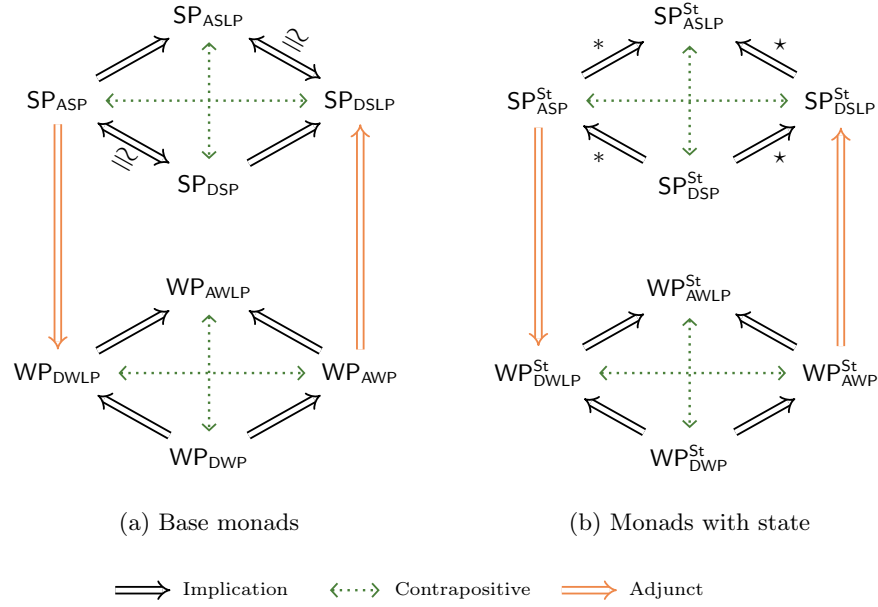
\begin{figure}[t]
  \centering
  \definecolor{taxgreen}{RGB}{74,128,58}
  \definecolor{taxorange}{RGB}{238,136,74}
  \newcommand{\montaxonomy}[2]{%
    \node[mon] (ASLP) at (0,7.3)    {$\csname ASLP#1\endcsname$};
    \node[mon] (ASP)  at (-2.7,5.9) {$\csname ASP#1\endcsname$};
    \node[mon] (DSLP) at (2.7,5.9)  {$\csname DSLP#1\endcsname$};
    \node[mon] (DSP)  at (0,4.5)    {$\csname DSP#1\endcsname$};
    \node[mon] (AWLP) at (0,2.8)    {$\csname AWLP#1\endcsname$};
    \node[mon] (DWLP) at (-2.7,1.4) {$\csname DWLP#1\endcsname$};
    \node[mon] (AWP)  at (2.7,1.4)  {$\csname AWP#1\endcsname$};
    \node[mon] (DWP)  at (0,0)      {$\csname DWP#1\endcsname$};

    #2

    \draw[imp] (DWP)  -- (DWLP);
    \draw[imp] (DWP)  -- (AWP);
    \draw[imp] (DWLP) -- (AWLP);
    \draw[imp] (AWP)  -- (AWLP);

    \draw[contra] (ASLP) -- (DSP);
    \draw[contra] (ASP)  -- (DSLP);
    \draw[contra] (DWLP) -- (AWP);
    \draw[contra] (AWLP) -- (DWP);

    \draw[adjunct] (ASP) -- (DWLP);
    \draw[adjunct] (AWP) -- (DSLP);
  }
  \begin{tikzpicture}[
      mon/.style={inner sep=4pt},
      imp/.style={-{Implies[]}, double, double distance=1.8pt, line width=0.7pt,
                  shorten >=1pt, shorten <=1pt},
      impboth/.style={imp, {Implies[]}-{Implies[]}},
      contra/.style={{Straight Barb[length=2.4pt,width=5.2pt,line width=0.7pt]}-%
                     {Straight Barb[length=2.4pt,width=5.2pt,line width=0.7pt]},
                     dotted, line width=0.9pt, draw=taxgreen,
                     shorten >=2pt, shorten <=2pt},
      adjunct/.style={-{Implies[]}, double, double distance=1.8pt, line width=0.7pt,
                     draw=taxorange, shorten >=1pt, shorten <=1pt},
      iso/.style={sloped, inner sep=0pt, font=\large},
      legend/.style={font=\scriptsize, anchor=west, inner sep=2pt,
                     text height=1.4ex, text depth=0.4ex},
      sublabel/.style={font=\small},
    ]
    \begin{scope}[x=0.74cm, y=0.8cm, name prefix=base-]
      \montaxonomy{}{%
        \draw[impboth] (ASP)  -- node[iso, below=3pt] {$\cong$} (DSP);
        \draw[impboth] (ASLP) -- node[iso, above=3pt] {$\cong$} (DSLP);
        \draw[imp] (ASP) -- (ASLP);
        \draw[imp] (DSP) -- (DSLP);
      }
      \node[sublabel] at (0,-1.1) {(a) Base monads};
    \end{scope}

    \begin{scope}[xshift=6.35cm, x=0.74cm, y=0.8cm, name prefix=state-]
      \montaxonomy{St}{%
        \draw[imp] (DSP)  -- (ASP);
        \draw[imp] (DSP)  -- (DSLP);
        \draw[imp] (ASP)  -- (ASLP);
        \draw[imp] (DSLP) -- (ASLP);
      }
      \path (DSP)  -- node[pos=0.5, below left=0pt]  {$\ast$}  (ASP);
      \path (DSP)  -- node[pos=0.5, below right=0pt] {$\star$} (DSLP);
      \path (ASP)  -- node[pos=0.5, above left=0pt]  {$\ast$}  (ASLP);
      \path (DSLP) -- node[pos=0.5, above right=0pt] {$\star$} (ASLP);
      \node[sublabel] at (0,-1.1) {(b) Monads with state};
    \end{scope}

    \begin{scope}[xshift=3.175cm, yshift=-1.75cm]
      \draw[imp]    (-4.45,0) -- (-3.65,0); \node[legend] at (-3.6,0) {Implication};
      \draw[contra] (-1.65,0) -- (-0.85,0); \node[legend] at (-0.8,0) {Contrapositive};
      \draw[adjunct] (1.45,0)  -- (2.25,0);  \node[legend] at (2.3,0)  {Adjunct};
    \end{scope}
  \end{tikzpicture}
  \caption{Monadic relations between the monads of predicate transformers, with the SP
    monads on top and the WP monads below, (a) for the base monads and (b) for
    the monads with state. In (a), $\DSP$ and $\ASLP$ are
    isomorphic ($\cong$) to $\ASP$ and $\DSLP$. Unmarked implications relate
    the elements of the two monads; in (b), $\ast$ and $\star$ mark implications
    that compare their $\mathsf{aspOf}$ and $\mathsf{dslpOf}$ components,
    respectively. An adjunction points from its left adjoint to its right adjoint.}
  \label{fig:monad-taxonomy}
\end{figure}

\subsection{Taxonomy of Triples}
\citeauthor{VerschtK25} present a taxonomy of triples. Each triple relates
a program $c$, a precondition $p$, and a postcondition $q$. Its validity
is expressed using a predicate transformer variant, for example as
$p \Rightarrow \mathit{dwlp} \llbracket c \rrbracket (q)$.
They relate triples using three relations: implications,
contrapositivity and Galois connections.
For implication, one triple implies another if validity of a triple in 
one logic implies its validity in the other logic. 
For contrapositivity, two triples are 
contrapositives if one triple is valid in one logic iff the 
other triple is valid in the other logic with the precondition 
{\em and} postcondition negated.
Given predicate transformers
$f : (X \to \mathbb{P}) \to (Y \to \mathbb{P})$ and
$g : (Y \to \mathbb{P}) \to (X \to \mathbb{P})$, $f$ and $g$ form a Galois
connection when for any $p : X \to \mathbb{P}$ and $ q : Y \to \mathbb{P}$, $~f(p) \Rightarrow q$ is valid 
if and only if $p \Rightarrow g(q)$ is valid.
Here, $f$ is the left adjoint and $g$ is the right adjoint. 

Analyzing their taxonomy 
of triples, which they construct by fixing 
the precondition, postcondition, and
program, except
that a contrapositive relation negates the precondition and the postcondition,
leads us to the following observations, which inspire our recast. 

First, for imperative programs, predicate transformers are defined 
recursively over the structure of the program. However, it need 
not be so. In our setting, predicate transformers can be studied apart from
programs and constructed using the \ret and \bind of the monadic
structure. Therefore, 
we claim that the relationships between triples can 
also be observed at the level of the bare predicate transformers, 
not indexed by a program, witnessed by monadic relations.

   Second, \citeauthor{VerschtK25}'s
taxonomy is comprised of triples that {\em over}- or {\em under}-approximate a
predicate transformer, where, for example,
$p \Rightarrow \mathit{dwlp} \llbracket c \rrbracket q$ is under-approximate and
$ \mathit{dwlp} \llbracket c \rrbracket q \Rightarrow p$ is over-approximate.
The implications that exist among under-approximate triples are dual to the ones
that exist among over-approximate triples. For example, in the under-approximate
version,
$(p \Rightarrow \mathit{dwp} \llbracket c \rrbracket(q)) \Rightarrow (p \Rightarrow \mathit{dwlp} \llbracket c \rrbracket(q))$.
However, for the over-approximate version,
$(\mathit{dwlp} \llbracket c \rrbracket(q) \Rightarrow p) \Rightarrow (\mathit{dwp} \llbracket c \rrbracket(q) \Rightarrow p)$.

Third, for contrapositives, each relation is between 
an under-approximate and an over-approximate 
triple, and it exists in a dualized form. For example, 
($p \Rightarrow \mathit{awp}\llbracket c \rrbracket(q)$) is contrapositive 
to ($\mathit{dwlp} \llbracket c \rrbracket(q) \Rightarrow p$). Dually,
($p \Rightarrow \mathit{dwlp} \llbracket c \rrbracket(q)$) is contrapositive 
to ($\mathit{awp} \llbracket c \rrbracket(q) \Rightarrow p$).

Fourth, for the adjunctions, the relations are dual with 
respect to which predicate transformer is the left adjoint or 
right adjoint, or vice versa. The Galois connection between 
predicate transformers when 
$\mathit{sp}$ is the left adjoint and the $\mathit{wp}$ is 
the right adjoint is represented by
($p \Rightarrow \mathit{dwlp} \llbracket c \rrbracket(q)) \iff  
(\mathit{asp} \llbracket c \rrbracket(p) \Rightarrow q$). 
The other Galois connection when 
$\mathit{wp}$ is the left adjoint and the $\mathit{sp}$ is 
the right adjoint is represented by
($q \Rightarrow \mathit{dslp} \llbracket c \rrbracket(p)) \iff  
(\mathit{awp} \llbracket c \rrbracket(q) \Rightarrow p$).

We demonstrate that, for implications and contrapositives, 
each pair of dual relationships is characterized by a single 
monadic relation, and
the two adjunctions are characterized by two monadic relations. 
Therefore, each implication or contrapositive edge of \autoref{fig:monad-taxonomy}
represents two arrows of \citeauthor{VerschtK25}'s diagram, and each
Galois edge one. 

\subsection{Monadic Relations}
\Citet{dm4all} introduce monadic relations as a way of relating a computation
monad to a specification monad, which are monads where \bind is monotone with 
respect to the ordering. 
 The definition of a monadic relation constrains \ret
and \bind but does not use an ordering, so it applies to any two monads. We use
monadic relations to relate our predicate transformer monads.

\begin{definition}[Monadic relation~\cite{dm4all}]
  A \emph{monadic relation} $\mathcal{R} : W_1 \leftrightarrow W_2$ between
  monads $W_1$ and $W_2$ is a family of relations
  $\mathcal{R}_{\mathsf{A}} : W_1\,\mathsf{A} \times W_2\,\mathsf{A} \to \mathbb{P}$,
  indexed by types $\mathsf{A}$, satisfying the following conditions:
\begin{itemize}
  \item[$\triangleright$] for every $a : \mathsf{A}$,
    $(\ret_{W_1}\,a)~\mathcal{R}_{\mathsf{A}}~(\ret_{W_2}\,a)$, 
  \item[$\triangleright$] and for every $w_1 : W_1\,A$, $w_2 : W_2\,A$,
        $k_1 : A \to W_1\,B$, and $k_2 : A \to W_2\,B$, we have
\end{itemize}
\[
  \frac{w_1~\mathcal{R}_{\mathsf{A}}~w_2 \qquad
        \forall (a : \mathsf{A}).~(k_1\,a)~\mathcal{R}_{\mathsf{B}}~(k_2\,a)}
       {(\bind_{W_1}\,w_1\,k_1)~\mathcal{R}_{\mathsf{B}}~(\bind_{W_2}\,w_2\,k_2)}
\]
\end{definition}

We build the following monadic relations based on relations between the
predicate transformers, as explained below. We present the monadic relations
between the stateful $\mathit{wp}$ and $\mathit{sp}$ monads to more 
directly connect to the stateful setting of the taxonomy of triples~\cite{VerschtK25}; however, the base
monadic relations can be similarly constructed for the relationships detailed in
the left diagram of \autoref{fig:monad-taxonomy}, the details of which we leave
to \cref{app:base-relations}. 

\subsubsection{Implication Relation}\label{sec:implication-relation}

The taxonomy of triples has implications among the $\mathit{wp}$
triples and among the $\mathit{sp}$ triples. We represent them 
as monadic relations among the $\mathit{wp}$ monads and among 
the $\mathit{sp}$ monads. Each monadic relation relates 
the elements of two monads, that is, two $\mathit{wps}$ or 
two $\mathit{sps}$.

The four stateful $\mathit{wp}$ monads have different carriers, but each is a
submonad of
$\StTType{\WPure}{S}{A} = (\mathsf{A} \times \mathsf{S} \to \mathbb{P}) \to (\mathsf{S} \to \mathbb{P})$,
the monad of stateful monotone $\mathit{wps}$, which we call $\WSt$ and discuss
in~\cref{sec:refined-wp-monads}, so any two $\mathit{wps}$ can be
compared within this monad. We compare them with the pointwise order:
\[
  wp_1 \sqsubseteq wp_2 \coloneqq \forall \mathit{post}~s.~wp_1(\mathit{post})(s) \Rightarrow wp_2(\mathit{post})(s)
\]
This order is a monadic relation between any two $\mathit{wp}$ monads, because
they share \ret and \bind, and \bind is monotone, which satisfies both
requirements of a monadic relation. If $wp_1 \sqsubseteq wp_2$, then every
under-approximate triple of $wp_1$ implies the same triple of $wp_2$, and every
over-approximate triple of $wp_2$ implies the same triple of $wp_1$. So the
opposite order gives the implications between under-approximate
triples.

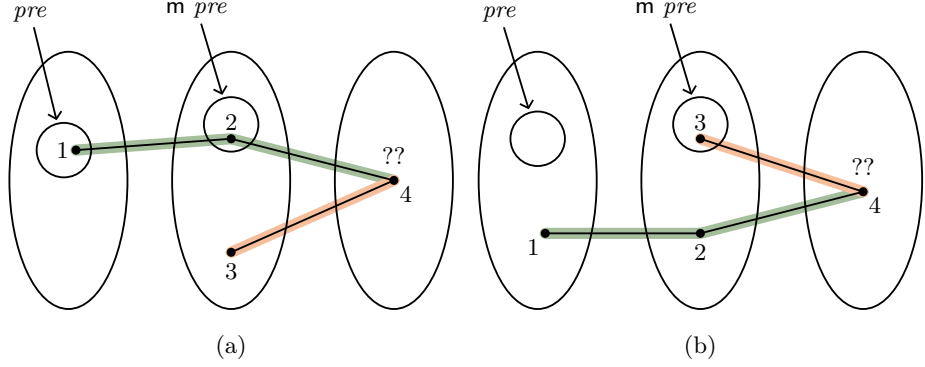
\begin{figure}[t]
  \centering
  \definecolor{taxgreen}{RGB}{74,128,58}
  \definecolor{taxorange}{RGB}{238,136,74}
  \begin{tikzpicture}[
      region/.style={draw, line width=0.7pt},
      pt/.style={circle, fill, inner sep=1.2pt},
      num/.style={font=\small, inner sep=2pt},
      trans/.style={line width=0.7pt},
      hl/.style={line width=4.2pt, line cap=round, line join=round},
      chosen/.style={draw, circle, line width=0.7pt, minimum size=0.72cm, inner sep=0pt},
      note/.style={font=\small, inner sep=2pt},
      ptr/.style={-{Straight Barb[length=2.4pt,width=5.2pt,line width=0.7pt]},
                  line width=0.7pt, shorten >=1pt},
      sublabel/.style={font=\small},
    ]
    \begin{scope}
      \coordinate (s1) at (0.1,0.4);
      \coordinate (s2) at (2.15,0.55);
      \coordinate (s3) at (2.15,-0.95);
      \coordinate (s4) at (4.3,0);

      \draw[hl, taxgreen!50]  (s1) -- (s2) -- (s4);
      \draw[hl, taxorange!55] (s3) -- (s4);

      \draw[region] (0,0)    ellipse [x radius=0.78, y radius=1.7];
      \draw[region] (2.15,0) ellipse [x radius=0.78, y radius=1.7];
      \draw[region] (4.3,0)  ellipse [x radius=0.78, y radius=1.7];

      \draw[trans] (s1) -- (s2);
      \draw[trans] (s2) -- (s4);
      \draw[trans] (s3) -- (s4);

      \foreach \s in {s1,s2,s3,s4} \node[pt] at (\s) {};
      \node[num, anchor=east, xshift=-0.5pt]   at (s1) {$1$};
      \node[num, anchor=south, yshift=1pt]     at (s2) {$2$};
      \node[num, anchor=north, yshift=-2pt]    at (s3) {$3$};
      \node[num, anchor=north west]            at (s4) {$4$};

      \node[chosen] (c1) at (-0.06,0.4) {};
      \node[chosen] (c3) at (2.15,0.74) {};
      \node[note] (pre)  at (-0.5,2.2)  {$\mathit{pre}$};
      \node[note] (mpre) at (1.7,2.25)  {$\mathsf{m}~\mathit{pre}$};
      \draw[ptr] (pre)  -- (c1);
      \draw[ptr] (mpre) -- (c3);

      \node[note, above=4pt] at (s4) {??};
      \node[sublabel] at (2.15,-2.2) {(a)};
    \end{scope}

    \begin{scope}[xshift=6.2cm]
      \coordinate (s1) at (0.1,-0.7);
      \coordinate (s2) at (2.15,-0.7);
      \coordinate (s3) at (2.15,0.55);
      \coordinate (s4) at (4.3,-0.15);

      \draw[hl, taxgreen!50]  (s1) -- (s2) -- (s4);
      \draw[hl, taxorange!55] (s3) -- (s4);

      \draw[region] (0,0)    ellipse [x radius=0.78, y radius=1.7];
      \draw[region] (2.15,0) ellipse [x radius=0.78, y radius=1.7];
      \draw[region] (4.3,0)  ellipse [x radius=0.78, y radius=1.7];

      \draw[trans] (s1) -- (s2);
      \draw[trans] (s2) -- (s4);
      \draw[trans] (s3) -- (s4);

      \foreach \s in {s1,s2,s3,s4} \node[pt] at (\s) {};
      \node[num, anchor=north east]            at (s1) {$1$};
      \node[num, anchor=north, yshift=-2pt]    at (s2) {$2$};
      \node[num, anchor=south, yshift=1pt]     at (s3) {$3$};
      \node[num, anchor=north west]            at (s4) {$4$};

      \node[chosen] (c1) at (0,0.55) {};
      \node[chosen] (c3) at (2.15,0.74) {};
      \node[note] (pre)  at (-0.5,2.2)  {$\mathit{pre}$};
      \node[note] (mpre) at (1.7,2.25)  {$\mathsf{m}~\mathit{pre}$};
      \draw[ptr] (pre)  -- (c1);
      \draw[ptr] (mpre) -- (c3);

      \node[note, above=4pt] at (s4) {??};
      \node[sublabel] at (2.15,-2.2) {(b)};
    \end{scope}
  \end{tikzpicture}
  \caption{Sequencing two computations, with initial state $1$, intermediate
  states $2$ and $3$, and final state $4$. The first computation $\mathsf{m}$
  maps $\mathit{pre}$ to $\mathsf{m}~\mathit{pre}$, which \bind passes to the
  continuation as its precondition.
  (a) The $\DSPSt$ case: state $4$ is reachable only from initial states in
  $\mathit{pre}$ (green path), so the postcondition should contain it. A \bind
  that only sees $\mathsf{m}~\mathit{pre}$ excludes it, because the
  continuation also reaches $4$ from state $3$, which is outside
  $\mathsf{m}~\mathit{pre}$ (orange path).
  (b) The $\ASLPSt$ case: state $4$ is reachable only from an initial state
  outside $\mathit{pre}$ (green path), so the postcondition should not contain
  it. A \bind that only sees $\mathsf{m}~\mathit{pre}$ includes it, because
  the continuation also reaches $4$ from state $3$, which is inside
  $\mathsf{m}~\mathit{pre}$ (orange path).}
  \label{fig:dsp-sequencing}
\end{figure}

Monads of $\mathit{sps}$ are not so straightforward. Without state, the monad
taxonomy collapses into the left taxonomy in
\autoref{fig:monad-taxonomy},
  since $\MonadType{\DSP}{A}$ and $\MonadType{\ASLP}{A}$
coincide with $\MonadType{\ASP}{A}$ and $\MonadType{\DSLP}{A}$, as described
in~\autoref{sec:collapse}. With state, the four stateful $\mathit{sp}$ monads
share the type of stateful monotone $\mathit{sps}$,
$\MonadType{\SPSt}{S, A} = (\mathsf{S} \to \mathbb{P}) \to (\mathsf{A} \times \mathsf{S} \to \mathbb{P})$,
and we, again, start from the pointwise order:
\[
  sp_1 \sqsubseteq sp_2 \coloneqq \forall \mathit{pre}~(a,s).~sp_1(\mathit{pre})(a,s) \Rightarrow sp_2(\mathit{pre})(a,s)
\]
However, this is not a monadic relation between any of the $\mathit{sp}$ monads
because the \bind requirement is not satisfied.
  Every implication between
$\mathit{sp}$ monads has $\DSPSt$ or $\ASLPSt$ at one end, and their \bind is
not monotone for $\sqsubseteq$, as demonstrated in \autoref{fig:dsp-sequencing}.
Since \bind is computed through the $\mathsf{aspOf}$ and $\mathsf{dslpOf}$
components, we instead compare components: a $\DSPSt$ element $\mathit{dsp}$ is
related to an $\ASPSt$ element $\mathit{asp}$ when
$\mathsf{aspOf}(\mathit{dsp}) \sqsubseteq \mathit{asp}$ and to a $\DSLPSt$
element $\mathit{dslp}$ when
$\mathsf{dslpOf}(\mathit{dsp}) \sqsubseteq \mathit{dslp}$. Likewise, an
$\ASLPSt$ element $\mathit{aslp}$ is related to them when
$\mathit{asp} \sqsubseteq \mathsf{aspOf}(\mathit{aslp})$ and when
$\mathit{dslp} \sqsubseteq \mathsf{dslpOf}(\mathit{aslp})$. These relations are
monadic, and are represented by the marks $\ast$ ($\mathsf{aspOf}$) and $\star$
($\mathsf{dslpOf}$) in \autoref{fig:monad-taxonomy}. Additionally, each implies
$\sqsubseteq$ between the $\mathit{sps}$ themselves: a $\mathit{dsp}$ is the
conjunction of its components and an $\mathit{aslp}$ their disjunction
(\autoref{lem:dsp-split} and \autoref{lem:aslp-split}), so
$\mathit{dsp} \sqsubseteq \mathsf{aspOf}(\mathit{dsp})$ and
$\mathsf{aspOf}(\mathit{aslp}) \sqsubseteq \mathit{aslp}$, and likewise for
$\mathsf{dslpOf}$. By transitivity,
$\mathsf{aspOf}(\mathit{dsp}) \sqsubseteq \mathit{asp}$ gives
$\mathit{dsp} \sqsubseteq \mathit{asp}$, and
$\mathit{asp} \sqsubseteq \mathsf{aspOf}(\mathit{aslp})$ gives
$\mathit{asp} \sqsubseteq \mathit{aslp}$.

\subsubsection{Contrapositive Relation}\label{sec:contrapositive-relation}

The taxonomy of triples has contrapositives between over-approximate 
and under-approximate $\mathit{wp}$ triples, and between the 
over-approximate and under-approximate $\mathit{sp}$ triples. 
We represent them as monadic relations between $\mathit{wp}$
monads and between $\mathit{sp}$ monads.

For $\mathit{wp}$, we again work with $\MonadType{\WSt}{S, A}$. We relate 
two $\mathit{wps}$ when each is the negation of the other applied to 
the negated postcondition:
\[
\begin{aligned}
  wp_1 \bowtie wp_2 &\coloneqq \forall \mathit{post}~s.~
     wp_1(\mathit{post})(s) \iff \neg\, wp_2(\lambda r.~\neg\, \mathit{post}(r))(s)
\end{aligned}
\]
This is a monadic relation between
the monads $\MonadType{\AWPSt}{S, A}$ and $\MonadType{\DWLPSt}{S, A}$ and the
monads $\MonadType{\DWPSt}{S, A}$ and $\MonadType{\AWLPSt}{S, A}$, pairing their
elements one to one. If $\mathit{wp}_1 \bowtie \mathit{wp}_2$, then the
under-approximate triple of $\mathit{wp}_1$ is the contrapositive of the
over-approximate triple of $\mathit{wp}_2$, that is,
$p \Rightarrow \mathit{wp}_1(q) \iff \mathit{wp}_2(\neg q) \Rightarrow \neg p$.
  The
relation is symmetric, so the over-approximate triple of $\mathit{wp}_2$ is
likewise the contrapositive of the under-approximate triple of $\mathit{wp}_1$.

For $\mathit{sp}$, we again work in $\MonadType{\SPSt}{S, A}$ and use the 
same relation, now negating the precondition:
\[
\begin{aligned}
sp_1 \bowtie sp_2 &\coloneqq \forall \mathit{pre}~r.~
     sp_1(\mathit{pre})(r) \iff \neg\, sp_2(\lambda s.~\neg\, \mathit{pre}(s))(r)
\end{aligned}
\]
This is a monadic relation between the monads $\MonadType{\ASPSt}{S, A}$ and $\MonadType{\DSLPSt}{S, A}$ 
 and between the monads $\MonadType{\DSPSt}{S, A}$ and $\MonadType{\ASLPSt}{S, A}$, 
pairing again their elements one to one. Unlike for implications, 
\bind preserves this relation, so no 
components are needed, even for $\DSPSt$ and $\ASLPSt$. 
The reason is that negation exchanges
the two components, $\neg\,\mathit{asp}(\neg\,\mathit{pre}) = \mathit{dslp}(\mathit{pre})$
and $\neg\,\mathit{dslp}(\neg\,\mathit{pre}) = \mathit{asp}(\mathit{pre})$, and,
by De Morgan, turns the conjunction used by $\DSPSt$ into the disjunction
used by $\ASLPSt$:
\[
\begin{aligned}
  \neg\,\mathit{dsp}(\neg\,\mathit{pre})
  &= \neg\,\bigl(\mathit{asp}(\neg\,\mathit{pre}) \land \mathit{dslp}(\neg\,\mathit{pre})\bigr)\\
  &= \neg\,\mathit{asp}(\neg\,\mathit{pre}) \lor \neg\,\mathit{dslp}(\neg\,\mathit{pre})\\
  &= \mathit{dslp}(\mathit{pre}) \lor \mathit{asp}(\mathit{pre})\\
  &= \mathit{aslp}(\mathit{pre}).
\end{aligned}
\]
If $\mathit{sp}_1 \bowtie \mathit{sp}_2$, then the under-approximate triple of
$\mathit{sp}_1$ is the contrapositive of the over-approximate triple of
$\mathit{sp}_2$, that is,
$q \Rightarrow \mathit{sp}_1(p) \iff \mathit{sp}_2(\neg p) \Rightarrow \neg q$,
and symmetrically.

\subsubsection{Adjunct Relations}
\label{sec:adjunct-relations}
The taxonomy of triples has two Galois connections connecting an over-approximate 
$\mathit{wp}$ triple and an under-approximate $\mathit{sp}$ triple, 
and vice versa.
Working with $\WSt$ and $\SPSt$, we define the relation
$\Ladj_{\mathsf{A}} : \MonadType{\SPSt}{S,A} \to \MonadType{\WSt}{S, A} \to \mathbb{P}$ capturing $\mathit{sps}$ that
are left adjoint to $\mathit{wps}$:
\[
\begin{aligned}
\Ladj_{\mathsf{A}}~sp~wp := \forall \mathit{pre}~\mathit{post}.~(sp(\mathit{pre}) \sqsubseteq
\mathit{post}) \iff (\mathit{pre} \sqsubseteq \mathit{wp}(\mathit{post}))
\end{aligned}
\]
In a dual fashion, the relation
$\Radj_{\mathsf{A}} : \MonadType{\SPSt}{S, A} \to \MonadType{\WSt}{S, A} \to \mathbb{P}$ relates
$\mathit{sps}$ right adjoint to $\mathit{wps}$:
\[
\begin{aligned}
\Radj_{\mathsf{A}}~sp~wp := \forall \mathit{pre}~\mathit{post}.~(wp(\mathit{post}) \sqsubseteq
\mathit{pre}) \iff (\mathit{post} \sqsubseteq \mathit{sp}(\mathit{pre}))
\end{aligned}
\]
Both relations are compatible with the monadic structures involved.
\begin{lemma}\label{lem:adjunct-relations}
\begin{enumerate}
\item $\Ladj$ is a monadic relation between $\SPSt$ equipped with the angelic
  monadic operations (\ref{eqn:angelic-sp-stateful}) and $\WSt$. Moreover, if
  $\Ladj_{\mathsf{A}}~sp~wp$, then $sp$ is disjunctive and strict while $wp$ is
  conjunctive and terminating.
\item $\Radj$ is a monadic relation between $\SPSt$ equipped with the demonic
  monadic operations (\ref{eqn:demonic-sp-stateful}) and $\WSt$. Moreover, if
  $\Radj_{\mathsf{A}}~sp~wp$, then $sp$ is conjunctive and terminating while
  $wp$ is disjunctive and strict.
\end{enumerate}
\end{lemma}
In both cases, the left adjoint preserves all disjunctions and the right
adjoint all conjunctions. As a consequence, $\Ladj$ restricts to a monadic
relation between $\ASPSt$ and $\DWLPSt$, and $\Radj$ to one between
$\DSLPSt$ and $\AWPSt$.
\Cref{sec:galois} will further show that both pair elements one to one.
The relations $\Ladj$ and $\Radj$ are dual in the following sense:
\[
(wp_1 \bowtie wp_2\ \wedge sp_1 \bowtie sp_2) \quad \Rightarrow \quad (\mathsf{L}_{\mathsf{A}}~sp_1~wp_1 \iff \mathsf{R}_{\mathsf{A}}~sp_2~wp_2)
\]
 As a result, one can traverse the above diagrams along the
adjunction relation and across the diamonds along the contrapositive
relation.

\subsection{Base Case Collapse}\label{sec:collapse}

Without state, $\MonadType{\DSP}{A}$ and $\MonadType{\ASLP}{A}$ coincide with
$\MonadType{\ASP}{A}$ and $\MonadType{\DSLP}{A}$. They have the same elements,
\ret, and \bind, and hence the same natural order. 
This is why the left
diagram of \autoref{fig:monad-taxonomy} relates them by isomorphisms and 
two-headed implications. 

Every element of these monads is the transformer of a program, a relation
$R$ from the initial point to results (\autoref{sec:base-sp-monads}), and
for every $R$,
\[
\begin{aligned}
  \mathit{dsp}(\mathit{pre})(a) &= (R~a \land \mathit{pre}) \land (R~a \Rightarrow \mathit{pre})
     = R~a \land \mathit{pre} = \mathit{asp}(\mathit{pre})(a),\\
  \mathit{aslp}(\mathit{pre})(a) &= (R~a \land \mathit{pre}) \lor (R~a \Rightarrow \mathit{pre})
     = R~a \Rightarrow \mathit{pre} = \mathit{dslp}(\mathit{pre})(a).
\end{aligned}
\]
So the carriers coincide, and since \ret and \bind compute the transformers
of the returned and the sequenced program, they coincide too. The identity
is an isomorphism of ordered monads
$\MonadType{\DSP}{A} \cong \MonadType{\ASP}{A}$ and
$\MonadType{\ASLP}{A} \cong \MonadType{\DSLP}{A}$.

This is the reversibility collapse of \citet{VerschtK25}. For programs in
which every final state has at most one initial state, angelic and demonic
$\mathit{sps}$ coincide. Without state, every computation has this property,
so the collapse holds for the monads themselves, not only for particular
programs. With state, a result can be reached from two different
initial states, one inside $\mathit{pre}$ and one outside, as in
\autoref{fig:dsp-sequencing}, so they can be differentiated.

In summary, the relations of the left diagram of 
\autoref{fig:monad-taxonomy} hold for the base monads, and those of the right 
diagram hold with state. They are proved for the base and state case 
independently, but we believe that future work can introduce a 
way to lift these monadic relations via the state monad 
and arrow transformers.

\section{Specification of Monadic Programs Using Predicate Transformer Monads}\label{sec:verification}

Thus far, we have considered $\mathit{wps}$ and $\mathit{sps}$ independent 
of computations. However, in order to use a predicate transformer to specify a
program, a predicate transformer semantics should be fixed to interpret the behavior of the program.
For monadic programs, \citet{dm4all} propose to use a {\em monad morphism}
from the computation monad into the monad of predicate transformers to capture a
specific semantics.
A monad morphism $\theta : \MonadType{M}{A} \to \MonadType{W}{A}$ consists of a
family of maps for every type $\mathsf{A}$, 
satisfying the following equations:
\[
\begin{aligned}
\theta(\ret~v) &= \ret~v, \\
\theta(\bind~m~f) &= \bind~\theta(m)~(\lambda a.~\theta(f~a)).
\end{aligned}
\]
The predicate transformer produced by a monad morphism 
represents the most precise embodiment of the program semantics 
with respect to the interpretation. 
We use monad morphisms below in one of two ways:
to define a predicate transformer triple instantiated with a
$\mathit{(pre,post)}$ specification~\cite{loom}, {\em or} to compare the
predicate transformer of a program to other predicate transformer specifications
via an ordering on the predicate transformer monad~\cite{mumon,dm4all,dm4free}.

On the one hand, a $\mathit{(pre,post)}$ specification can be validated by checking:
$$\mathit{pre} \sqsubseteq \theta_{\WP}(m)(\mathit{post})$$
This approach is at the core of the Loom verification framework in Lean~\cite{loom}. 
Since this check works with any monad morphism, we can use {\em all} of
our predicate transformer monads to construct monadic variants of
{\em all} the triples from~\cite{VerschtK25}.

On the other hand, a $\mathit{wp}$ specification $w$ can be validated by checking:
$$\theta_{\WP}(m) \sqsubseteq \mathit{w}$$
In order to facilitate a structured interaction between composition and
the weakening of specifications,
the \bind of the monad of predicate transformers 
has to be monotone with respect to the ordering, as follows:
\[
  \frac{w_1 \sqsubseteq w_2 \qquad
        \forall (x : \mathsf{A}).~k_1\,x \sqsubseteq k_2\,x}
       {\bind^{\WP}\,w_1\,k_1 \sqsubseteq \bind^{\WP}\,w_2\,k_2}
\]
Monads with this property, also called specification monads, are integral to the
construction of Dijkstra monads and the heart of verification frameworks like
\fstar{}~\cite{mumon,dm4all,dm4free}. Of our new $\mathit{sp}$ monads, $\ASP$
and $\DSLP$ are specification monads and can, therefore, be used to construct
new Dijkstra monads based on $\mathit{sps}$.

The rest of this section works out and compares these two styles:
\autoref{sec:stateful-morphisms} gives monad morphisms from the stateful 
computation monad $\StTType{NDiv}{S}{A}$ into our eight stateful predicate transformer monads.
\autoref{sec:lifting-morphisms} derives them from morphisms
from the base computation monad $\mathsf{NDiv}$ into our base predicate
transformer monads. This lifting uses the state monad and arrow transformers
on the $\mathit{wp}$ and $\mathit{sp}$ sides, respectively.
\autoref{sec:triples} uses these morphisms to define under- and
over-approximate triples in the style of~\cite{loom} 
and discusses how they compose. 
\Cref{sec:dms} turns to the Dijkstra monad-style and 
discusses when the two styles coincide. Finally, we build 
Dijkstra monads over the six specification monads, excluding 
$\DSP$ and $\ASLP$, whose \bind{}s are not monotone for their natural orders.

\subsection{Constructing Stateful Monad Morphisms}\label{sec:stateful-morphisms}

As stated above, a common approach to derive a predicate transformer 
that interprets the behavior of a program is to construct a 
monad morphism from the computation monad to the monad of predicate 
transformers.
Such a morphism typically interperets the outcomes of 
the computation. For nondeterminism and nontermination, a computation 
applied to a start state gives the set of all its outcomes, each a result 
or divergence, which the morphism can interpret 
angelically (some result must satisfy the postcondition) or demonically 
(all of them must), forbiding divergence (total) or permitting it (liberal).

We start with a base computation monad with {\em only} nondeterminism 
and nontermination, which we call $\mathsf{NDiv}$, and construct the 
stateful version by applying the state monad transformer,
$\StTType{NDiv}{S}{A}$, which we will call $\mathsf{State}$.
A computation $m$ of the computation monad maps a start state $s$ to a non-empty set
$m~s$ of outcomes. An outcome $o$ is either a result $r : \mathsf{A} \times \mathsf{S}$, written $\result{r}$, or
divergence, written $\mathsf{div}$. We write $\theta^{\mathsf{St}}_{\vDWP}$ for the 
monad morphism into $\MonadType{\DWPSt}{S, A}$, with the interpreted effect on top and the variant below, and likewise for the other seven. 
The four $\mathit{wp}$ monad morphisms are:
\begin{equation}
  \label{eqn:monad-morphisms-stateful-wp}
\begin{aligned}
\theta^{\mathsf{St}}_{\vDWP}(m)
  &:= \lambda~\mathit{post}~s.\,\forall o \in m~s.\;
      \mathsf{match}~o~\mathsf{with}~\result{r} \mapsto \mathit{post}~r \mid \mathsf{div} \mapsto \bot,\\
\theta^{\mathsf{St}}_{\vDWLP}(m)
  &:= \lambda~\mathit{post}~s.\,\forall o \in m~s.\;
      \mathsf{match}~o~\mathsf{with}~\result{r} \mapsto \mathit{post}~r \mid \mathsf{div} \mapsto \top,\\
\theta^{\mathsf{St}}_{\vAWP}(m)
  &:= \lambda~\mathit{post}~s.\,\exists o \in m~s.\;
      \mathsf{match}~o~\mathsf{with}~\result{r} \mapsto \mathit{post}~r \mid \mathsf{div} \mapsto \bot,\\
\theta^{\mathsf{St}}_{\vAWLP}(m)
  &:= \lambda~\mathit{post}~s.\,\exists o \in m~s.\;
      \mathsf{match}~o~\mathsf{with}~\result{r} \mapsto \mathit{post}~r \mid \mathsf{div} \mapsto \top.
\end{aligned}
\end{equation}
The four $\mathit{sp}$ monad morphisms quantify over the start states and their outcomes:
\begin{equation}
  \label{eqn:monad-morphisms-stateful-sp}
\begin{aligned}
\theta^{\mathsf{St}}_{\vASP}(m)
  &:= \lambda~\mathit{pre}~r.\,\exists s.\;\exists o \in m~s.\;\mathit{pre}~s \land o = r,\\
\theta^{\mathsf{St}}_{\vDSLP}(m)
  &:= \lambda~\mathit{pre}~r.\,\forall s.\;\forall o \in m~s.\;o = r \Rightarrow \mathit{pre}~s,\\
\theta^{\mathsf{St}}_{\vDSP}(m)
  &:= \lambda~\mathit{pre}~r.\,\theta^{\mathsf{St}}_{\vASP}(m)~\mathit{pre}~r
      \land \theta^{\mathsf{St}}_{\vDSLP}(m)~\mathit{pre}~r,\\
\theta^{\mathsf{St}}_{\vASLP}(m)
  &:= \lambda~\mathit{pre}~r.\,\theta^{\mathsf{St}}_{\vASP}(m)~\mathit{pre}~r
      \lor \theta^{\mathsf{St}}_{\vDSLP}(m)~\mathit{pre}~r.
\end{aligned}
\end{equation}
As in \citet{VerschtK25}, where $\mathit{dsp}$ and $\mathit{aslp}$ are
defined as the conjunction and the disjunction of $\mathit{asp}$ and
$\mathit{dslp}$, we define
$\theta^{\mathsf{St}}_{\vDSP}$ and $\theta^{\mathsf{St}}_{\vASLP}$ as the
pointwise conjunction and disjunction of $\theta^{\mathsf{St}}_{\vASP}$
and $\theta^{\mathsf{St}}_{\vDSLP}$.

\subsection{Lifting Morphisms}\label{sec:lifting-morphisms}

Now, we describe how we could have derived those same 
stateful monad morphisms from base monad morphisms 
using the state monad and arrow transformers. 

For $\mathit{wps}$, we construct base monad morphisms from $\mathsf{NDiv}$ into the base 
$\mathit{wp}$ monads, which we write without the superscript. As shown below, they largely mirror 
their stateful counterparts:
\begin{equation}
\begin{aligned}
\theta_{\vDWP}(m)
  &:= \lambda~\mathit{post}.\,\forall o \in m.\;
      \mathsf{match}~o~\mathsf{with}~\result{a} \mapsto \mathit{post}~a \mid \mathsf{div} \mapsto \bot,\\
\theta_{\vDWLP}(m)
  &:= \lambda~\mathit{post}.\,\forall o \in m.\;
      \mathsf{match}~o~\mathsf{with}~\result{a} \mapsto \mathit{post}~a \mid \mathsf{div} \mapsto \top,\\
\theta_{\vAWP}(m)
  &:= \lambda~\mathit{post}.\,\exists o \in m.\;
      \mathsf{match}~o~\mathsf{with}~\result{a} \mapsto \mathit{post}~a \mid \mathsf{div} \mapsto \bot,\\
\theta_{\vAWLP}(m)
  &:= \lambda~\mathit{post}.\,\exists o \in m.\;
      \mathsf{match}~o~\mathsf{with}~\result{a} \mapsto \mathit{post}~a \mid \mathsf{div} \mapsto \top.
\end{aligned}
\end{equation}
The state monad transformer also lifts monad morphisms: if $\theta$ is a
monad morphism from $\mathsf{NDiv}$ to $\WP$, then applying it at each
initial state, $\lambda m~s.~\theta(m~s)$, is a monad morphism from
$\StTType{NDiv}{S}{A}$ to $\StTType{\WP}{S}{A} = \MonadType{\WSt}{S, A}$.
In fact, this is how the stateful $\mathit{wp}$ morphisms of
\cref{eqn:monad-morphisms-stateful-wp} are defined. 
We also use the exception monad transformer, $\mathsf{ExcT}$, below the state
monad transformer, giving monad morphisms from $\StTType{ExcT(NDiv)}{S}{A}$
to $\MonadType{\WStExc}{S, A}$, defined in \cref{app:exc-morphisms}.

For $\mathit{sps}$, we construct base \emph{arrow} morphisms from 
the Kleisli arrow of $\mathsf{NDiv}$ into the four base $\mathit{sp}$ arrows.
They largely mirror their stateful counterparts, and since arrows have a
non-unit input type, unlike the base $\mathit{sp}$ monads, the morphisms
into $\DSPArr$ and $\ASPArr$, and into $\DSLPArr$ and $\ASLPArr$, remain distinct:
\begin{equation}
\begin{aligned}
\theta_{\ASPArr}(f)
  &:= \lambda~\mathit{pre}~a.\,\exists b.\;\exists o \in f~b.\;\mathit{pre}~b \land o = a,\\
\theta_{\DSLPArr}(f)
  &:= \lambda~\mathit{pre}~a.\,\forall b.\;\forall o \in f~b.\;o = a \Rightarrow \mathit{pre}~b,\\
\theta_{\DSPArr}(f)
  &:= \lambda~\mathit{pre}~a.\,\theta_{\ASPArr}(f)~\mathit{pre}~a
      \land \theta_{\DSLPArr}(f)~\mathit{pre}~a,\\
\theta_{\ASLPArr}(f)
  &:= \lambda~\mathit{pre}~a.\,\theta_{\ASPArr}(f)~\mathit{pre}~a
       \lor \theta_{\DSLPArr}(f)~\mathit{pre}~a.
\end{aligned}
\end{equation}
The state arrow transformer likewise lifts arrow morphisms: if $\theta$ is an
arrow morphism from the Kleisli arrow of $\mathsf{NDiv}$,
$\mathsf{Kl}(\mathsf{NDiv})$, to an $\mathit{sp}$ arrow $\mathsf{Arr}$, then
the same $\theta$, at the input and output types $\mathsf{B} \times \mathsf{S}$
and $\mathsf{A} \times \mathsf{S}$, is an arrow morphism from
$\StTArrType{Kl(NDiv)}{S}{B}{A}$ to $\StTArrType{Arr}{S}{B}{A}$. Fixing the
input to $1$ turns it into a monad morphism from $\StTType{NDiv}{S}{A}$ to
the stateful $\mathit{sp}$ monad, since a stateful computation $m$ is the
arrow $\lambda(\ast, s).~m~s$ from $1 \times \mathsf{S}$. The stateful $\mathit{sp}$ morphisms of \cref{eqn:monad-morphisms-stateful-sp}
are exactly these lifts; for instance,
$\theta^{\mathsf{St}}_{\vASP}(m) = \theta_{\ASPArr}(\lambda(\ast, s).~m~s)$.
We also use the exception arrow transformer, $\mathsf{ExcArrT}$, below the
state arrow transformer, which in the same way gives monad morphisms from
$\StTType{ExcT(NDiv)}{S}{A}$ into the $\mathit{sp}$ monads with state and
exceptions, defined in \cref{app:exc-morphisms}.

\subsection{Triples from Monad Morphisms}\label{sec:triples}

We use the generic $\mathit{wp}$ triples of Loom~\cite{loom}, of the
following form:
\[ 
  \mathit{pre} \sqsubseteq \theta_{\WP}(m)(\mathit{post}) 
\]
where the above ordering corresponds to under-approximate 
triples and the flipped ordering corresponds to over-approximate 
triples. 

We define a generic $\mathit{sp}$ triple analogously:
\[ 
  \mathit{post} \sqsubseteq \theta_{\SP}(m)(\mathit{pre}) 
\]

For each stateful monad morphism of \autoref{sec:stateful-morphisms},
we fix a program $\mathit{m}$, a precondition $\mathit{p}$, and
a postcondition $\mathit{q}$, and construct the corresponding over-
and under-approximate
$\mathit{wp}$ and $\mathit{sp}$ triples, where the over-approximate version uses the opposite ordering on predicates
in the triple definitions above.

Each edge of diagram (b) of \autoref{fig:monad-taxonomy} yields a connection between triples; we give one
example per kind of edge.
$\ASLPSt$ is related to $\ASPSt$ by the implication relation,
so we obtain an implication between the over-approximate triples
$\theta^{\mathsf{St}}_{\vASLP}(m)(p) \sqsubseteq q$ and $\theta^{\mathsf{St}}_{\vASP}(m)(p) \sqsubseteq q$,
and the reverse one between the under-approximate triples.
$\ASPSt$ is related to $\DSLPSt$ by the 
contrapositive relation, so we obtain a contrapositive
relation between $\theta^{\mathsf{St}}_{\vASP}(m)(p) \sqsubseteq q$ and 
$\neg q \sqsubseteq \theta^{\mathsf{St}}_{\vDSLP}(m)(\neg p)$.
Finally, $\ASPSt$ is related to $\DWLPSt$ by the 
adjunction relation, so we obtain an adjunction, which
at this level is a Galois connection,
between $\theta^{\mathsf{St}}_{\vASP}(m)(p) \sqsubseteq q$ and
$p \sqsubseteq \theta^{\mathsf{St}}_{\vDWLP}(m)(q)$.
This allows us to reconstruct all of the 
triples in~\cite{VerschtK25} and derive the same relationships.

We can also use these triples 
to verify monadic programs with a concrete pre- and 
postcondition specification in a style inspired by Loom~\cite{loom}.
Notably, we can verify incorrectness logic triples, focusing here on two
variants: O'Hearn's incorrectness
logic~\cite{OHearn20}, where triple validity is determined by
$q \Rightarrow \mathit{asp}\llbracket c \rrbracket(p)$, and
Lisbon logic~\cite{Moller21,loom}, where triple validity 
is determined by $p \Rightarrow \mathit{awp}\llbracket c \rrbracket(q)$.

Consider the following monadic program $m : \MonadType{State}{\mathbb{N},
  1}$:
\[
m := \lambda s.
\begin{cases}
  \{\,((), 0),\; ((), s)\,\} & \text{if $s$ is even,} \\
  \{\,((), s)\,\}            & \text{otherwise.}
\end{cases}
\]
Evaluating the monad morphisms targeting $\AWPSt$ and $\ASPSt$ gives:
\[
\begin{aligned}
\theta^{\mathsf{St}}_{\vAWP}(m)
  &= \lambda~\mathit{post}~s.\;
     \mathit{post}~((), s) \lor (s \text{ is even} \land \mathit{post}~((), 0)), \\
\theta^{\mathsf{St}}_{\vASP}(m)
  &= \lambda~\mathit{pre}~((), s').\;
     \mathit{pre}~s' \lor (s' = 0 \land \exists s.\; s \text{ is even} \land \mathit{pre}~s).
\end{aligned}
\]
Consider the following precondition, $p := \lambda s.~s > 0$, and 
postcondition, $q := \lambda (\_, s').~s' = 0$. Now, evaluating the triple 
validity as shown above gives:
\[
\begin{aligned}
\theta^{\mathsf{St}}_{\vAWP}(m)(q) &= \lambda~s.\; s \text{ is even}
  && \text{so } p \sqsubseteq \theta^{\mathsf{St}}_{\vAWP}(m)(q) \text{ fails at } s = 1, \\
\theta^{\mathsf{St}}_{\vASP}(m)(p) &= \lambda~((), s').\; \top
  && \text{so } q \sqsubseteq \theta^{\mathsf{St}}_{\vASP}(m)(p) \text{ holds}.
\end{aligned}
\]

This example illustrates that the two incorrectness logics are indeed
different: as noted in the introduction, O'Hearn's incorrectness logic has no
{\em wp} characterization~\cite{VerschtK25}. Intuitively,
the Lisbon triple needs {\em every} initial state in $\mathit{p}$ to reach 
some final state in $\mathit{q}$, 
while incorrectness logic only needs every final state in $\mathit{q}$ to be reached 
by {\em some} initial state in $\mathit{p}$. 
If we slightly modify the precondition to $p' := \lambda s.~s > 0 \land s \text{~is even}$,
then both triples hold:
\[
\begin{aligned}
\theta^{\mathsf{St}}_{\vAWP}(m)(q) &= \lambda~s.\; s \text{ is even}
  && \text{so } p' \sqsubseteq \theta^{\mathsf{St}}_{\vAWP}(m)(q) \text{ holds}, \\
\theta^{\mathsf{St}}_{\vASP}(m)(p') &= \lambda~((), s').\; s' \text{ is even}
  && \text{so } q \sqsubseteq \theta^{\mathsf{St}}_{\vASP}(m)(p') \text{ holds}.
\end{aligned}
\]

\subsubsection{Composition}

Monadic programs are sequenced using \bind, where the continuation
takes the returned value of the first computation. Triples can be sequenced
similarly. Given $p \sqsubseteq \theta^{\mathsf{St}}_{\WP}(m)(q_0)$ and, for every returned
value $a$, $p_a \sqsubseteq \theta^{\mathsf{St}}_{\WP}(f~a)(q)$, assume also that
$q_0(a, s) \Rightarrow p_a(s)$ for all $a$ and $s$.
By monotonicity of $\theta^{\mathsf{St}}_{\WP}(m)$, we have
$p \sqsubseteq \theta^{\mathsf{St}}_{\WP}(m)(\lambda (a, s).~\theta^{\mathsf{St}}_{\WP}(f~a)(q)(s))$.
Here we recognize $\WP$'s \bind from \cref{eqn:monadic-operations-wp}.
By the monad morphism law, this simplifies to
$p \sqsubseteq \theta^{\mathsf{St}}_{\WP}(\bind~m~f)(q)$.
The over-approximate rule flips
the three orderings and side condition, $p_a(s) \Rightarrow q_0(a, s)$, and holds for the same reason, so all eight kinds of
$\mathit{wp}$ triples compose.

The $\mathit{sp}$ triples are less uniform. For $\ASPSt$ and
$\DSLPSt$, the analogous rules hold in both directions. This covers
Hoare logic, as the over-approximate $\ASPSt$ triple, and
incorrectness logic~\cite{OHearn20}, as the under-approximate one.
An under-approximate $\DSPSt$ triple is the conjunction of an
$\ASPSt$ triple and a $\DSLPSt$ triple, and so is an
over-approximate $\ASLPSt$ triple, so these compose through their
components.

The two remaining kinds, over-approximate $\DSPSt$ triples and
under-approximate $\ASLPSt$ triples, do {\em not} compose.
\autoref{fig:dsp-sequencing}(a) is a counterexample for $\DSPSt$:
take $p = \{1\}$, $q_0 = p_a = \{2\}$, and $q = \emptyset$. The first
computation takes $p$ exactly to $q_0$. From $p_a$, the second computation
gives the empty postcondition, because state $4$ is also reachable from
state $3$, which is outside of $p_a$. Both premises and the side condition
hold, but the conclusion fails: in the sequenced computation, state $4$ is
reachable only from $p$, since the first computation never produces
state $3$.
\autoref{fig:dsp-sequencing}(b) gives the dual counterexample for
$\ASLPSt$. In both cases, the intermediate assertions do not say
which intermediate states the first computation can actually produce.

\subsection{Predicate Transformers as Ordered Specifications}\label{sec:dms}

We next cover another way of specifying a monadic program: directly by a
predicate transformer, which is compared against the one computed by the monad
morphism. While the generic triple definitions differ between $\mathit{wp}$
and $\mathit{sp}$, comparing a specification given as a
predicate transformer $\mathit{pt}$ against the one computed by the monad
morphism takes the same form for all variants: an inequation in the order of
the predicate transformer monad,
\[
  \theta(m) \sqsubseteq \mathit{pt}
  \qquad\text{or}\qquad
  \mathit{pt} \sqsubseteq \theta(m),
\]
depending on the direction of approximation. The first states that
$\mathit{pt}$ over-approximates the transformer computed for $m$, and the
second that it under-approximates it.

As explained at the start of \autoref{sec:verification},
for specification monads, \bind must be monotone
with respect to the ordering. This allows us to arbitrarily
weaken specifications and compose them in the expected manner, as shown for $\mathit{wps}$
\[
  \frac{\theta_{\WP}(m) \sqsubseteq w_m \qquad
        \forall (a : \mathsf{A}).~\theta_{\WP}(f~a) \sqsubseteq w_f\,a}
       {\theta_{\WP}(\bind\,m\,f) \sqsubseteq \bind\,w_m\,w_f}
\]
where the conclusion follows from the monad morphism laws. The same rule
holds with the order reversed, for specifications that under-approximate.

For all the $\mathit{wp}$ monads, $\ASPSt$, and $\DSLPSt$, which are specification
monads, such comparisons compose as above and give rise to the corresponding
Dijkstra monads, which we will shortly define.
For $\DSPSt$ and $\ASLPSt$, \bind is not monotone with
respect to the natural ordering on $\mathit{sps}$ defined in~\autoref{sec:implication-relation}.

Such a comparison verifies a monadic program against a concrete
predicate transformer specification. 
We demonstrate this technique on Lisbon logic, for the same program as before and the following
$\AWPSt$ specification: $\lambda~\mathit{post}~s.\; s > 0 \land s \text{ is even} \land \mathit{post}~((), 0)$.
Evaluating the monad morphism gives:
\[
\begin{aligned}
\theta^{\mathsf{St}}_{\vAWP}(m)
  &= \lambda~\mathit{post}~s.\; \mathit{post}~((), s) \lor (s \text{ is even} \land \mathit{post}~((), 0)).
\end{aligned}
\]
Since Lisbon logic is under-approximate, the comparison is
$\mathit{pt} \sqsubseteq \theta^{\mathsf{St}}_{\vAWP}(m)$, which unfolds to
\[
\forall p~s.\quad s > 0 \land s \text{ is even} \land p~((), 0)
  \Rightarrow p~((), s) \lor (s \text{ is even} \land p~((), 0))
\]
and this holds by the second disjunct. 

In fact, this $\AWPSt$ predicate transformer is equivalent 
to the $\AWPSt$ pre- and postcondition pair from before. 
This is an established result between some sets of triples and orderings
over predicate transformers; for example, \citet{mumon} write pre- and postcondition
specifications as abbreviations for $\mathit{wps}$.
It is more natural to work with pairs of pre- and postconditions~\cite{Leino05}.
To connect these specifications to predicate transformers, Swamy et al. showed
that any pre- and postcondition pair, $(\mathit{pre}, \mathit{post})$,
can be embedded into a $\mathit{wp}$ as follows
\[ 
  \mathsf{pp2wp}(\mathit{pre, post}) := \lambda \mathit{p}~s.~\mathit{pre}~s \land \forall r.~\mathit{post~r} \Rightarrow \mathit{p~r}
\]
such that $\mathit{pre} \sqsubseteq \theta^{\mathsf{St}}_{\WP}(m)(\mathit{post})$ if and only if 
$\mathsf{pp2wp}(\mathit{pre, post}) \sqsubseteq \theta^{\mathsf{St}}_{\WP}(m)$.
  We name this 
embedding $\mathsf{pp2wp_{under}}$ to represent that it corresponds to 
the under-approximate $\mathit{wp}$ triple. 

Over-approximate $\mathit{wp}$ triples can also be expressed in terms of the
order on $\WP$ thanks to a dual embedding $\mathsf{pp2wp_{over}}$:
\[ 
  \mathsf{pp2wp_{over}}~(\mathit{pre}, \mathit{post}) := \lambda~\mathit{p}~s.~(\forall r.~\mathit{p}~r \Rightarrow \mathit{post}~r) \Rightarrow \mathit{pre}~s
\]
With both orderings flipped, we have that $\theta^{\mathsf{St}}_{\WP}(m)(\mathit{post}) \sqsubseteq \mathit{pre}$ if and only if 
$\theta^{\mathsf{St}}_{\WP}(m) \sqsubseteq \mathsf{pp2wp_{over}}(\mathit{pre, post})$.
Both embeddings produce monotone $\mathit{wps}$, but they may lie outside the carrier of the
$\mathit{wp}$ monad of the triple. The comparison therefore uses the order of
$\WSt$, of which the four $\mathit{wp}$ monads are submonads.

A natural question is whether similar embeddings exist for the 
two $\mathit{sp}$ specification monads. We show that they exist for over-approximate 
$\ASPSt$ triples and under-approximate $\DSLPSt$
triples, exactly those triples that are in a Galois connection 
with a $\mathit{wp}$ variant:
\[
\begin{aligned}
\mathsf{pp2asp_{over}}(\mathit{pre}, \mathit{post})
  &:= \lambda~p~r.~\exists s.~p~s \land (\mathit{pre}~s \Rightarrow \mathit{post}~r), \\
\mathsf{pp2dslp_{under}}(\mathit{pre}, \mathit{post})
  &:= \lambda~p~r.~\forall s.~(\mathit{post}~r \Rightarrow \mathit{pre}~s) \Rightarrow p~s
\end{aligned}
\]
such that $\theta^{\mathsf{St}}_{\vASP}(m)(\mathit{pre}) \sqsubseteq \mathit{post}$ if and only if
$\theta^{\mathsf{St}}_{\vASP}(m) \sqsubseteq \mathsf{pp2asp_{over}}(\mathit{pre}, \mathit{post})$
and 
$\mathit{post} \sqsubseteq \theta^{\mathsf{St}}_{\vDSLP}(m)(\mathit{pre})$ if and only if
$\mathsf{pp2dslp_{under}}(\mathit{pre}, \mathit{post}) \sqsubseteq \theta^{\mathsf{St}}_{\vDSLP}(m)$.

The other two $\mathit{sp}$ triples have no such embedding:

\begin{lemma}\label{lem:no-embedding}
Let $\mathsf{S}$ have two distinct states and let $\mathsf{A}$ be inhabited.
\begin{enumerate}
\item The under-approximate $\ASPSt$ triple, that is, the incorrectness
  triple, has no embedding: for some $\mathit{pre}$ and $\mathit{post}$,
  there is no element $w$ of $\ASPSt$ such that, for every $m$,
  $w \sqsubseteq \theta^{\mathsf{St}}_{\vASP}(m)$ holds exactly when
  $\mathit{post} \sqsubseteq \theta^{\mathsf{St}}_{\vASP}(m)(\mathit{pre})$
  does.
\item The over-approximate $\DSLPSt$ triple has no embedding: for some
  $\mathit{pre}$ and $\mathit{post}$, there is no element $w$ of $\DSLPSt$
  such that, for every $m$, $\theta^{\mathsf{St}}_{\vDSLP}(m) \sqsubseteq w$
  holds exactly when
  $\theta^{\mathsf{St}}_{\vDSLP}(m)(\mathit{pre}) \sqsubseteq \mathit{post}$
  does.
\end{enumerate}
\end{lemma}
\begin{proof}
We prove the first part; the second is dual, with conjunctivity in place of
disjunctivity, and we prove it in \cref{app:dslp-no-embedding}.
An element $w$ of $\ASPSt$ is strict and disjunctive, so it satisfies
$w(p) = \bigvee_{s \in p} w(\{s\})$: it commits to the start state from
which each result is reached, whereas the incorrectness triple only asks that
the results in $\mathit{post}$ be reachable from {\em some} state of
$\mathit{pre}$. Take distinct states $s_0$ and $s_1$, a result $r_0$,
$\mathit{pre} = \{s_0, s_1\}$, and $\mathit{post} = \{r_0\}$.
A program $m_1$ whose only result is $r_0$, reached only from $s_0$, and a
program $m_2$ whose only result is $r_0$, reached only from $s_1$, both
satisfy the triple, so such a $w$ would satisfy
$w \sqsubseteq \theta^{\mathsf{St}}_{\vASP}(m_1)$ and
$w \sqsubseteq \theta^{\mathsf{St}}_{\vASP}(m_2)$.
At a single start state $s$, one of the two bounds is always $\bot$, so
$w(\{s\}) = \bot$, and hence $w = \bot$ by disjunctivity.
But $\bot \sqsubseteq \theta^{\mathsf{St}}_{\vASP}(m)$ holds for every $m$,
so the triple would have to hold for every program. It does not hold for a
program that always diverges: such a program reaches no result, so $r_0$ is
not reachable from $\mathit{pre}$.
\end{proof}

Now, this does not stop one from specifying a program with an 
$\ASPSt$ predicate transformer, and some of those predicate 
transformers might correspond to incorrectness logic triples. 
However, there does not exist an embedding that converts 
an incorrectness triple into a corresponding $\ASPSt$ predicate 
transformer for every triple.

\subsubsection{Dijkstra Monads}

The monad laws make specifications compose. Monotonicity of \bind
extends this structured composition to weakened specifications, as described above.
This property is necessary for constructing Dijkstra monads. 

\begin{definition}[Dijkstra monad~{\cite[\S5]{dm4all}}]
A \emph{Dijkstra monad} $D$ over a specification monad $W$, that is, a monad
with a preorder $\sqsubseteq_A$ on each $W\,A$ for which $\bind^{W}$
is monotone, is given by:
\begin{itemize}
  \item[$\triangleright$] for every type $A$ and specification $w : W\,A$, a
    type $D\,A\,w$ of computations specified by $w$,
  \item[$\triangleright$] a return and a bind specified by the return and
    bind of $W$, satisfying the monad laws (which are well-typed because $W$
    satisfies them),
  \item[$\triangleright$] and a weakening along the order of $W$, compatible
    with \bind:
\end{itemize}
\[
\begin{gathered}
  \frac{v : A}{\ret^{D}\,v : D\,A\,(\ret^{W}\,v)}
  \qquad
  \frac{m : D\,A\,w_1 \qquad
        \forall (x : A).~f\,x : D\,B\,(w_2\,x)}
       {\bind^{D}\,m\,f : D\,B\,(\bind^{W}\,w_1\,w_2)}
  \\[1.5ex]
  \frac{m : D\,A\,w \qquad w \sqsubseteq_A w'}
       {\mathsf{weaken}\,m : D\,A\,w'}
\end{gathered}
\]
\end{definition}

The comparison rule for \bind at the start of this subsection is exactly
this typing rule for \bind, with $D\,A\,w$ being the computations $m$ such
that $\theta(m) \sqsubseteq w$. For specifications that under-approximate, the
computations with $w \sqsubseteq \theta(m)$ form a Dijkstra monad for the
reversed order, for which \bind is monotone as well. 
We construct Dijkstra monads indexed with specifications from $\WSt$, through
each of the four $\mathit{wp}$ monad morphisms, and from $\ASPSt$ and
$\DSLPSt$, which are specification monads, but not for $\DSPSt$ and
$\ASLPSt$.
We construct them from the monad morphisms from the
computation monads into these specification monads.
As shown in~\cite{dm4all}, such monad morphisms
naturally give rise to Dijkstra monads.
A key benefit of Dijkstra monads is that they only contain computations that
satisfy their specification, which was the foundation of verification frameworks
like \fstar{}.

As demonstrated above, triples for all the $\mathit{wp}$ monads,
over-approximate $\ASPSt$ triples, and under-approximate $\DSLPSt$ triples
can be represented by predicate transformers. The corresponding Dijkstra
monads can therefore contain {\em exactly} the computations that satisfy
those triples.
Additionally, we show that
all of the examples that we verified using the corresponding 
triple can also be proved to be elements within the 
corresponding Dijkstra monad.
For under-approximate $\ASPSt$ 
and over-approximate $\DSLPSt$, one can still specify 
programs by the predicate transformer, but there does not 
exist an embedding from the triple into a predicate transformer 
for an arbitrary triple.

\section{Galois Connections Between Specification Monads}\label{sec:galois}

Beyond the taxonomy established for the various monads of {\em wps} and
{\em sps}, one can wonder how expressive these variants are.
To that end, \citet{dm4all} show that $\WP$ and $\SPost$ can be
related through Galois connections
to two simpler monads $\PrePost$ and
$\Pred$, where $\WP^\op$ indicates that the order is flipped.
\begin{equation*}
  \label{eqn:galois-dm4all}
\begin{tikzcd}
	{\SPost} && {\Pred} && {\PrePost} && {\WP^\op}
	\arrow[""{name=0, anchor=center, inner sep=0}, shift left=2, from=1-1, to=1-3]
	\arrow[""{name=1, anchor=center, inner sep=0}, shift left=2, from=1-3, to=1-1]
	\arrow[""{name=2, anchor=center, inner sep=0}, shift left=2, from=1-3, to=1-5]
	\arrow[""{name=3, anchor=center, inner sep=0}, shift left=2, from=1-5, to=1-3]
	\arrow[""{name=4, anchor=center, inner sep=0}, shift left=2, from=1-5, to=1-7]
	\arrow[""{name=5, anchor=center, inner sep=0}, shift left=2, from=1-7, to=1-5]
	\arrow["\dashv"{anchor=center, rotate=90}, draw=none, from=1, to=0]
	\arrow["\dashv"{anchor=center, rotate=-90}, draw=none, from=2, to=3]
	\arrow["\dashv"{anchor=center, rotate=-90}, draw=none, from=4, to=5]
\end{tikzcd}
\tag{$\dagger$}
\end{equation*}
A Galois connection $f \dashv g$ between two specification monads $\mathsf{W_1}$ and
$\mathsf{W_2}$ consists of monotone natural transformations $f_A : \MonadType{W_1}{A} \to
\MonadType{W_2}{A}$ and $g_A : \MonadType{W_2}{A} \to \MonadType{W_1}{A}$, such
that $f_A\,w_1 \sqsubseteq^{\mathsf{W_2}} w_2$ iff $w_1 \sqsubseteq^{\mathsf{W_1}} g_A\,w_2$ for any
$w_1 : \MonadType{W_1}{A}, w_2 : \MonadType{W_2}{A}$ and $\mathsf{A}$.
$\PrePost$ consists of pairs of a precondition and a postcondition, while
$\Pred$ only retains the postcondition.
\citet{Maillard19} further showed that in the pure case, the Galois connection
between $\SPost$ and $\Pred$ on the left establishes an isomorphism of
specification monads.

In the rest of this section,
we extend this analysis in terms of Galois connections to the stateful case.
We consider submonads of monotone, stateful backward predicate transformers $\WSt$ and
monotone, stateful forward predicate transformers $\SPSt$ described by healthiness properties.
While the angelic (\ref{eqn:angelic-sp-stateful}) and demonic
(\ref{eqn:demonic-sp-stateful}) operations on $\SPSt$ do not satisfy all monad
laws, it is still meaningful to ask when these operations are preserved by a
natural transformation from or to other monads.\footnote{When considering merely monotone {\em sps}
  rather than conjunctive and strict ones, the demonic $\bind$ is more adequately
  expressed as
 $\lambda \mathit{pre}~b~s_0.\,
     \forall a~s_i.\,(\forall\,\mathit{pre}_i.\,f~a~\mathit{pre}_i~b \Rightarrow
     \mathit{pre}_i\,s_i) \Rightarrow m~pre~(a,s_i)$ as it satisfies some
     inequations refining the monad laws.}
Specifically, we compare it with $\MonadType{\PredSt}{S, A} := \mathsf{S} \to \mathsf{A} \times \mathsf{S}
\to \mathbb{P}$ ordered with pointwise logical implication and endowed with the monadic operations:
\begin{equation}
  \label{eq:monad-pred}
  \begin{aligned}
    \ret_\PredSt\,a &:= \lambda (s_0 : S)(r : A \times S).~r = (a, s_0) \\
    \bind_\PredSt\,m\,f &:= \lambda (s_0 : S)(r : B\times S).~ \exists (a : A)(s_i : S).~m\,s_0\,(a,s_i) \wedge f\,a\,s_i\,r
  \end{aligned}
\end{equation}
$\ret_\PredSt\,a$ asserts that the returned value is the given $a : \mathsf{A}$ and the state is left
unchanged, while $\bind_\PredSt\,m\,f$ for $m : \MonadType{\PredSt}{S,A}$ and $f :
\mathsf{A} \to \MonadType{\PredSt}{S,B}$ existentially quantifies over an intermediate result
value $a$ and state $s_i$, satisfying the postcondition $m$ and passed to the
continuation $f$.

\newcommand{\predtospL}{\varphi}
\newcommand{\sptopredL}{\varphi'}
\newcommand{\predtowpL}{\psi}
\newcommand{\wptopredL}{\psi'}

Recall from \cref{sec:adjunct-relations} that $\Ladj$ relates $\mathit{sps}$
that are left adjoint to $\mathit{wps}$.
We now show that any predicate $p : \MonadType{\PredSt}{S,A}$ induces such a pair
$\Ladj_{\mathsf{A}}~(\predtospL\,p)~(\predtowpL\,p)$ where $\predtospL :
\MonadType{\PredSt}{S, A} \to \MonadType{\SPSt}{S, A}$ and $\predtowpL :
\MonadType{\PredSt}{S, A} \to \MonadType{\WStOp}{S, A}$ are stateful variants of the
maps displayed at (\ref{eqn:galois-dm4all}), defined as
\begin{align*}
  \predtospL~p&:= \lambda (pre : \mathsf{S}\to\mathbb{P})(r : \mathsf{A} \times \mathsf{S}).~\exists s_0 : \mathsf{S}.~ pre\,s_0 \wedge p\,s_0\,r\\
  \predtowpL~p&:= \lambda (post : \mathsf{A}\times \mathsf{S} \to \mathbb{P})(s_0 : \mathsf{S}).~\forall r : \mathsf{A} \times \mathsf{S}.~ p\,s_0\,r \Rightarrow post\,r.
\end{align*}
Equipping $\SPSt$ with the angelic monadic structure (\ref{eqn:angelic-sp-stateful}), $\predtospL$ and $\predtowpL$ are both monotone monad morphisms.

Returning to $\Ladj$, given $p : \MonadType{\PredSt}{S, A}$, $pre : \mathsf{S} \to
\mathbb{P}$ and $post : \mathsf{A} \times \mathsf{S} \to \mathbb{P}$,
we have that
\begin{align*}
  \predtospL\,p\,pre \sqsubseteq post
  &\iff (\forall s_0~r.~pre\,s_0 \wedge p\,s_0\,r \Rightarrow post\,r) \\
  &\iff (\forall s_0~r.~pre\,s_0\Rightarrow p\,s_0\,r \Rightarrow post\,r)\\
  &\iff (\forall s_0.~pre\,s_0 \Rightarrow \forall r.~p\,s_0\,r \Rightarrow post\,r)\\
  &\iff pre \sqsubseteq \predtowpL\,p\,post,
\end{align*}
which establishes $\Ladj_{\mathsf{A}}~(\predtospL\,p)~(\predtowpL\,p)$.
This observation yields a characterization of the relation $\Ladj$ in the following
sense.
\begin{lemma}
  For $sp : \MonadType{\SPSt}{S, A}$ and $wp : \MonadType{\WStOp}{S, A}$, $\Ladj_{\mathsf{A}}~sp~wp$
  holds if and only if there exists $p : \MonadType{\PredSt}{S, A}$, such that $sp =
  \predtospL~p$ and $wp = \predtowpL~p$.
\end{lemma}
Intuitively, this characterization of $\Ladj$ holds because we can reconstruct a
predicate out of a disjunctive and strict {\em sp} or a conjunctive and
terminating {\em wp}.
Indeed, paralleling the pure case (\ref{eqn:galois-dm4all}), $\predtospL$ and $\predtowpL$ both have
right adjoints $\sptopredL : \MonadType{\SPSt}{S, A} \to \MonadType{\PredSt}{S, A}$
and $\wptopredL : \MonadType{\WStOp}{S, A} \to \MonadType{\PredSt}{S, A}$ given by
\begin{align*}
  \sptopredL~sp&:= \lambda (s_0 : \mathsf{S})(r :\mathsf{A} \times \mathsf{S}).~sp\,(\lambda s.~s = s_0)\,r\\
  \wptopredL~wp&:= \lambda (s_0 : \mathsf{S})(r:\mathsf{A} \times \mathsf{S}).~\forall (p : \mathsf{A} \times \mathsf{S} \to \mathbb{P}).~wp\,p\,s_0 \Rightarrow p\,r.
\end{align*}
$\sptopredL$ preserves the demonic monadic structure on $\SPSt$ while $\wptopredL$
does not preserve $\WStOp$'s \bind.
We observed that $\Ladj$ restricted to a relation between $\ASPSt$ and $\DWLPSt$.
We can further this remark by inspecting the nuclei of the adjunctions
$\predtospL \dashv \sptopredL$ and $\predtowpL \dashv \wptopredL$:
(1) every $p : \MonadType{\PredSt}{S, A}$ is a fix point of both $\sptopredL \circ \predtospL$
and $\wptopredL \circ \predtowpL$, (2) the fix points of $\predtospL \circ
\sptopredL$ are exactly elements of $\ASPSt$, and (3) the fix points of
$\predtowpL \circ \wptopredL$ are exactly elements of $\DWLPSt^\op$.
As a result, we obtain the characterization of $\Ladj$ and a chain of correspondences
\[
  \ASPSt \quad \cong \quad \PredSt \quad \cong \quad \DWLPSt^\op.
\]

\newcommand{\predtospR}{\xi}
\newcommand{\sptopredR}{\xi'}
\newcommand{\predtowpR}{\zeta}
\newcommand{\wptopredR}{\zeta'}
We can now turn to the case of $\Radj$, where a similar story unfolds with
respect to the specification monad $\PredStOp$, with the same carrier and monadic
structure as $\PredSt$
but with order reversed.
Consider the maps
$\predtospR : \MonadType{\PredStOp}{S, A} \to \MonadType{\SPSt}{S, A}$,
$\predtowpR : \MonadType{\PredStOp}{S, A} \to \MonadType{\WStOp}{S, A}$,
$\sptopredR : \MonadType{\SPSt}{S, A} \to \MonadType{\PredStOp}{S, A}$, and
$\wptopredR : \MonadType{\WStOp}{S, A} \to \MonadType{\PredStOp}{S, A}$, defined as
\begin{align*}
  \predtospR~p&:= \lambda (pre : \mathsf{S} \to \mathbb{P})(r : \mathsf{A} \times \mathsf{S}).~\forall s_0 : \mathsf{S}.~p\,s_0\,r \Rightarrow pre\,s_0\\
  \predtowpR~p&:= \lambda (post : \mathsf{A} \times \mathsf{S} \to \mathbb{P})(s_0 : \mathsf{S}).~\exists r : \mathsf{A} \times \mathsf{S}.~p\,s_0\,r \wedge post\,r\\
  \sptopredR~sp&:= \lambda (s_0 : \mathsf{S})(r : \mathsf{A} \times \mathsf{S}).~\forall (pre : \mathsf{S} \to \mathbb{P}).~sp\,pre\,r \Rightarrow pre\,s_0\\
  \wptopredR~wp&:= \lambda (s_0 : \mathsf{S})(r : \mathsf{A} \times \mathsf{S}).~wp\,(\lambda (r' : \mathsf{A} \times \mathsf{S}).~r = r')\,s_0.
\end{align*}
We can show that all maps are monotone, natural in $\mathsf{A}$, and the corresponding
pairs form Galois connections: $\sptopredR \dashv \predtospR$ and $\wptopredR
\dashv \predtowpR$.
Moreover, $\predtospR, \predtowpR$ and $\sptopredR$ are monad morphisms if we
equip $\SPSt$ with the demonic monadic operations
(\ref{eqn:demonic-sp-stateful}), while $\wptopredR$
fails to preserve $\WStOp$'s \bind.

For any $p : \MonadType{\PredStOp}{S, A}$, we have that
$\Radj_{\mathsf{A}}~(\predtospR\,p)~(\predtowpR\,p)$.
Indeed, given $pre : \mathsf{S} \to \mathbb{P}$ and $post : \mathsf{A} \times \mathsf{S} \to \mathbb{P}$,
\begin{align*}
  \predtowpR\,p\,post \sqsubseteq pre
  &\iff \forall s_0 : \mathsf{S}.~(\exists r : \mathsf{A} \times \mathsf{S}.~ p\,s_0\,r \wedge post\,r) \Rightarrow  pre\,s_0 \\
  &\iff \forall (s_0 : \mathsf{S})(r : \mathsf{A} \times \mathsf{S}).~post\,r \Rightarrow p\,s_0\,r \Rightarrow pre\,s_0 \\
  &\iff \forall r : \mathsf{A} \times \mathsf{S}.~post\,r \Rightarrow \forall s_0 : \mathsf{S}.~p\,s_0\,r \Rightarrow pre\,s_0 \\
  &\iff post \sqsubseteq \predtospR\,p\,pre.
\end{align*}
We also obtain a characterization of $\Radj$ in terms of $\PredStOp$.
\begin{lemma}
  For $sp : \MonadType{\SPSt}{S, A}$ and $wp : \MonadType{\WStOp}{S, A}$, $\Radj_{\mathsf{A}}~sp~wp$
  holds if and only if there exists $p : \MonadType{\PredStOp}{S, A}$, such that $sp =
  \predtospR~p$ and $wp = \predtowpR~p$.
\end{lemma}

Again, we can investigate the nuclei of the adjunctions $\sptopredR \dashv
\predtospR$ and $\wptopredR \dashv \predtowpR$: (1) every $p :
\MonadType{\PredStOp}{S,A}$ is a fix point of both $\sptopredR \circ \predtospR$ and
$\wptopredR \circ \predtowpR$, (2) the fix points of $\predtospR \circ
\sptopredR$ are exactly elements of $\DSLPSt$, and (3) the fix points of
$\predtowpR \circ \wptopredR$ are exactly elements of $\AWPSt^\op$.
As a result, we get another chain of correspondences:
\[
  \DSLPSt \quad \cong \quad \PredStOp \quad \cong \quad \AWPSt^\op
\]

\section{Formalization in Rocq}

All definitions and results presented in this paper are formalized in Rocq
(about 28k lines), assuming the axioms of functional extensionality,
propositional extensionality, and excluded middle.

We use the triples of \autoref{sec:verification} to verify about 30 small
example programs, including all the sequential examples of~\citet{OHearn20}.
These examples cover state, state and exceptions, recursion,
and loops and together use all sixteen kinds of triples.
The programs are written as interaction
trees~\cite{XiaZHHMPZ20}, and a monad morphism maps them to the
computation monad of \autoref{sec:verification}.

\section{Related Work}

\paragraph{Hoare-like logics.}
Incorrectness logic was first proposed by \citet{reverse-hoare} under the name
reverse Hoare logic. The term ``incorrectness logic'' was coined by
\citet{OHearn20} later. Since then, many variants of it have been studied,
including partial incorrectness logic~\citep{ZhangK22}, underspecified
incorrectness logic~\citep{uil}, sufficient incorrectness logic~(a.k.a.~Lisbon
logic)~\citep{sil, u-turn},~\textit{etc.}\@ \Citet{VerschtK25} taxonomize
Hoare-like triples in an imperative setting and inspired our work. 
Starting from the classical strongest postcondition and the strongest liberal
postcondition of \citet{ZhangK22}, which they call $\mathit{asp}$ and
$\mathit{dslp}$, they complete the picture with two novel variants,
$\mathit{dsp}$ and $\mathit{aslp}$.
They also observe that $\mathit{aslp}$ and $\mathit{dsp}$ cannot be defined by
induction on the program structure and propose that they can be represented
using disjunctions and conjunctions of
$\mathit{asp}$ and $\mathit{dslp}$, respectively~\citep[\S3.2.4]{VerschtK25}. We extend their work to
monadic functional programs with different combinations of effects, and
construct a predicate transformer monad for each of the four $\mathit{sp}$
variants.

\paragraph{Outcome logics.}
Outcome logic~\citep{ZilbersteinDS23} unifies correctness and incorrectness
reasoning in a single program logic, whose semantics is expressed in a monad
and whose assertions describe the collection of outcomes that a program produces.
It expresses the $\mathit{dwlp}$ and
$\mathit{awp}$ triples but not the $\mathit{dwp}$ and $\mathit{awlp}$
triples, which need nontermination as an outcome~\citep{LiZS24}.
Total outcome logic~\citep{LiZS24} adds this outcome and expresses every logic
of the taxonomy of \citet{VerschtK25}, and logics in between, with one proof
system. 
It expresses the $\mathit{dsp}$ and
$\mathit{aslp}$ logics as combinations of two triples, without defining the
two transformers directly; here, they are elements of their own monads,
$\DSPSt$ and $\ASLPSt$.
Both logics use a monad to give imperative programs their semantics, and one
logic to state all the triples. We instead make the predicate transformers
themselves monads. Each triple is read off a monad morphism from the computation
monad, the relations of the taxonomy become relations between these monads, and
monad transformers extend these monads to state and exceptions. Unlike our work, outcome
logics also cover probabilistic choice.

\paragraph{Predicate transformers forming monads.}
\Citet{Jacobs15} studies predicate transformers as a monad, with
$\mathit{wp}$ as a monad morphism into it, \citet{HasuoGWP15} derives
$\mathit{wps}$ generically from monads enriched with an order, and
\citet{AguirreKK22} develop $\mathit{wps}$ in fibrations, where $\mathit{sp}$
arises as the left adjoint of $\mathit{wp}$, as in our relation $\Ladj$
(\autoref{sec:galois}). 
Instead of deriving predicate transformers from the computation monad, we
restrict the carriers of our monads to combinations of the homomorphism
properties of \citet{BackW98}.

\paragraph{Dijkstra monads and specification monads.}
\citet{fstar-pldi13} introduced Dijkstra monads in the context of \fstar{} for
reasoning intrinsically about stateful programs using backward predicate
transformers represented as elements of a continuation monad to stateful
predicates.
The Dijkstra monad methodology was then extended to a variety of
effects~\cite{dm4free, dm4all}.
\citet{dm4all} decompose Dijkstra monads as a triple of a computation monad
$\mathsf{M}$, a specification monad $\mathsf{W}$ and a monad morphism $\theta :
\mathsf{M} \to \mathsf{W}$ interpreting programs as specifications.
While their work mention forward predicate transformers, their strategy to
effect modularity using monad transformers does not apply directly to provide
an intuitive specification monad for stateful forward predicate transformers as
we show in this paper. Using arrow transformers, we obtain predicate transformers
of the expected shape.

\paragraph{Verification using monad morphisms.}
\Citet{loom} build foundational verifiers in Lean based on 
predicate transformer triples. They derive $\mathit{wp}$ calculi 
from ordered monad algebras and lift
them through a library of monad transformers, for state, exceptions,
divergence with total or partial correctness, and nondeterminism with demonic
or angelic choice. Their incorrectness proofs use angelic $\mathit{wps}$, as in
Lisbon logic, and they do not consider $\mathit{sps}$. Like us, they obtain
effect modularity from transformers, but they lift algebras through monad
transformers, whereas we lift monad morphisms through monad transformers and
arrow morphisms through arrow transformers.

\section{Conclusion and Future Work}

This work introduces strongest postcondition monads for all
combinations of angelic or demonic nondeterminism and total or liberal
interpretations of nontermination, and extends them with state and exceptions
by generalizing to arrows and applying arrow transformers. Together with the
corresponding weakest precondition monads, these allow us to recast the taxonomy of
triples presented by~\citet{VerschtK25} as a taxonomy of predicate
transformer monads, both for the base monads and for their versions with state.
We demonstrate how they can be used to specify simple monadic programs, with
and without state and exceptions, in a Loom-inspired style~\cite{loom} and in
the style of Dijkstra monads~\cite{dm4all}, and relate the two styles.
All our results are mechanized in Rocq.

Further work is needed to investigate whether these foundational
results can underpin a practical verification framework in the style of
Loom~\cite{loom} that supports {\em sps}, general {\em wps}, and both under- and
over-approximation, ideally allowing these different kinds of reasoning to be
meaningfully combined. To that end, considering further computational
effects such as IO and continuations would widen the applicability of our
findings. This would benefit from making the taxonomy and the
Galois connections of \autoref{sec:galois} effect-modular, which is not
straightforward: some of the maps and carriers involved lack a full monad or
arrow structure, so one would need transformers that, like the state
transformers, do not depend on the full structure.
Finally, this work also reveals unexpected skew monadic structures on monotone forward
predicate transformers $\SP$ that could be interesting to investigate on their own,
inspired by the seminal work of \citet{LackStreet12} on skew monoidales.

\ifanon\else
\begin{credits}
\subsubsection*{\ackname}
This work was in part funded
by the Deutsche Forschungsgemeinschaft (DFG, German Research Foundation)
under Germany's Excellence Strategy -- EXC 2092 CASA -- 390781972.
Exequiel Rivas was funded by the Estonian Research Council starting grant PSG749.
\end{credits}
\fi

\appendix
\crefalias{section}{appendix}
\renewcommand{\theHsection}{\Alph{section}}

\makeatletter
\let\originalseccntformat\@seccntformat
\renewcommand{\@seccntformat}[1]{%
  \ifcsname appendixprefix@#1\endcsname
    \csname appendixprefix@#1\endcsname
  \fi
  \originalseccntformat{#1}}
\newcommand{\appendixprefix@section}{\appendixname~}
\makeatother

\section{Arrow and Arrow Morphism Laws}\label{app:arrow-laws}

All arrow laws regarding arrows with application are shown below: 
  \[
  \begin{aligned}
    \arrcomp{\mathsf{id}}{f} & \equiv f \\
    \arrcomp{f}{\mathsf{id}} & \equiv f \\
    \arrcomp{\arrcomp{f}{g}}{h} & \equiv \arrcomp{f}{(\arrcomp{g}{h})} \\
    \arrcomp{\arr{f}}{\arr{g}} &\equiv \arr{(g \cdot f)} \\
    \first{(\arr{f})} & \equiv \arr{(\lambda (a, c).\, (f\,a, c))} \\
    \first{(\arrcomp{f}{g})} & \equiv \arrcomp{\first{f}}{\first{g}} \\
    \arrcomp{\first{f}}{\arr{\mathsf{fst}}} & \equiv \arrcomp{\arr{\mathsf{fst}}}{f} \\
    \arrcomp{\first{(\first{f})}}{\mathsf{assoc}} & \equiv \arrcomp{\mathsf{assoc}}{\first{f}} \\
    \arrcomp{\first{f}}{\arr{(\lambda (b,c).\,(b, g\,c))}} & \equiv \arrcomp{\arr{(\lambda (a,c).\,(a,g\,c))}}{\first{f}} \\
    \arrcomp{\arr{(\lambda a.\,(f, a))}}{\app} &\equiv f \\
    \arrcomp{\first{(\arr{(\lambda f.\,\arrcomp{g}{f})})}}{\app} &\equiv \arrcomp{\mathsf{second}\,g}{\app} \\
    \arrcomp{\first{(\arr{(\lambda f.\,\arrcomp{f}{h})})}}{\app} &\equiv  \arrcomp{\app}{h}
  \end{aligned}
\]
We use $\equiv$ to represent the equality defined on the arrow, which is defined
as logical equivalence on \emph{sps} arrows. All these laws are defined by
\citet{Hughes00}. We also state laws of arrow transformers defined by
\citeauthor{Hughes00} and prove that state and option arrow transformers satisfy
these laws.

In addition, we show that all our \emph{sps} arrows are choice arrows, and they
satisfy relevant laws defined by \citeauthor{Hughes00} and the additional choice
arrow distribution law defined by \citet{paterson2003arrows}.

An arrow morphism $\theta$ from an arrow $\mathsf{A}$ to an arrow
$\mathsf{B}$, both with choice and application, is a family
$\theta_{\mathsf{X},\mathsf{Y}} : \ArrowType{A}{X}{Y} \to \ArrowType{B}{X}{Y}$
that respects $\equiv$, so that $f \equiv g$ implies
$\theta(f) \equiv \theta(g)$, and satisfies the following laws:
\[
\begin{aligned}
  \theta(\arr[A]{h}) & \equiv \arr[B]{h} \\
  \theta(\arrcomp[A]{f}{g}) & \equiv \arrcomp[B]{\theta(f)}{\theta(g)} \\
  \theta(\first[A]{f}) & \equiv \first[B]{\theta(f)} \\
  \theta(\mathsf{left}_{\mathsf{A}}\,f) & \equiv \mathsf{left}_{\mathsf{B}}\,\theta(f) \\
  \theta(\app_{\mathsf{A}}) & \equiv \arrcomp[B]{\arr[B]{(\lambda (f, a).\,(\theta(f), a))}}{\app_{\mathsf{B}}}
\end{aligned}
\]
The first three are the laws of \citet[Def.~4.4]{Atkey2011}, and the last two
preserve choice and application. Since $\app_{\mathsf{A}}$ takes an arrow of
$\mathsf{A}$ as input and $\app_{\mathsf{B}}$ an arrow of $\mathsf{B}$, the
law for $\app$ first maps the input arrow with $\theta$. We use it, together
with the laws of $\mathsf{B}$, to show that an arrow morphism is, at each fixed
input, a monad morphism between the monads of the two arrows.

\section{The $\DSPArr$ and $\ASLPArr$ Arrows}\label{app:dsp-aslp-arrows}

For $\ArrowType{\DSPArr}{B}{A}$, we restrict the arrow carrier to conjunctive,
strict $\mathit{sps}_{\text{arr}}$, mirroring the carrier restriction of the
$\MonadType{\DSP}{A}$ monad. For every non-empty family of preconditions
$(p_i)_{i \in I}$ on the input $\mathsf{B}$:
\[
  \mathit{sp}_{\text{arr}}\Bigl(\bigwedge_{i \in I} p_i\Bigr) = \bigwedge_{i \in I} \mathit{sp}_{\text{arr}}(p_i)
  \qquad\text{and}\qquad
  \mathit{sp}_{\text{arr}}(\bot) = \bot.
\]
As for the monad, a $\mathit{dsp}$ arrow splits into an $\mathit{asp}$ and a
$\mathit{dslp}$ component, where $\mathit{sp}~(\lambda\_.\,\top)~a$ states
that $a$ is reachable:
\[
\begin{aligned}
  &\mathsf{aspOf} : \ArrowType{\DSPArr}{B}{A} \to \ArrowType{\ASPArr}{B}{A}\\
  &\mathsf{aspOf}~\mathit{sp} := \lambda \mathit{pre}~a.~\mathit{sp}~(\lambda\_.\,\top)~a \land
       \neg\,\mathit{sp}~(\lambda b.\,\neg\,\mathit{pre}~b)~a, \\
  &\mathsf{dslpOf} : \ArrowType{\DSPArr}{B}{A} \to \ArrowType{\DSLPArr}{B}{A}\\
  &\mathsf{dslpOf}~\mathit{sp} := \lambda \mathit{pre}~a.~\mathit{sp}~\mathit{pre}~a \lor
       \neg\,\mathit{sp}~(\lambda\_.\,\top)~a
\end{aligned}
\]
Every $\mathit{sp}$ admits the same outputs as
$\mathsf{aspOf}~\mathit{sp} \land \mathsf{dslpOf}~\mathit{sp}$. Each arrow
operation applies the $\ASPArr$ operation to the $\mathit{asp}$ components
and the $\DSLPArr$ operation to the $\mathit{dslp}$ components, and
conjoins the results:
\[
\begin{aligned}
  \arr{f} &: (\mathsf{B} \to \mathsf{A}) \to \ArrowType{\DSPArr}{B}{A} \\
    &:= \lambda \mathit{pre}~a.~\arr[\vASP]{f}~\mathit{pre}~a \land
        \arr[\vDSLP]{f}~\mathit{pre}~a, \\
  \arrcomp{f}{g} &: \ArrowType{\DSPArr}{B}{A} \to \ArrowType{\DSPArr}{A}{C} \to \ArrowType{\DSPArr}{B}{C} \\
    &:= \lambda \mathit{pre}~c.~
        (\arrcomp[\vASP]{\mathsf{aspOf}~f}{\mathsf{aspOf}~g})~\mathit{pre}~c \\
    &\phantom{{}:= \lambda \mathit{pre}~c.~}\land
        (\arrcomp[\vDSLP]{\mathsf{dslpOf}~f}{\mathsf{dslpOf}~g})~\mathit{pre}~c, \\
  \first{f} &: \ArrowType{\DSPArr}{B}{A} \to \ArrowType{\DSPArr}{B \times C}{A \times C} \\
    &:= \lambda \mathit{pre}~q.~
        \first[\vASP]{(\mathsf{aspOf}~f)}~\mathit{pre}~q \land
        \first[\vDSLP]{(\mathsf{dslpOf}~f)}~\mathit{pre}~q, \\
  \app &: \ArrowType{\DSPArr}{\ArrowType{\DSPArr}{B}{A} \times B}{A} \\
    &:= \lambda \mathit{pre}~a.~
        (\exists f~b.~\mathit{pre}~(f, b) \land \mathsf{aspOf}~f~(\lambda x.\,x = b)~a) \\
    &\phantom{{}:= \lambda \mathit{pre}~a.~}\land
        (\forall f~b.~\mathsf{aspOf}~f~(\lambda x.\,x = b)~a \Rightarrow \mathit{pre}~(f, b))
\end{aligned}
\]
In $\app$, $\mathsf{aspOf}~f~(\lambda x.\,x = b)~a$ states that $f$ can
produce $a$ from $b$. Composing the $\mathit{dsp}$ transformers directly,
$\lambda \mathit{pre}~c.~g~(f~\mathit{pre})~c$, would not work: it breaks the
law $\arrcomp{\arr{f}}{\arr{g}} \equiv \arr{(g \cdot f)}$ whenever $g$ maps
a value produced by $f$ and a value never produced by $f$ to the same
output, as in \autoref{fig:dsp-sequencing}.

Dually, for $\ArrowType{\ASLPArr}{B}{A}$, we restrict the arrow carrier to
disjunctive, terminating $\mathit{sps}_{\text{arr}}$, mirroring the carrier
restriction of the $\MonadType{\ASLP}{A}$ monad. For every non-empty family
of preconditions $(p_i)_{i \in I}$ on the input $\mathsf{B}$:
\[
  \mathit{sp}_{\text{arr}}\Bigl(\bigvee_{i \in I} p_i\Bigr) = \bigvee_{i \in I} \mathit{sp}_{\text{arr}}(p_i)
  \qquad\text{and}\qquad
  \mathit{sp}_{\text{arr}}(\top) = \top.
\]
Here $\neg\,\mathit{sp}~(\lambda\_.\,\bot)~a$ states that $a$ is reachable,
and the components are
\[
\begin{aligned}
  &\mathsf{aspOf} : \ArrowType{\ASLPArr}{B}{A} \to \ArrowType{\ASPArr}{B}{A}\\
  &\mathsf{aspOf}~\mathit{sp} := \lambda \mathit{pre}~a.~\mathit{sp}~\mathit{pre}~a \land
       \neg\,\mathit{sp}~(\lambda\_.\,\bot)~a, \\
  &\mathsf{dslpOf} : \ArrowType{\ASLPArr}{B}{A} \to \ArrowType{\DSLPArr}{B}{A}\\
  &\mathsf{dslpOf}~\mathit{sp} := \lambda \mathit{pre}~a.~\neg\,\mathit{sp}~(\lambda b.\,\neg\,\mathit{pre}~b)~a \lor
       \mathit{sp}~(\lambda\_.\,\bot)~a
\end{aligned}
\]
Every $\mathit{sp}$ admits the same outputs as
$\mathsf{aspOf}~\mathit{sp} \lor \mathsf{dslpOf}~\mathit{sp}$, and the arrow
operations are those of $\ArrowType{\DSPArr}{B}{A}$ with $\lor$ in place of
$\land$.

Both are lawful arrows with application: they satisfy all the laws of
\cref{app:arrow-laws}. Fixing $\mathsf{B} = 1$ gives the \ret and \bind
of $\MonadType{\DSP}{A}$ and $\MonadType{\ASLP}{A}$ from
\autoref{sec:spmonad}.

\section{Fixing the Input of the $\ASPArr$ and $\DSLPArr$ Arrows}\label{app:fixing-input}

This appendix shows that fixing $\mathsf{B} = 1$ in
$\ArrowType{\ASPArr}{B}{A}$ and $\ArrowType{\DSLPArr}{B}{A}$ gives the
carriers and the \ret and \bind of
$\MonadType{\ASP}{A}$~\eqref{eqn:angelic-sp} and
$\MonadType{\DSLP}{A}$~\eqref{eqn:demonic-sp}.

A precondition $\mathit{pre} : 1 \to \mathbb{P}$ is determined by the
proposition $\mathit{pre}~\ast$, where $\ast$ is the unique element of $1$,
and a proposition $p$ is the precondition $\lambda\_.\,p$. We identify the
two, $(1 \to \mathbb{P}) \cong \mathbb{P}$, and write $\mathit{pre}$ for
both. The identification commutes with pointwise disjunctions and
conjunctions, so an
$\mathit{sp}_{\text{arr}} : (1 \to \mathbb{P}) \to (\mathsf{A} \to \mathbb{P})$
preserves all disjunctions, respectively all conjunctions, exactly when the
corresponding $\mathit{sp} : \mathbb{P} \to (\mathsf{A} \to \mathbb{P})$
does. Hence the carrier of $\ArrowType{\ASPArr}{1}{A}$ corresponds to that of
$\MonadType{\ASP}{A}$, and the carrier of $\ArrowType{\DSLPArr}{1}{A}$ to
that of $\MonadType{\DSLP}{A}$, up to logical equivalence, which is the
equality of the arrows. All these $\mathit{sps}$ are monotone, so they
respect logical equivalence of preconditions, which we use without mention.

Recall the \ret and \bind derived from an arrow with application
(\autoref{sec:base-sp-arrows}), where $h := \lambda (a, i).\,(f\,a, i)$:
\[
\begin{aligned}
  \mathsf{ret}\,v &:= \arr{(\lambda \_.\,v)},\\
  \mathsf{bind}\,m\,f &:=
    \arrcomp{(\arrcomp{\arr{(\lambda x.\,(x, x))}}{\first{m}})}{(\arrcomp{\arr{h}}{\app})}.
\end{aligned}
\]
Since $\arrcomp{f}{g}$ passes the postcondition of $f$ to $g$ as its
precondition, $\mathsf{bind}\,m\,f~\mathit{pre} = \app~P_3$, where
\[
  P_1 := \arr{(\lambda x.\,(x, x))}~\mathit{pre}, \qquad
  P_2 := \first{m}~P_1, \qquad
  P_3 := \arr{h}~P_2.
\]
Below, $u$, $i$, $j$, $x$, and $y$ range over $1$, $a$ over $\mathsf{A}$,
$b$ over $\mathsf{B}$, and $g$ over the arrows from $1$ to $\mathsf{B}$.

\subsection{The Angelic Arrow}\label{app:fixing-input-asp}

\begin{lemma}\label{lem:asp-singleton}
For every $\mathit{sp} : \MonadType{\ASP}{A}$, $\mathit{pre} : \mathbb{P}$,
and $a : \mathsf{A}$,
$\mathit{sp}~\mathit{pre}~a \iff \mathit{pre} \land \mathit{sp}~\top~a$.
\end{lemma}
\begin{proof}
Take the index set $I := \{\ast \mid \mathit{pre}\}$, which is inhabited
exactly when $\mathit{pre}$ holds, and the constant family $p(i) := \top$.
Then $\exists i.\,p(i)$ is equivalent to $\mathit{pre}$, so disjunctivity
gives $\mathit{sp}~\mathit{pre}~a \iff \exists i \in I.\,\mathit{sp}~\top~a$,
which is $\mathit{pre} \land \mathit{sp}~\top~a$.
\end{proof}

So an element of $\MonadType{\ASP}{A}$ admits $a$ exactly when $a$ is
produced, $\mathit{sp}~\top~a$, and $\mathit{pre}$ holds. For \ret,
unfolding $\arr{}$ gives
\[
  \mathsf{ret}\,v~\mathit{pre}~a \iff \exists x.\,\mathit{pre} \land v = a
  \iff \mathit{pre} \land a = v,
\]
the \ret of~\eqref{eqn:angelic-sp}. For \bind, unfolding $\arr{}$,
$\first{}$, $\arr{}$, and $\app$ in turn gives
\[
\begin{aligned}
  P_1~(u, i) &\iff \exists x.\,\mathit{pre} \land (x, x) = (u, i)
    \iff \mathit{pre},\\
  P_2~(a, i) &\iff m~(\lambda x.\,P_1~(x, i))~a \iff m~\mathit{pre}~a,\\
  P_3~(g, i) &\iff \exists (a, j).\,P_2~(a, j) \land (f\,a, j) = (g, i)
    \iff \exists a.\,m~\mathit{pre}~a \land f\,a = g,\\
  \app~P_3~b &\iff \exists g.\,\exists i.\,P_3~(g, i) \land g~(\lambda y.\,y = i)~b
    \iff \exists a.\,m~\mathit{pre}~a \land f\,a~\top~b.
\end{aligned}
\]
In the last line, $\lambda y.\,y = i$ is $\top$ because $1$ has a single
element, so at $\mathsf{B} = 1$, $\app$ asks whether $g$ produces $b$. By
\autoref{lem:asp-singleton} for $f\,a$, the conjunction
$m~\mathit{pre}~a \land f\,a~\top~b$ is $f\,a~(m~\mathit{pre}~a)~b$, so
\[
  \mathsf{bind}\,m\,f~\mathit{pre}~b \iff \exists a.\,f\,a~(m~\mathit{pre}~a)~b,
\]
the \bind of~\eqref{eqn:angelic-sp}. Every step is constructive.

\subsection{The Demonic Arrow}\label{app:fixing-input-dslp}

\begin{lemma}\label{lem:dslp-singleton}
For every $\mathit{sp} : \MonadType{\DSLP}{A}$, $\mathit{pre} : \mathbb{P}$,
and $a : \mathsf{A}$,
$\mathit{sp}~\mathit{pre}~a \iff \mathit{pre} \lor \mathit{sp}~\bot~a$.
\end{lemma}
\begin{proof}
Take the index set $J := \{\ast \mid \neg\mathit{pre}\}$ and the constant
family $p(j) := \bot$. Then $\forall j.\,p(j)$ is $\neg\neg\mathit{pre}$,
which by excluded middle is equivalent to $\mathit{pre}$, so conjunctivity
gives $\mathit{sp}~\mathit{pre}~a \iff \forall j \in J.\,\mathit{sp}~\bot~a$,
which is $\neg\mathit{pre} \Rightarrow \mathit{sp}~\bot~a$ and, again by
excluded middle, $\mathit{pre} \lor \mathit{sp}~\bot~a$.
\end{proof}

So an element of $\MonadType{\DSLP}{A}$ admits $a$ exactly when
$\mathit{pre}$ holds or $a$ is not produced, $\mathit{sp}~\bot~a$. For
\ret,
\[
  \mathsf{ret}\,v~\mathit{pre}~a \iff \forall x.\,v = a \Rightarrow \mathit{pre}
  \iff a = v \Rightarrow \mathit{pre},
\]
the \ret of~\eqref{eqn:demonic-sp}. For \bind, the same unfolding
gives
\[
\begin{aligned}
  P_1~(u, i) &\iff \forall x.\,(x, x) = (u, i) \Rightarrow \mathit{pre}
    \iff \mathit{pre},\\
  P_2~(a, i) &\iff m~(\lambda x.\,P_1~(x, i))~a \iff m~\mathit{pre}~a,\\
  P_3~(g, i) &\iff \forall (a, j).\,(f\,a, j) = (g, i) \Rightarrow P_2~(a, j)
    \iff \forall a.\,f\,a = g \Rightarrow m~\mathit{pre}~a,\\
  \app~P_3~b &\iff \forall g.\,\forall i.\,(\forall p.\,g~p~b \Rightarrow p~i)
    \Rightarrow P_3~(g, i)\\
  &\iff \forall a.\,(\forall p.\,f\,a~p~b \Rightarrow p) \Rightarrow m~\mathit{pre}~a.
\end{aligned}
\]
In the last two lines, $p$ ranges over $1 \to \mathbb{P} \cong \mathbb{P}$,
so $p~i$ is $p$, and at $\mathsf{B} = 1$, $\app$ again asks whether $g$
produces $b$: every precondition that $g$ accepts for $b$ holds. By
\autoref{lem:dslp-singleton} for $f\,a$, this is $\neg(f\,a~\bot~b)$: taking
$p := \bot$ gives one direction, and for the other, $f\,a~p~b$ gives
$p \lor f\,a~\bot~b$. Then, by excluded middle and
\autoref{lem:dslp-singleton} again,
\[
\begin{aligned}
  \mathsf{bind}\,m\,f~\mathit{pre}~b
  &\iff \forall a.\,\neg(f\,a~\bot~b) \Rightarrow m~\mathit{pre}~a\\
  &\iff \forall a.\,m~\mathit{pre}~a \lor f\,a~\bot~b\\
  &\iff \forall a.\,f\,a~(m~\mathit{pre}~a)~b,
\end{aligned}
\]
the \bind of~\eqref{eqn:demonic-sp}. Unlike the angelic derivation, this
one is classical.

\section{Stateful $\mathit{sp}$ Arrows and Monads}\label{app:stateful-sp}

This appendix gives the stateful arrows of \autoref{sec:stateful-sp} and the
\ret and \bind of $\MonadType{\DSPSt}{S, A}$ and $\MonadType{\ASLPSt}{S, A}$.

\subsection{The Stateful Arrows}\label{app:stateful-arrows}

The state arrow transformer~\citep{Hughes00} threads a state through the
input and the output of an arrow,
$\StTArrType{Arr}{S}{B}{A} := \ArrowType{Arr}{B \times S}{A \times S}$, with
the operations
\[
\begin{aligned}
  \arr[St]{f} &:= \arr{(\lambda (b, s).\,(f\,b, s))},\\
  \arrcomp[St]{f}{g} &:= \arrcomp{f}{g},\\
  \first[St]{f} &:= \arrcomp{\arrcomp{\arr{(\lambda ((b, c), s).\,((b, s), c))}}{\first{f}}}{\arr{(\lambda ((a, s), c).\,((a, c), s))}},\\
  \app_{\mathsf{St}} &:= \arrcomp{\arr{(\lambda ((f, b), s).\,(f, (b, s)))}}{\app}.
\end{aligned}
\]
The state passes through $\arr{}$ unchanged, $\first{}$ moves it next to the
input of $f$ and back, and $\app$ hands it to the applied arrow together with
its input. For $\StTArrType{\ASPArr}{S}{B}{A}$, unfolding the operations of
\autoref{sec:base-sp-arrows} gives
\[
\begin{aligned}
  \arr[St]{f} &= \lambda \mathit{pre}~(a, s).~\exists b.~\mathit{pre}~(b, s) \land f~b = a,\\
  \first[St]{f} &= \lambda \mathit{pre}~((a, c), s).~f~(\lambda (b, s').~\mathit{pre}~((b, c), s'))~(a, s),\\
  \app_{\mathsf{St}} &= \lambda \mathit{pre}~(a, s).~\exists f~b~s'.~\mathit{pre}~((f, b), s') \land f~(\lambda x.\,x = (b, s'))~(a, s),
\end{aligned}
\]
and for $\StTArrType{\DSLPArr}{S}{B}{A}$
\[
\begin{aligned}
  \arr[St]{f} &= \lambda \mathit{pre}~(a, s).~\forall b.~f~b = a \Rightarrow \mathit{pre}~(b, s),\\
  \app_{\mathsf{St}} &= \lambda \mathit{pre}~(a, s).~\forall f~b~s'.~(\forall p.~f~p~(a, s) \Rightarrow p~(b, s')) \Rightarrow \mathit{pre}~((f, b), s').
\end{aligned}
\]
In both stateful arrows, $\arrcomp[St]{}{}$ uses the corresponding base-arrow
composition. Their definitions of $\first[St]{}$ coincide because the base
definitions of $\first{}$ coincide and the stateful definitions apply
$\arr{}$ only to bijections. For such bijections $h$, both arrows give
$\arr{h} = \lambda \mathit{pre}~y.~\mathit{pre}~(h^{-1}\,y)$. Applying the
transformer to the operations of $\ArrowType{\DSPArr}{B}{A}$ and
$\ArrowType{\ASLPArr}{B}{A}$ (\cref{app:dsp-aslp-arrows}) gives their
stateful arrows in the same way.

\subsection{The Stateful $\mathit{dsp}$ and $\mathit{aslp}$ Monads}\label{app:stateful-dsp-aslp}

Fixing $\mathsf{B} = 1$ gives the monads of \autoref{sec:stateful-sp}, whose
elements are $\mathit{dsp}$, respectively $\mathit{aslp}$, arrows from
$1 \times \mathsf{S}$ to $\mathsf{A} \times \mathsf{S}$, with preconditions
over states, $(1 \times \mathsf{S} \to \mathbb{P}) \cong (\mathsf{S} \to \mathbb{P})$.
Their $\mathit{asp}$ and $\mathit{dslp}$ components are therefore the
$\mathsf{aspOf}$ and $\mathsf{dslpOf}$ of \cref{app:dsp-aslp-arrows} at
input $1 \times \mathsf{S}$, and, as for the base monads, \ret and \bind are
computed through them. For $\MonadType{\DSPSt}{S, A}$,
\[
\begin{aligned}
\ret_{\vDSP}\,v
  &\equiv \ret_{\vASP}~v \land \ret_{\vDSLP}~v, \\
\bind_{\vDSP}\,m\,f
  &\equiv \lambda~\mathit{pre}~r.\;
     \bind_{\vASP}~(\mathsf{aspOf}~m)~(\lambda a.~\mathsf{aspOf}~(f~a))~\mathit{pre}~r~\land \\
  &\phantom{{}\equiv \lambda~\mathit{pre}~r.\;}
     \bind_{\vDSLP}~(\mathsf{dslpOf}~m)~(\lambda a.~\mathsf{dslpOf}~(f~a))~\mathit{pre}~r,
\end{aligned}
\]
where $\ret_{\vASP}$, $\bind_{\vASP}$, $\ret_{\vDSLP}$, and $\bind_{\vDSLP}$
are those of~\eqref{eqn:angelic-sp-stateful}
and~\eqref{eqn:demonic-sp-stateful}. For $\MonadType{\ASLPSt}{S, A}$, \ret and
\bind are the same with $\lor$ in place of $\land$. Since $\ret_{\vASP}~v$
implies $\ret_{\vDSLP}~v$, these simplify to
$\ret_{\vDSP}~v \equiv \ret_{\vASP}~v$ and
$\ret_{\vASLP}~v \equiv \ret_{\vDSLP}~v$.

\section{Counterexamples to the Monad and Arrow Laws}\label{app:law-counterexamples}

This appendix works through the counterexamples behind the carrier
restrictions of \autoref{sec:spmonad}. Each keeps the operations of one of
our monads or arrows but drops its homomorphism restriction:
\begin{itemize}
  \item[$\triangleright$] \cref{app:assoc-angelic}: the stateful \ret
    and \bind of $\MonadType{\ASP}{S, A}$ on monotone, strict
    $\mathit{sps}$, without disjunctivity;
  \item[$\triangleright$] \cref{app:assoc-demonic}: the stateful \ret
    and \bind of $\MonadType{\DSLP}{S, A}$ on monotone, terminating
    $\mathit{sps}$, without conjunctivity;
  \item[$\triangleright$] \cref{app:first-fst}: the operations of
    $\ArrowType{\ASPArr}{B}{A}$ and $\ArrowType{\DSLPArr}{B}{A}$ on monotone,
    strict, respectively terminating, $\mathit{sps}_{\text{arr}}$.
\end{itemize}
The two monads fail associativity and the two arrows the arrow law
$\arrcomp{\first{f}}{\arr{\mathsf{fst}}} \equiv \arrcomp{\arr{\mathsf{fst}}}{f}$.
Each
counterexample uses a transformer that inspects two inputs at once, which
disjunctivity, respectively conjunctivity, rules out.

\subsection{Associativity Without Disjunctivity}\label{app:assoc-angelic}

Consider $\MonadType{\SPostAng}{S, A}$, whose elements are monotone stateful
$\mathit{sps}$ of type
$(\text{pre}: \mathsf{S} \to \mathbb{P}) \to (\mathsf{A} \times \mathsf{S} \to \mathbb{P}_{\text{/pre}})$.
Here $\mathbb{P}_{\text{/pre}} = \{p : \mathbb{P} ~|~ p \Rightarrow \exists s.~\text{pre}~s\}$
is the codomain condition of $\MonadType{\SPost}{A}$; for monotone
$\mathit{sps}$, it amounts to strictness.
We use the \ret and \bind of $\MonadType{\ASP}{S, A}$
from~\eqref{eqn:angelic-sp-stateful}.
The unit laws hold, but \bind is not associative as soon as there are
two distinct states $s_1, s_2$ and two distinct values $a_1, a_2$. Take,
for any inhabited types $\mathsf{B}$ and $\mathsf{C}$,
\[
\begin{aligned}
m &: \MonadType{\SPostAng}{S, A}
  := \lambda~\mathit{pre}~\_.\,\exists s.~\mathit{pre}~s,\\
f &: \mathsf{A} \to \MonadType{\SPostAng}{S, B}
  := \lambda a.\,\lambda~\mathit{pre}~(b,s).\,(\exists s'.\,\mathit{pre}~s')\\
  &\phantom{{}: \mathsf{A} \to \MonadType{\SPostAng}{S, B} := {}}
     \land ((a = a_1 \land s = s_1) \lor (a = a_2 \land s = s_2)),\\
g &: \MonadType{\SPostAng}{S, C}
  := \lambda~\mathit{pre}~\_.\,\mathit{pre}~s_1 \land \mathit{pre}~s_2.
\end{aligned}
\]
When its precondition is satisfiable, $m$ admits every result; $f~a$ ends
in $s_1$ when $a = a_1$, in $s_2$ when $a = a_2$, and has no result
otherwise; and $g$ has a result only if its precondition contains {\em
both} $s_1$ and $s_2$. All three are monotone and have the restricted type.
For $\mathit{pre} := \lambda \_.\,\top$ and $r := (c_0, s_1)$, compare
\[
  \mathsf{bind}~(\mathsf{bind}~m~f)~(\lambda \_.\,g)~\mathit{pre}~r
  \quad\text{and}\quad
  \mathsf{bind}~m~(\lambda a.\,\mathsf{bind}~(f~a)~(\lambda \_.\,g))~\mathit{pre}~r.
\]
Since $m~\mathit{pre}~(a, s'')$ holds for all $a$ and $s''$, the
intermediate states that $f~a$ passes on to $g$ are
$P_a := \lambda s'.\,(a = a_1 \land s' = s_1) \lor (a = a_2 \land s' = s_2)$,
so $P_{a_1} = \{s_1\}$, $P_{a_2} = \{s_2\}$, and $P_a = \emptyset$ otherwise.
Unfolding \bind, the left-hand side is equivalent to
$\exists b.~g~(\lambda s'.\,\exists a.\,P_a~s')~r$ and the right-hand side to
$\exists a.\,\exists b.~g~(\lambda s'.\,P_a~s')~r$. On the left, $g$ receives
the union $\bigcup_a P_a = \{s_1, s_2\}$ and admits $r$. On the right, $a$ is
chosen first and $g$ receives a single $P_a$, none of which contains both
states, so it does not. A disjunctive $g$ would satisfy
$g(\bigcup_a P_a) = \bigcup_a g(P_a)$ and make the two sides agree.

\subsection{Associativity Without Conjunctivity}\label{app:assoc-demonic}

The demonic monad $\MonadType{\SPostDem}{S, A}$ has the dual codomain,
which amounts to termination,
$(\text{pre}: \mathsf{S} \to \mathbb{P}) \to (\mathsf{A} \times \mathsf{S} \to \mathbb{P}^{\forall}_{\text{/pre}})$
with $\mathbb{P}^{\forall}_{\text{/pre}} = \{p : \mathbb{P} ~|~ (\forall s.~\text{pre}~s) \Rightarrow p\}$,
and the \ret and \bind~\eqref{eqn:demonic-sp-stateful} of
$\MonadType{\DSLP}{S, A}$.
Again the unit laws hold and associativity fails, now with the dual
transformers
\[
\begin{aligned}
m &:= \lambda~\mathit{pre}~\_.\,\forall s.~\mathit{pre}~s,\\
f &:= \lambda a.\,\lambda~\mathit{pre}~(b,s).\,
     ((a = a_1 \land s = s_1) \lor (a = a_2 \land s = s_2)) \Rightarrow (\forall s'.\,\mathit{pre}~s'),\\
g &:= \lambda~\mathit{pre}~\_.\,\mathit{pre}~s_1 \lor \mathit{pre}~s_2,
\end{aligned}
\]
where $g$ accepts a precondition that contains {\em at least one} of $s_1$
and $s_2$. Take $\mathit{pre} := \lambda \_.\,\bot$ and $r := (c_0, s_1)$.
Since $m~\mathit{pre}~(a, s'')$ is false, the intermediate states passed on
to $g$ are
$P_a := \lambda s'.\,\neg (a = a_1 \land s' = s_1) \land \neg (a = a_2 \land s' = s_2)$,
so $P_{a_1} = \mathsf{S} \setminus \{s_1\}$, $P_{a_2} = \mathsf{S} \setminus \{s_2\}$,
and $P_a = \mathsf{S}$ otherwise. The left-hand side is equivalent to
$\forall b.~g~(\lambda s'.\,\forall a.\,P_a~s')~r$, which fails because
$\bigcap_a P_a = \mathsf{S} \setminus \{s_1, s_2\}$ contains neither state.
The right-hand side is equivalent to
$\forall a.\,\forall b.~g~(\lambda s'.\,P_a~s')~r$, which holds because
every $P_a$ contains $s_1$ or $s_2$. A conjunctive $g$ would make the two
sides agree.

\subsection{Arrow Laws Without Disjunctivity or Conjunctivity}\label{app:first-fst}

The same phenomenon breaks the arrow laws. Let
$\ArrowType{\SPArr_{\mathsf{ang}}}{B}{A}$ consist of the monotone, strict
$\mathit{sps}_{\text{arr}}$ of type
$(\text{pre}: \mathsf{B} \to \mathbb{P}) \to (\mathsf{A} \to \mathbb{P}_{\text{/pre}})$,
where $\mathbb{P}_{\text{/pre}} = \{p : \mathbb{P} ~|~ p \Rightarrow \exists b.~\text{pre}~b\}$,
with the operations of $\ArrowType{\ASPArr}{B}{A}$. Fixing
$\mathsf{B} = 1$ gives the lawful monad $\MonadType{\SPostAng}{A}$, yet the
arrow is not lawful: the law
$\arrcomp{\first{f}}{\arr{\mathsf{fst}}} \equiv \arrcomp{\arr{\mathsf{fst}}}{f}$
fails. Take distinct $b_1, b_2 : \mathsf{B}$ and $c_1, c_2 : \mathsf{C}$, and
\[
\begin{aligned}
  g &: \ArrowType{\SPArr_{\mathsf{ang}}}{B}{A}
    := \lambda~\mathit{pre}~\_.\,\mathit{pre}~b_1 \land \mathit{pre}~b_2,\\
  \mathit{pre} &:= \lambda (b, c).\,(b = b_1 \land c = c_1) \lor (b = b_2 \land c = c_2).
\end{aligned}
\]
Writing $P_c := \lambda b.\,\mathit{pre}~(b, c)$, we have
$P_{c_1} = \{b_1\}$, $P_{c_2} = \{b_2\}$, and $P_c = \emptyset$ otherwise.
Unfolding the operations, for any $a : \mathsf{A}$,
\[
\begin{aligned}
  (\arrcomp{\first{g}}{\arr{\mathsf{fst}}})~\mathit{pre}~a &\iff \exists c.~g~P_c~a,\\
  (\arrcomp{\arr{\mathsf{fst}}}{g})~\mathit{pre}~a &\iff g~\Bigl(\bigcup_c P_c\Bigr)~a.
\end{aligned}
\]
The left-hand side is false, since no $P_c$ contains both $b_1$ and $b_2$,
and the right-hand side is true, since $\bigcup_c P_c = \{b_1, b_2\}$. A
disjunctive $g$ would make the two sides agree, which is why
$\ArrowType{\ASPArr}{B}{A}$ restricts its carrier to disjunctive
$\mathit{sps}_{\text{arr}}$. Dually, for the demonic arrow
$\ArrowType{\SPArr_{\mathsf{dem}}}{B}{A}$ of monotone, terminating
$\mathit{sps}_{\text{arr}}$ with the operations of
$\ArrowType{\DSLPArr}{B}{A}$, the same law fails for
$g := \lambda~\mathit{pre}~\_.\,\mathit{pre}~b_1 \lor \mathit{pre}~b_2$ and
the complementary precondition
$\mathit{pre} := \lambda (b, c).\,\neg((b = b_1 \land c = c_1) \lor (b = b_2 \land c = c_2))$,
now with the left-hand side true and the right-hand side false.

\section{Decomposing $\mathit{dsp}$ and $\mathit{aslp}$ Predicate Transformers}\label{app:dsp-decomposition}

This appendix proves that an element of $\MonadType{\DSP}{A}$ admits the same
outputs as the conjunction of its $\mathit{asp}$ and $\mathit{dslp}$
components, that an element of $\MonadType{\ASLP}{A}$ admits the same outputs
as their disjunction, and that in both cases the components lie in
$\MonadType{\ASP}{A}$ and $\MonadType{\DSLP}{A}$.

\subsection{Components of $\MonadType{\DSP}{A}$}\label{app:dsp-components}

Recall the components of an element $\mathit{sp}$ of $\MonadType{\DSP}{A}$
from \autoref{sec:base-sp-monads},
\[
\begin{aligned}
\mathsf{aspOf}~\mathit{sp} &:= \lambda~\mathit{pre}~a.\,\mathit{sp}~\top~a \land \neg(\mathit{sp}~(\neg \mathit{pre})~a),\\
\mathsf{dslpOf}~\mathit{sp} &:= \lambda~\mathit{pre}~a.\,\mathit{sp}~\mathit{pre}~a \lor \neg(\mathit{sp}~\top~a),
\end{aligned}
\]
and that the carrier of $\MonadType{\DSP}{A}$ consists of the strict,
conjunctive, and hence monotone, $\mathit{sps}$.

\begin{lemma}\label{lem:dsp-split}
For every $\mathit{sp} : \MonadType{\DSP}{A}$, $\mathit{pre} : \mathbb{P}$,
and $a : \mathsf{A}$,
\[
  \mathit{sp}~\mathit{pre}~a \iff
  \mathsf{aspOf}~\mathit{sp}~\mathit{pre}~a \land \mathsf{dslpOf}~\mathit{sp}~\mathit{pre}~a.
\]
\end{lemma}
\begin{proof}
Suppose $\mathit{sp}~\mathit{pre}~a$. Then $\mathsf{dslpOf}~\mathit{sp}~\mathit{pre}~a$
holds by its first disjunct, and $\mathit{sp}~\top~a$ holds by monotonicity.
If $\mathit{sp}~(\neg\mathit{pre})~a$ also held, conjunctivity would give
$\mathit{sp}~(\mathit{pre} \land \neg\mathit{pre})~a$, and monotonicity
$\mathit{sp}~\bot~a$, which strictness rules out. So
$\neg(\mathit{sp}~(\neg\mathit{pre})~a)$, and
$\mathsf{aspOf}~\mathit{sp}~\mathit{pre}~a$ holds. Conversely,
$\mathsf{aspOf}~\mathit{sp}~\mathit{pre}~a$ gives $\mathit{sp}~\top~a$, which rules
out the second disjunct of $\mathsf{dslpOf}~\mathit{sp}~\mathit{pre}~a$, so its first
disjunct, $\mathit{sp}~\mathit{pre}~a$, holds.
\end{proof}

This proof is constructive. Classical reasoning is needed only to show that
the components satisfy the carrier restrictions of $\MonadType{\ASP}{A}$ and
$\MonadType{\DSLP}{A}$.

\begin{lemma}\label{lem:dsp-components}
For every $\mathit{sp} : \MonadType{\DSP}{A}$, $\mathsf{aspOf}~\mathit{sp}$ is strict
and disjunctive, and $\mathsf{dslpOf}~\mathit{sp}$ is terminating and conjunctive.
\end{lemma}
\begin{proof}
Both components are monotone because $\mathit{sp}$ is and negation reverses
implication. $\mathsf{aspOf}~\mathit{sp}~\bot~a$ is
$\mathit{sp}~\top~a \land \neg(\mathit{sp}~\top~a)$, which is false, so
$\mathsf{aspOf}~\mathit{sp}$ is strict. For a non-empty family $(p_i)_{i \in I}$,
$\neg \exists i.\,p_i$ is $\forall i.\,\neg p_i$, and conjunctivity of
$\mathit{sp}$ turns $\neg(\mathit{sp}~(\neg \exists i.\,p_i)~a)$ into
$\neg \forall i.\,\mathit{sp}~(\neg p_i)~a$. By excluded middle, this is
$\exists i.\,\neg(\mathit{sp}~(\neg p_i)~a)$, which makes
$\mathsf{aspOf}~\mathit{sp}$ disjunctive. $\mathsf{dslpOf}~\mathit{sp}~\top~a$ is
$\mathit{sp}~\top~a \lor \neg(\mathit{sp}~\top~a)$, which holds by excluded
middle, so $\mathsf{dslpOf}~\mathit{sp}$ is terminating. Its conjunctivity follows
from that of $\mathit{sp}$ by a case split, again by excluded middle, on
whether $\mathit{sp}~\top~a$ holds.
\end{proof}

\subsection{Components of $\MonadType{\ASLP}{A}$}\label{app:aslp-components}

The carrier of $\MonadType{\ASLP}{A}$ consists of the terminating,
disjunctive, and hence monotone, $\mathit{sps}$. For an element
$\mathit{sp}$, the output $a$ is admitted under the $\bot$ precondition,
$\mathit{sp}~\bot~a$, exactly when $a$ is not produced, and we calculate
the two components as follows:
\[
\begin{aligned}
&\mathsf{aspOf} : \MonadType{\ASLP}{A} \to \MonadType{\ASP}{A}\\
&\mathsf{aspOf}~\mathit{sp}
  := \lambda~\mathit{pre}~a.\, \mathit{sp}~\mathit{pre}~a \land
  \neg(\mathit{sp}~\bot~a)\\[1ex]
&\mathsf{dslpOf} : \MonadType{\ASLP}{A} \to \MonadType{\DSLP}{A}\\
&\mathsf{dslpOf}~\mathit{sp}
  := \lambda~\mathit{pre}~a.\, \neg(\mathit{sp}~(\neg \mathit{pre})~a) \lor
  \mathit{sp}~\bot~a
\end{aligned}
\]
In $\mathsf{aspOf}~\mathit{sp}$, the second conjunct, $\neg(\mathit{sp}~\bot~a)$,
states that $a$ is produced, and the first, $\mathit{sp}~\mathit{pre}~a$,
states that $a$ being produced implies $\mathit{pre}$ holding; together they
admit $a$ exactly when $a$ is produced and $\mathit{pre}$ holds, which is
the behavior of an $\mathit{asp}$ predicate transformer.
In $\mathsf{dslpOf}~\mathit{sp}$, the first disjunct,
$\neg(\mathit{sp}~(\neg \mathit{pre})~a)$, negates the statement that if
$a$ is produced, then $\neg\mathit{pre}$ holds. Classically, this means
that $a$ is produced and $\mathit{pre}$ holds. The second disjunct,
$\mathit{sp}~\bot~a$, states that $a$ is not produced. Together, the two
disjuncts admit $a$ exactly when the implication ``if $a$ is produced,
then $\mathit{pre}$ holds'' is true, as required by a $\mathit{dslp}$
predicate transformer. As for $\MonadType{\DSP}{A}$, these readings are classical.
The \bind of $\MonadType{\ASLP}{A}$ is the disjunction of the
\bind{}s of the components:
\[
\begin{aligned}
\mathsf{bind_{\vASLP}}\,m\,f
  &:= \lambda~\mathit{pre}~b.\,
     \mathsf{bind_{\vASP}~(aspOf~m)}~(\lambda a. \mathsf{~aspOf~(f~}a))~\mathit{pre}~b~\lor \\
  &\phantom{{}:= \lambda~\mathit{pre}~b.\,}\mathsf{bind_{\vDSLP}~(dslpOf~m)}~(\lambda a. \mathsf{~dslpOf~(f~}a))~\mathit{pre}~b.
\end{aligned}
\]

\begin{lemma}\label{lem:aslp-split}
For every $\mathit{sp} : \MonadType{\ASLP}{A}$, $\mathit{pre} : \mathbb{P}$,
and $a : \mathsf{A}$,
\[
  \mathit{sp}~\mathit{pre}~a \iff
  \mathsf{aspOf}~\mathit{sp}~\mathit{pre}~a \lor \mathsf{dslpOf}~\mathit{sp}~\mathit{pre}~a.
\]
\end{lemma}
\begin{proof}
Suppose $\mathit{sp}~\mathit{pre}~a$. By excluded middle, either
$\mathit{sp}~\bot~a$ holds, and then so does the second disjunct of
$\mathsf{dslpOf}~\mathit{sp}~\mathit{pre}~a$, or it does not, and then
$\mathsf{aspOf}~\mathit{sp}~\mathit{pre}~a$ holds. Conversely,
$\mathsf{aspOf}~\mathit{sp}~\mathit{pre}~a$ contains $\mathit{sp}~\mathit{pre}~a$ as
its first conjunct, and $\mathit{sp}~\bot~a$ gives
$\mathit{sp}~\mathit{pre}~a$ by monotonicity. It remains to derive
$\mathit{sp}~\mathit{pre}~a$ from $\neg(\mathit{sp}~(\neg \mathit{pre})~a)$.
Termination gives $\mathit{sp}~\top~a$, and by excluded middle $\top$ is
$\mathit{pre} \lor \neg\mathit{pre}$, so disjunctivity gives
$\mathit{sp}~\mathit{pre}~a \lor \mathit{sp}~(\neg\mathit{pre})~a$, whose
second disjunct is ruled out.
\end{proof}

Unlike the decomposition of $\MonadType{\DSP}{A}$, this one is classical in
both directions.

\begin{lemma}\label{lem:aslp-components}
For every $\mathit{sp} : \MonadType{\ASLP}{A}$, $\mathsf{aspOf}~\mathit{sp}$ is strict
and disjunctive, and $\mathsf{dslpOf}~\mathit{sp}$ is terminating and conjunctive.
\end{lemma}
\begin{proof}
Both components are monotone because $\mathit{sp}$ is and negation reverses
implication. $\mathsf{aspOf}~\mathit{sp}~\bot~a$ is
$\mathit{sp}~\bot~a \land \neg(\mathit{sp}~\bot~a)$, which is false, so
$\mathsf{aspOf}~\mathit{sp}$ is strict, and it is disjunctive because $\mathit{sp}$
is and $\neg(\mathit{sp}~\bot~a)$ does not depend on the precondition; both
arguments are constructive. $\mathsf{dslpOf}~\mathit{sp}~\top~a$ is
$\neg(\mathit{sp}~\bot~a) \lor \mathit{sp}~\bot~a$, which holds by excluded
middle, so $\mathsf{dslpOf}~\mathit{sp}$ is terminating. For a non-empty family
$(p_i)_{i \in I}$, excluded middle turns $\neg \forall i.\,p_i$ into
$\exists i.\,\neg p_i$, so disjunctivity of $\mathit{sp}$ turns
$\neg(\mathit{sp}~(\neg \forall i.\,p_i)~a)$ into
$\forall i.\,\neg(\mathit{sp}~(\neg p_i)~a)$, and a case split, again by
excluded middle, on whether $\mathit{sp}~\bot~a$ holds makes
$\mathsf{dslpOf}~\mathit{sp}$ conjunctive.
\end{proof}

The same definitions and proofs apply to the arrow components of
\cref{app:dsp-aslp-arrows}, with $\top$, $\bot$, and $\neg$ taken
pointwise.

\section{Monadic Relations Between the Base Monads}\label{app:base-relations}

This appendix gives the monadic relations of the left diagram of
\autoref{fig:monad-taxonomy}. They are the relations of
\autoref{sec:implication-relation} with propositions in place of
predicates on states: the base $\mathit{wp}$ monads are submonads of
$\MonadType{\WPure}{A} = (\mathsf{A} \to \mathbb{P}) \to \mathbb{P}$, and the
elements of the base $\mathit{sp}$ monads have type
$\mathbb{P} \to (\mathsf{A} \to \mathbb{P})$.

\paragraph{Implication.} For $\mathit{wp}$, the pointwise order
\[
  wp_1 \sqsubseteq wp_2 \iff \forall \mathit{post}.~wp_1(\mathit{post}) \Rightarrow wp_2(\mathit{post})
\]
is a monadic relation between any two base $\mathit{wp}$ monads, because
they share \ret and their \bind is monotone. For $\mathit{sp}$, the
pointwise order
\[
  sp_1 \sqsubseteq sp_2 \iff \forall \mathit{pre}~a.~sp_1(\mathit{pre})(a) \Rightarrow sp_2(\mathit{pre})(a)
\]
is a monadic relation without the components that it needs with state.
Each implication of the left diagram goes from $\ASP$ or $\DSP$ to $\ASLP$
or $\DSLP$, and without state every element of the former is below every
element of the latter: by \autoref{lem:asp-singleton},
$sp_1(\mathit{pre})(a)$ implies $\mathit{pre}$, and by
\autoref{lem:dslp-singleton}, $\mathit{pre}$ implies $sp_2(\mathit{pre})(a)$.
Both lemmas apply to $\MonadType{\DSP}{A}$ and $\MonadType{\ASLP}{A}$ as
well, since these coincide with $\MonadType{\ASP}{A}$ and
$\MonadType{\DSLP}{A}$, as described next. The relation is therefore total
between these monads, and so it is monadic.

\paragraph{Isomorphism.} Without state, $\MonadType{\DSP}{A}$ and
$\MonadType{\ASP}{A}$ have the same elements, \ret, and \bind, and so do
$\MonadType{\ASLP}{A}$ and $\MonadType{\DSLP}{A}$
(\autoref{sec:collapse}). The equality of $\mathit{sps}$,
\[
  sp_1 \equiv sp_2 \iff \forall \mathit{pre}~a.~sp_1(\mathit{pre})(a) \iff sp_2(\mathit{pre})(a),
\]
is therefore a monadic relation between each pair, and it relates their
elements one to one. It gives the implication in both directions: the
double-headed arrows marked $\cong$ in \autoref{fig:monad-taxonomy}.

\paragraph{Contrapositive.} For $\mathit{wp}$, the monadic relation
\[
  wp_1 \bowtie wp_2 \iff \forall \mathit{post}.~wp_1(\mathit{post}) \iff \neg\,wp_2(\lambda a.\,\neg\,\mathit{post}(a))
\]
relates $\MonadType{\AWP}{A}$ to $\MonadType{\DWLP}{A}$, and also
$\MonadType{\DWP}{A}$ to $\MonadType{\AWLP}{A}$. For $\mathit{sp}$, the
monadic relation
\[
  sp_1 \bowtie sp_2 \iff \forall \mathit{pre}~a.~sp_1(\mathit{pre})(a) \iff \neg\,sp_2(\neg\,\mathit{pre})(a)
\]
relates $\MonadType{\ASP}{A}$ to $\MonadType{\DSLP}{A}$, and also
$\MonadType{\DSP}{A}$ to $\MonadType{\ASLP}{A}$. As with state, both pair
elements one to one, and no components are needed.

\paragraph{Adjunction.} The relations
\[
\begin{aligned}
  \Ladj_{\mathsf{A}}~sp~wp &:= \forall \mathit{pre}~\mathit{post}.~(sp(\mathit{pre}) \sqsubseteq \mathit{post}) \iff (\mathit{pre} \Rightarrow wp(\mathit{post})),\\
  \Radj_{\mathsf{A}}~sp~wp &:= \forall \mathit{pre}~\mathit{post}.~(wp(\mathit{post}) \Rightarrow \mathit{pre}) \iff (\mathit{post} \sqsubseteq sp(\mathit{pre}))
\end{aligned}
\]
relate an $\mathit{sp}$ and a $\mathit{wp}$ that form a Galois connection,
with the $\mathit{sp}$ as the left adjoint in $\Ladj$ and as the right
adjoint in $\Radj$. They are monadic relations between
$\MonadType{\ASP}{A}$ and $\MonadType{\DWLP}{A}$, and between
$\MonadType{\AWP}{A}$ and $\MonadType{\DSLP}{A}$, respectively. Without
state, an element of $\MonadType{\ASP}{A}$ or $\MonadType{\DSLP}{A}$ is
determined by the outputs it produces, an element of
$\MonadType{\Pred}{A}$, so we prove both relations between
$\MonadType{\WPure}{A}$ and $\MonadType{\Pred}{A}$
(\autoref{sec:collapse}).

\section{Monad Morphisms with State and Exceptions}\label{app:exc-morphisms}

The exception monad transformer adds an exceptional outcome to a monad,
$\mathsf{ExcT}(\mathsf{M})(\mathsf{A}) := \MonadType{M}{option~A}$, where
$\mathsf{None}$ is the exception. We apply it below the state monad
transformer, on both sides:
\[
\begin{aligned}
  \StTType{ExcT(NDiv)}{S}{A} &= \mathsf{S} \to \MonadType{NDiv}{option~(A \times S)},\\
  \MonadType{\WStExc}{S, A} = \StTType{ExcT(\WPure)}{S}{A}
    &\cong (\mathsf{option}~(\mathsf{A} \times \mathsf{S}) \to \mathbb{P}) \to (\mathsf{S} \to \mathbb{P}).
\end{aligned}
\]
A computation maps a start state to a non-empty set of outcomes, each a
result $\mathsf{Some}~(a, s')$, an exception $\mathsf{None}$, which discards
the state, or divergence, and the postconditions of the specifications also
constrain the exception. Lifting the base
$\mathit{wp}$ morphisms of \autoref{sec:lifting-morphisms} through both
transformers applies them at each start state, at the type
$\mathsf{option}~(\mathsf{A} \times \mathsf{S})$:
\[
\begin{aligned}
\theta^{\mathsf{StExc}}_{\vDWP}(m)
  &:= \lambda~\mathit{post}~s.\,\forall o \in m~s.\;
      \mathsf{match}~o~\mathsf{with}~v \mapsto \mathit{post}~v \mid \mathsf{div} \mapsto \bot,\\
\theta^{\mathsf{StExc}}_{\vDWLP}(m)
  &:= \lambda~\mathit{post}~s.\,\forall o \in m~s.\;
      \mathsf{match}~o~\mathsf{with}~v \mapsto \mathit{post}~v \mid \mathsf{div} \mapsto \top,\\
\theta^{\mathsf{StExc}}_{\vAWP}(m)
  &:= \lambda~\mathit{post}~s.\,\exists o \in m~s.\;
      \mathsf{match}~o~\mathsf{with}~v \mapsto \mathit{post}~v \mid \mathsf{div} \mapsto \bot,\\
\theta^{\mathsf{StExc}}_{\vAWLP}(m)
  &:= \lambda~\mathit{post}~s.\,\exists o \in m~s.\;
      \mathsf{match}~o~\mathsf{with}~v \mapsto \mathit{post}~v \mid \mathsf{div} \mapsto \top,
\end{aligned}
\]
where $v$ ranges over $\mathsf{option}~(\mathsf{A} \times \mathsf{S})$, a
result or the exception. The exception and the state monad transformers
both lift monad morphisms, so these are monad morphisms from
$\StTType{ExcT(NDiv)}{S}{A}$ to $\MonadType{\WStExc}{S, A}$.

On the $\mathit{sp}$ side, the exception arrow transformer changes the output
of an arrow, $\mathsf{ExcArrT}(\mathsf{Arr})(\mathsf{B}, \mathsf{A}) :=
\ArrowType{Arr}{B}{option~A}$, and we apply it below the state arrow
transformer, so that an arrow of $\StTArrType{ExcArrT(Arr)}{S}{B}{A}$ is an
arrow of $\mathsf{Arr}$ from $\mathsf{B} \times \mathsf{S}$ to
$\mathsf{option}~(\mathsf{A} \times \mathsf{S})$. Lifting the base arrow
morphisms of \autoref{sec:lifting-morphisms} through both transformers and
fixing the input to $1$ gives, for a computation $m$ of
$\StTType{ExcT(NDiv)}{S}{A}$,
\[
\begin{aligned}
\theta^{\mathsf{StExc}}_{\vASP}(m)
  &:= \lambda~\mathit{pre}~v.\,\exists s.\;\exists o \in m~s.\;\mathit{pre}~s \land o = v,\\
\theta^{\mathsf{StExc}}_{\vDSLP}(m)
  &:= \lambda~\mathit{pre}~v.\,\forall s.\;\forall o \in m~s.\;o = v \Rightarrow \mathit{pre}~s,\\
\theta^{\mathsf{StExc}}_{\vDSP}(m)
  &:= \lambda~\mathit{pre}~v.\,\theta^{\mathsf{StExc}}_{\vASP}(m)~\mathit{pre}~v
      \land \theta^{\mathsf{StExc}}_{\vDSLP}(m)~\mathit{pre}~v,\\
\theta^{\mathsf{StExc}}_{\vASLP}(m)
  &:= \lambda~\mathit{pre}~v.\,\theta^{\mathsf{StExc}}_{\vASP}(m)~\mathit{pre}~v
      \lor \theta^{\mathsf{StExc}}_{\vDSLP}(m)~\mathit{pre}~v,
\end{aligned}
\]
where $v$ ranges over $\mathsf{option}~(\mathsf{A} \times \mathsf{S})$, so
postconditions also constrain the exception, and divergence is never an
output. These are the definitions of
\cref{eqn:monad-morphisms-stateful-sp} with a result or the exception in
place of a result, and they are monad morphisms because both arrow
transformers lift arrow morphisms.

\section{No Embedding for the Over-Approximate $\mathit{dslp}$ Triple}\label{app:dslp-no-embedding}

We prove the second part of \autoref{lem:no-embedding}. The over-approximate
$\DSLPSt$ triple,
$\theta^{\mathsf{St}}_{\vDSLP}(m)(\mathit{pre}) \sqsubseteq \mathit{post}$,
has no embedding either: there is no element $w$ of $\DSLPSt$ such that, for
every $m$, $\theta^{\mathsf{St}}_{\vDSLP}(m) \sqsubseteq w$ holds exactly when
the triple does. The argument is dual to the one for the incorrectness
triple in \autoref{sec:dms}, with conjunctivity in place of disjunctivity.
An element $w$ of $\DSLPSt$ satisfies
\[
  w(p) = \bigwedge_{s \notin p} w(\mathsf{S} \setminus \{s\}),
\]
where the empty meet, at $p = \mathsf{S}$, is $\top$ by termination. So $w$
commits to the start state from which each result it excludes is reached,
whereas the triple only asks that each result outside $\mathit{post}$ be
reachable from {\em some} state outside $\mathit{pre}$.

For example, let $\mathit{pre} = \mathsf{S} \setminus \{s_0, s_1\}$ and let
$\mathit{post}$ contain every result except $r_0$. A result whose start
states all lie in $\mathit{pre}$ must then differ from $r_0$, so the triple
holds exactly when $r_0$ is reachable from $s_0$ or $s_1$. A program $m_1$
whose only result is $r_0$, reached only from $s_0$, and a program $m_2$
whose only result is $r_0$, reached only from $s_1$, both satisfy the
triple, so such a $w$ would satisfy
$\theta^{\mathsf{St}}_{\vDSLP}(m_1) \sqsubseteq w$ and
$\theta^{\mathsf{St}}_{\vDSLP}(m_2) \sqsubseteq w$. At a precondition
$\mathsf{S} \setminus \{s\}$, one of the two lower bounds is always $\top$:
at $r_0$ they are $s_0 \neq s$ and $s_1 \neq s$, and at every other result
both hold vacuously. So $w(\mathsf{S} \setminus \{s\}) = \top$, and hence
$w = \top$ by conjunctivity. But
$\theta^{\mathsf{St}}_{\vDSLP}(m) \sqsubseteq \top$ holds for every $m$, so
the triple would have to hold for every program. It does not: a program that
always diverges reaches no result, so $r_0$ is not reachable from $s_0$ or
$s_1$.

This is also an instance of the contrapositive relation of
\autoref{sec:contrapositive-relation}. The dual
$\overline{w}(p) := \neg\,w(\neg p)$ of an element of $\DSLPSt$ is an
element of $\ASPSt$, duality reverses the order and maps
$\theta^{\mathsf{St}}_{\vDSLP}(m)$ to $\theta^{\mathsf{St}}_{\vASP}(m)$, and
the over-approximate $\DSLPSt$ triple for $(\mathit{pre}, \mathit{post})$
holds exactly when the incorrectness triple for
$(\neg\mathit{pre}, \neg\mathit{post})$ does. So the dual of an embedding for
the first would be an embedding for the second; with
$\neg\mathit{pre} = \{s_0, s_1\}$ and $\neg\mathit{post} = \{r_0\}$, this is
the example of \autoref{sec:dms}.

\bibliographystyle{abbrvnaturl} %
\bibliography{fstar}

\begin{thebibliography}{40}
\providecommand{\natexlab}[1]{#1}

\bibitem[Aguirre et~al.(2022)Aguirre, Katsumata, and Kura]{AguirreKK22}
A.~Aguirre, S.~Katsumata, and S.~Kura.
\newblock \href {http://dx.doi.org/10.1017/S0960129522000330} {Weakest
  preconditions in fibrations}.
\newblock \emph{Math. Struct. Comput. Sci.}, 32\penalty0 (4):\penalty0
  472--510, 2022.

\bibitem[Ahman et~al.(2017)Ahman, Hri\c{t}cu, Maillard, Mart\'inez, Plotkin,
  Protzenko, Rastogi, and Swamy]{dm4free}
D.~Ahman, C.~Hri\c{t}cu, K.~Maillard, G.~Mart\'inez, G.~Plotkin, J.~Protzenko,
  A.~Rastogi, and N.~Swamy.
\newblock \href {http://dx.doi.org/10.1145/3009837.3009878} {Dijkstra monads
  for free}.
\newblock \iflongrefs{In \emph{44th ACM SIGPLAN Symposium on Principles of
  Programming Languages (POPL)}}\else{\emph{POPL}}\fi{}. \iflongrefs{Jan.
  }\fi{}2017.

\bibitem[Apt(1984)]{Apt84}
K.~R. Apt.
\newblock \href {http://dx.doi.org/10.1016/0304-3975(83)90066-X} {Ten years of
  {Hoare}'s logic: A survey -- part {II}: Nondeterminism}.
\newblock \emph{Theor. Comput. Sci.}, 28\penalty0 (1--2):\penalty0 83--109,
  1984.

\bibitem[Asada(2010)]{arrows-strong-monad}
K.~Asada.
\newblock \href {http://dx.doi.org/10.1145/1863597.1863607} {Arrows are strong
  monads}.
\newblock In V.~Capretta and J.~Chapman, editors, \emph{Proceedings of the 3rd
  {ACM} {SIGPLAN} Workshop on Mathematically Structured Functional Programming,
  MSFP@ICFP 2010, Baltimore, MD, USA, September 25, 2010}. 2010.

\bibitem[Ascari et~al.(2025)Ascari, Bruni, Gori, and Logozzo]{sil}
F.~Ascari, R.~Bruni, R.~Gori, and F.~Logozzo.
\newblock \href {http://dx.doi.org/10.1145/3720486} {Revealing sources of
  (memory) errors via backward analysis}.
\newblock \emph{Proc. {ACM} Program. Lang.}, 9\penalty0 ({OOPSLA1}):\penalty0
  1321--1348, 2025.

\bibitem[Ascari et~al.(2026)Ascari, Bruni, Gori, and Raad]{u-turn}
F.~Ascari, R.~Bruni, R.~Gori, and A.~Raad.
\newblock \href {http://dx.doi.org/10.1145/3776688} {U-turn: Enhancing
  incorrectness analysis by reversing direction}.
\newblock \emph{Proc. {ACM} Program. Lang.}, 10\penalty0 ({POPL}):\penalty0
  1326--1352, 2026.

\bibitem[Atkey(2011)]{Atkey2011}
R.~Atkey.
\newblock \href {http://dx.doi.org/10.1016/j.entcs.2011.02.014} {What is a
  categorical model of arrows?}
\newblock \emph{Electronic Notes in Theoretical Computer Science}, 229\penalty0
  (5):\penalty0 19--37, 2011.
\newblock Proceedings of the Second Workshop on Mathematically Structured
  Functional Programming (MSFP 2008).

\bibitem[Back and von Wright(1998)]{BackW98}
R.-J. Back and J.~von Wright.
\newblock \href {http://dx.doi.org/10.1007/978-1-4612-1674-2} {\emph{Refinement
  Calculus: A Systematic Introduction}}.
\newblock Graduate Texts in Computer Science. Springer, 1998.

\bibitem[de~Vries and Koutavas(2011)]{reverse-hoare}
E.~de~Vries and V.~Koutavas.
\newblock \href {http://dx.doi.org/10.1007/978-3-642-24690-6\_12} {Reverse
  {Hoare} logic}.
\newblock In G.~Barthe, A.~Pardo, and G.~Schneider, editors, \emph{Software
  Engineering and Formal Methods - 9th International Conference, {SEFM} 2011,
  Montevideo, Uruguay, November 14-18, 2011. Proceedings}. 2011.

\bibitem[Dijkstra(1975)]{Dijkstra75}
E.~W. Dijkstra.
\newblock \href {http://dx.doi.org/10.1145/360933.360975} {Guarded commands,
  nondeterminacy and formal derivation of programs}.
\newblock \iflongrefs{\emph{Commun. ACM}}\else{\emph{CACM}}\fi{}, 18\penalty0
  (8):\penalty0 453--457, \iflongrefs{Aug. }\fi{}1975.

\bibitem[Dijkstra(1976)]{Dijkstra97}
E.~W. Dijkstra.
\newblock \href {https://dl.acm.org/citation.cfm?id=550359} {\emph{A Discipline
  of Programming}}.
\newblock Prentice Hall PTR, Upper Saddle River, NJ, USA, 1st edition, 1976.

\bibitem[Dijkstra and Scholten(1990)]{DijkstraS90}
E.~W. Dijkstra and C.~S. Scholten.
\newblock \href {http://dx.doi.org/10.1007/978-1-4612-3228-5} {\emph{Predicate
  Calculus and Program Semantics}}.
\newblock Texts and Monographs in Computer Science. Springer-Verlag, 1990.

\bibitem[Ebner et~al.(2025)Ebner, Mart{\'{\i}}nez, Rastogi, Dardinier,
  Frisella, Ramananandro, and Swamy]{pulsecore}
G.~Ebner, G.~Mart{\'{\i}}nez, A.~Rastogi, T.~Dardinier, M.~Frisella,
  T.~Ramananandro, and N.~Swamy.
\newblock \href {http://dx.doi.org/10.1145/3729311} {{PulseCore}: An
  impredicative concurrent separation logic for dependently typed programs}.
\newblock \emph{Proc. {ACM} Program. Lang.}, 9\penalty0 ({PLDI}), 2025.

\bibitem[Gladshtein et~al.(2026)Gladshtein, P{\^\i}rlea, Zhao, Kurin, and
  Sergey]{loom}
V.~Gladshtein, G.~P{\^\i}rlea, Q.~Zhao, V.~Kurin, and I.~Sergey.
\newblock \href {http://dx.doi.org/10.1145/3776719} {Foundational multi-modal
  program verifiers}.
\newblock \emph{Proc. {ACM} Program. Lang.}, 10\penalty0 ({POPL}):\penalty0
  2233--2264, 2026.

\bibitem[Hasuo(2015)]{HasuoGWP15}
I.~Hasuo.
\newblock \href {http://dx.doi.org/10.1016/j.tcs.2015.03.047} {Generic weakest
  precondition semantics from monads enriched with order}.
\newblock \emph{Theor. Comput. Sci.}, 604:\penalty0 2--29, 2015.

\bibitem[Hoare(1969)]{Hoare69}
C.~A.~R. Hoare.
\newblock \href {http://dx.doi.org/10.1145/363235.363259} {An axiomatic basis
  for computer programming}.
\newblock \emph{Commun. {ACM}}, 12\penalty0 (10):\penalty0 576--580, 1969.

\bibitem[Hughes(2000)]{Hughes00}
J.~Hughes.
\newblock \href {http://dx.doi.org/10.1016/S0167-6423(99)00023-4} {Generalising
  monads to arrows}.
\newblock \emph{Sci. Comput. Program.}, 37\penalty0 (1--3):\penalty0 67--111,
  2000.

\bibitem[Jacobs(2015)]{Jacobs15}
B.~Jacobs.
\newblock \href {http://dx.doi.org/10.1016/j.tcs.2015.03.020} {{Dijkstra} and
  {Hoare} monads in monadic computation}.
\newblock \emph{Theor. Comput. Sci.}, 604:\penalty0 30--45, 2015.

\bibitem[Jensen(1978)]{Jensen78}
K.~Jensen.
\newblock \href
  {http://ojs.statsbiblioteket.dk/index.php/daimipb/article/download/6502/5621}
  {Connection between {Dijkstra's} predicate-transformers and denotational
  continuation-semantics}.
\newblock DAIMI Report Series 7.86, 1978.

\bibitem[Katsumata and Sato(2013)]{KatsumataS13}
S.~Katsumata and T.~Sato.
\newblock \href {http://dx.doi.org/10.1007/978-3-642-37075-5\_10} {Preorders on
  monads and coalgebraic simulations}.
\newblock \iflongrefs{In \emph{Foundations of Software Science and Computation
  Structures - 16th International Conference, {FOSSACS} 2013, Held as Part of
  the European Joint Conferences on Theory and Practice of Software, {ETAPS}
  2013, Rome, Italy, March 16-24, 2013.
  Proceedings}}\else{\emph{FOSSACS}}\fi{}, 2013.

\bibitem[Lack and Street(2012)]{LackStreet12}
S.~Lack and R.~Street.
\newblock \href {http://www.tac.mta.ca/tac/volumes/26/15/26-15abs.html} {Skew
  monoidales, skew warpings and quantum categories}.
\newblock \emph{Theory and Applications of Categories}, 26\penalty0
  (15):\penalty0 385--402, 2012.

\bibitem[Leino(2005)]{Leino05}
K.~R.~M. Leino.
\newblock \href {http://dx.doi.org/10.1016/j.ipl.2004.10.015} {Efficient
  weakest preconditions}.
\newblock \emph{Inf. Process. Lett.}, 93\penalty0 (6):\penalty0 281--288, 2005.

\bibitem[Li et~al.(2024)Li, Zilberstein, and Silva]{LiZS24}
J.~Li, N.~Zilberstein, and A.~Silva.
\newblock \href {http://dx.doi.org/10.48550/ARXIV.2411.00197} {Total outcome
  logic: Unified reasoning for a taxonomy of program logics}.
\newblock \emph{CoRR}, abs/2411.00197, 2024.

\bibitem[Lindley et~al.(2011)Lindley, Wadler, and Yallop]{LindleyWY11}
S.~Lindley, P.~Wadler, and J.~Yallop.
\newblock \href {http://dx.doi.org/10.1016/j.entcs.2011.02.018} {Idioms are
  oblivious, arrows are meticulous, monads are promiscuous}.
\newblock \emph{Electron. Notes Theor. Comput. Sci.}, 229\penalty0
  (5):\penalty0 97--117, 2011.

\bibitem[Maillard(2019)]{Maillard19}
K.~Maillard.
\newblock \href {https://theses.hal.science/tel-03441023} {\emph{Principles of
  Program Verification for Arbitrary Monadic Effects}}.
\newblock PhD thesis, Universit\'e Paris Sciences et Lettres, 2019.

\bibitem[Maillard et~al.(2019)Maillard, Ahman, Atkey, Mart\'{\i}nez,
  Hri\c{t}cu, Rivas, and Tanter]{dm4all}
K.~Maillard, D.~Ahman, R.~Atkey, G.~Mart\'{\i}nez, C.~Hri\c{t}cu, E.~Rivas, and
  E.~Tanter.
\newblock \href {http://dx.doi.org/10.1145/3341708} {Dijkstra monads for all}.
\newblock \iflongrefs{\emph{Proc. ACM Program.
  Lang.}}\else{\emph{PACMPL}}\fi{}, 3\penalty0 (ICFP), \iflongrefs{July
  }\fi{}2019.

\bibitem[M{\"{o}}ller et~al.(2021)M{\"{o}}ller, O'Hearn, and Hoare]{Moller21}
B.~M{\"{o}}ller, P.~O'Hearn, and T.~Hoare.
\newblock \href {http://dx.doi.org/10.1007/978-3-030-88701-8_20} {On algebra of
  program correctness and incorrectness}.
\newblock In \emph{Relational and Algebraic Methods in Computer Science - 19th
  International Conference, RAMiCS 2021, Marseille, France, November 2-5, 2021,
  Proceedings}. 2021.

\bibitem[Nanevski et~al.(2008{\natexlab{a}})Nanevski, Morrisett, Shinnar,
  Govereau, and Birkedal]{ynot-icfp08}
A.~Nanevski, G.~Morrisett, A.~Shinnar, P.~Govereau, and L.~Birkedal.
\newblock \href {http://dx.doi.org/10.1145/1411204.1411237} {{Ynot}: dependent
  types for imperative programs}.
\newblock \iflongrefs{In J.~Hook and P.~Thiemann, editors, \emph{Proceeding of
  the 13th {ACM} {SIGPLAN} international conference on Functional programming,
  {ICFP} 2008, Victoria, BC, Canada, September 20-28,
  2008}}\else{\emph{ICFP}}\fi{}. 2008{\natexlab{a}}.

\bibitem[Nanevski et~al.(2008{\natexlab{b}})Nanevski, Morrisett, and
  Birkedal]{nmb08htt}
A.~Nanevski, J.~G. Morrisett, and L.~Birkedal.
\newblock \href {http://ynot.cs.harvard.edu/papers/jfpsep07.pdf} {Hoare type
  theory, polymorphism and separation}.
\newblock \iflongrefs{\emph{J. Funct. Program.}}\else{\emph{JFP}}\fi{},
  18\penalty0 (5-6):\penalty0 865--911, 2008{\natexlab{b}}.

\bibitem[O'Hearn(2020)]{OHearn20}
P.~W. O'Hearn.
\newblock \href {http://dx.doi.org/10.1145/3371078} {Incorrectness logic}.
\newblock \emph{Proc. {ACM} Program. Lang.}, 4\penalty0 ({POPL}):\penalty0
  10:1--10:32, 2020.

\bibitem[Paterson(2003)]{paterson2003arrows}
R.~Paterson.
\newblock \href {https://www.staff.city.ac.uk/~ross/papers/fop.html} {Arrows
  and computation}.
\newblock In J.~Gibbons and O.~de~Moor, editors, \emph{The Fun of Programming},
  pages 201--222. Palgrave, 2003.

\bibitem[Silver and Zdancewic(2021)]{SilverZ21}
L.~Silver and S.~Zdancewic.
\newblock \href {http://dx.doi.org/10.1145/3434307} {Dijkstra monads forever:
  termination-sensitive specifications for interaction trees}.
\newblock \emph{Proc. {ACM} Program. Lang.}, 5\penalty0 ({POPL}):\penalty0
  1--28, 2021.

\bibitem[Swamy et~al.(2013)Swamy, Weinberger, Schlesinger, Chen, and
  Livshits]{fstar-pldi13}
N.~Swamy, J.~Weinberger, C.~Schlesinger, J.~Chen, and B.~Livshits.
\newblock \href
  {https://www.microsoft.com/en-us/research/publication/verifying-higher-order-programs-with-the-dijkstra-monad/}
  {Verifying higher-order programs with the {Dijkstra} monad}.
\newblock \iflongrefs{In \emph{Proceedings of the 34th annual ACM SIGPLAN
  conference on Programming Language Design and
  Implementation}}\else{\emph{PLDI}}\fi{}, 2013.

\bibitem[Swamy et~al.(2016)Swamy, Hri\c{t}cu, Keller, Rastogi, Delignat-Lavaud,
  Forest, Bhargavan, Fournet, Strub, Kohlweiss, Zinzindohou\'e, and
  {Zanella-B\'eguelin}]{mumon}
N.~Swamy, C.~Hri\c{t}cu, C.~Keller, A.~Rastogi, A.~Delignat-Lavaud, S.~Forest,
  K.~Bhargavan, C.~Fournet, P.-Y. Strub, M.~Kohlweiss, J.-K. Zinzindohou\'e,
  and S.~{Zanella-B\'eguelin}.
\newblock \href {https://www.fstar-lang.org/papers/mumon/} {Dependent types and
  multi-monadic effects in {F*}}.
\newblock \iflongrefs{In \emph{43rd ACM SIGPLAN-SIGACT Symposium on Principles
  of Programming Languages (POPL)}}\else{\emph{POPL}}\fi{}. \iflongrefs{Jan.
  }\fi{}2016.

\bibitem[Swierstra and Baanen(2019)]{SwierstraB19}
W.~Swierstra and T.~Baanen.
\newblock \href {http://dx.doi.org/10.1145/3341707} {A predicate transformer
  semantics for effects (functional pearl)}.
\newblock \iflongrefs{\emph{Proc. {ACM} Program.
  Lang.}}\else{\emph{PACMPL}}\fi{}, 3\penalty0 ({ICFP}):\penalty0
  103:1--103:26, 2019.

\bibitem[Verbeek et~al.(2025)Verbeek, Sefat, Fu, and Ravindran]{uil}
F.~Verbeek, M.~S. Sefat, Z.~Fu, and B.~Ravindran.
\newblock \href {http://dx.doi.org/10.1145/3704850} {On extending incorrectness
  logic with backwards reasoning}.
\newblock \emph{Proc. {ACM} Program. Lang.}, 9\penalty0 ({POPL}):\penalty0
  391--415, 2025.

\bibitem[Verscht and Kaminski(2025)]{VerschtK25}
L.~Verscht and B.~L. Kaminski.
\newblock \href {http://dx.doi.org/10.1145/3704896} {A taxonomy of {Hoare}-like
  logics: Towards a holistic view using predicate transformers and {Kleene}
  algebras with top and tests}.
\newblock \emph{Proc. {ACM} Program. Lang.}, 9\penalty0 ({POPL}):\penalty0
  1782--1811, 2025.

\bibitem[Xia et~al.(2020)Xia, Zakowski, He, Hur, Malecha, Pierce, and
  Zdancewic]{XiaZHHMPZ20}
L.~Xia, Y.~Zakowski, P.~He, C.~Hur, G.~Malecha, B.~C. Pierce, and S.~Zdancewic.
\newblock \href {http://dx.doi.org/10.1145/3371119} {Interaction trees:
  representing recursive and impure programs in {Coq}}.
\newblock \emph{Proc. {ACM} Program. Lang.}, 4\penalty0 ({POPL}):\penalty0
  51:1--51:32, 2020.

\bibitem[Zhang and Kaminski(2022)]{ZhangK22}
L.~Zhang and B.~L. Kaminski.
\newblock \href {http://dx.doi.org/10.1145/3527331} {Quantitative strongest
  post: a calculus for reasoning about the flow of quantitative information}.
\newblock \emph{Proc. {ACM} Program. Lang.}, 6\penalty0 ({OOPSLA1}):\penalty0
  1--29, 2022.

\bibitem[Zilberstein et~al.(2023)Zilberstein, Dreyer, and
  Silva]{ZilbersteinDS23}
N.~Zilberstein, D.~Dreyer, and A.~Silva.
\newblock \href {http://dx.doi.org/10.1145/3586045} {Outcome logic: {A}
  unifying foundation for correctness and incorrectness reasoning}.
\newblock \emph{Proc. {ACM} Program. Lang.}, 7\penalty0 ({OOPSLA1}):\penalty0
  522--550, 2023.

\end{thebibliography}

\end{document}